\documentclass[11pt]{article}
\pdfoutput=1 %for arxiv
\usepackage[sc]{mathpazo}

\usepackage[margin=1in]{geometry}
\usepackage[english]{babel}
\usepackage[utf8x]{inputenc}
\usepackage[compact]{titlesec}

\usepackage{cmap}
\usepackage[T1]{fontenc}
\usepackage{bm}
\usepackage{amsfonts}
\usepackage{amsmath}
\usepackage{amssymb}
\usepackage{amsthm}
\usepackage{float}
\usepackage{graphics}

\usepackage{hyperref}
\usepackage[svgnames]{xcolor}
\hypersetup{colorlinks={true},urlcolor={blue},linkcolor={DarkBlue},citecolor=[named]{DarkGreen},linktoc=all}
\usepackage{natbib}

\usepackage{microtype}
\usepackage[capitalise,nameinlink,noabbrev]{cleveref}

\usepackage{doi}

\usepackage{booktabs} % For formal tables
\usepackage{amsfonts,amsbsy,amsthm}%amsthm,amsmath,amssymb,
\usepackage{thmtools,thm-restate}
\usepackage{enumerate}
\usepackage{bbm}
\usepackage{nicefrac}
\usepackage{wrapfig}
\usepackage[caption=false]{subfig}
\usepackage{mathtools}
\usepackage{array,multirow}

\newcommand{\remove}[1]{}

\newenvironment{myquote}[1]%
  {\list{}{\leftmargin=#1\rightmargin=#1}\item[]}%
  {\endlist}

\newtheorem{examplex}{Example}
\newenvironment{example}{\pushQED{\qed}\examplex}{\popQED\endexamplex}%Requires amsthm
\newtheorem{claim}{Claim}
\newtheorem{corollary}{Corollary}
\newtheorem{definition}{Definition}
\newtheorem{lemma}{Lemma}
\newtheorem{proposition}{Proposition}
\newtheorem{remark}{Remark}
\newtheorem{theorem}{Theorem}

\usepackage[capitalise,noabbrev]{cleveref}
\Crefname{claim}{Claim}{Claims}
\Crefname{corollary}{Corollary}{Corollaries}
\Crefname{definition}{Definition}{Definitions}
\Crefname{example}{Example}{Examples}
\Crefname{lemma}{Lemma}{Lemmas}
\Crefname{property}{Property}{Properties}
\Crefname{proposition}{Proposition}{Propositions}
\Crefname{remark}{Remark}{Remarks}
\Crefname{theorem}{Theorem}{Theorems}

\usepackage[linesnumbered,lined,boxed,ruled,vlined]{algorithm2e} % For algorithms

\SetKwInOut{Parameter}{Parameter}
\SetAlFnt{\small}
\SetAlCapFnt{\small}
\SetAlCapNameFnt{\small}
\SetAlCapHSkip{0pt}
\IncMargin{-\parindent}

\usepackage{mdframed}

\input{insbox}%%%%%%%%%%%%%% TeX macro,
\usepackage{tikz}
\usetikzlibrary{shapes,shapes.multipart,decorations.text,arrows,decorations.markings,decorations.pathmorphing,shapes.geometric,positioning,decorations.pathreplacing}
\tikzset{snake it/.style={decorate, decoration=snake}}

\usepackage{xparse,zref-savepos}
\makeatletter
\@ifundefined{zsaveposx}{\let\zsaveposx\zsavepos}{}
\newcounter{hposcnt}
\renewcommand*{\thehposcnt}{hpos\number\value{hposcnt}}
\NewDocumentCommand{\lplabel}{m m}{%
  \stepcounter{hposcnt}%
  \zsaveposx{\thehposcnt l}%
  \zref@refused{\thehposcnt l}%
  \zref@refused{hpos0l}%
  \makebox[0pt][r]{\makebox[\dimexpr\zposx{\thehposcnt l}sp-\zposx{hpos0l}sp][l]{#2}}%
  \IfNoValueF{#1}
   {\def\@currentlabel{#2}\ltx@label{#1}}
}
\makeatother
\AtBeginDocument{\zsaveposx{hpos0l}}

\allowdisplaybreaks
\renewcommand{\>}{{\succ}}

\newcommand{\eps}{{\varepsilon}}

\newcommand{\I}{\mathcal{I}}

\renewcommand{\O}{{\mathcal{O}}}
\newcommand{\bigO}{\mathcal{O}}
\newcommand{\OptStab}{\textup{\textsc{Optimal Stable Fractional Matching}}}
\newcommand{\OptEpsStab}{\textup{\textsc{Optimal $\eps$-Stable Fractional Matching}}}

\newcommand{\R}{\mathbb{R}}

\newcommand{\SMC}{\textup{\textrm{SMC}}}

\newcommand{\W}{\mathcal{W}}

\title{Stable Fractional Matchings}

\author{
	\begin{tabular}{cc}
		& \\
		\textbf{Ioannis Caragiannis} & \textbf{Aris Filos-Ratsikas}\\
		\small{Aarhus University, Denmark} & \small{University of Liverpool, United Kingdom}\\
		\href{mailto:iannis@cs.au.dk}{\small{\texttt{iannis@cs.au.dk}}} & \href{mailto:Aris.Filos-Ratsikas@liverpool.ac.uk}{\small{\texttt{Aris.Filos-Ratsikas@liverpool.ac.uk}}}\\
		& \\
		\textbf{Panagiotis Kanellopoulos} & \textbf{Rohit Vaish}\\
		\small{University of Essex, United Kingdom} & \small{Tata Institute of Fundamental Research, India}\\
		\href{mailto:panagiotis.kanellopoulos@essex.ac.uk}{\small{\texttt{panagiotis.kanellopoulos@essex.ac.uk}}} & \href{mailto:rohit.vaish@tifr.res.in}{\small{\texttt{rohit.vaish@tifr.res.in}}}\\
		& \\
	\end{tabular}
}

\date{}

\begin{document}

\maketitle

\begin{abstract}
We study a generalization of the classical stable matching problem that allows for \emph{cardinal} preferences (as opposed to ordinal) and \emph{fractional} matchings (as opposed to integral). In this cardinal setting, stable fractional matchings can have much larger social welfare than stable integral ones. Our goal is to understand the computational complexity of finding an \emph{optimal} (i.e., welfare-maximizing) stable fractional matching. We consider both \emph{exact} and \emph{approximate} stability notions, and provide simple approximation algorithms with weak welfare guarantees. Our main result is that, somewhat surprisingly, achieving better approximations is computationally hard. To the best of our knowledge, these are the first computational complexity results for stable fractional matchings in the cardinal model. En route to these results, we provide a number of structural observations that could be of independent interest.
\end{abstract}

\section{Introduction}
\label{sec:Introduction}
The stable matching problem is one of the most extensively studied problems at the interface of economics and computer science~\citep{GS62stable,GI89stable,RS90book,K97stable,M13algorithmics}. The input to the problem consists of the preference lists of two sets of agents, commonly referred to as the \emph{men} and the \emph{women}. The goal is to find a \emph{stable} matching, i.e., a matching in which no pair of man and woman prefer each other over their assigned partners.

While the problem was originally motivated by college admissions~\citep{GS62stable}, its applicability has subsequently expanded to various other domains such as medical residency~\citep{R84evolution,RP99redesign}
and school choice~\citep{APR05new}. In addition, the insights gained from the study of the stable matching problem, together with the development of computational tools and techniques in artificial intelligence, have shaped the design of modern two-sided matching platforms such as organ exchanges~\citep{RSU04kidney} and ridesharing platforms~\citep{WAE18stable}.

The standard formulation of the stable matching problem involves two important assumptions, namely that the matching is \emph{integral} (i.e., two agents are either completely matched or completely unmatched) and that the agents have \emph{ordinal} preferences (typically in the form of rank-ordered lists). Although these assumptions suffice in a number of applications, including those mentioned above, there are natural examples where they could be inadequate. For instance, consider a \emph{time-sharing} scenario~\citep{RRV93stable} wherein a set of employees are matched with a set of supervisors. Assuming that each individual can spend one unit of time at work, an integral matching prescribes that every employee should work full time with a single supervisor. On the other hand, fractional matchings allow the employees to divide their time in working with multiple supervisors, making them a more natural modeling choice in such situations. Fractional matchings are also useful in the context of \emph{randomization}, as they can be used to model lotteries over integral matchings~\citep{BM04random,BAH+18redesigning,G19compromises}.

In a similar vein, \emph{ordinal} preferences, despite their simplicity and ease of elicitation, can often be quite restrictive. Indeed, in many real-world matching applications, the outcomes experienced by the participants are inherently \emph{cardinal} in nature (e.g., wages in labor markets or quality of transplants in kidney exchange~\citep{LLM+19incorporating}). In such settings, it is decidedly more natural to model the \emph{intensity} of preferences, as has been noted both in theory \citep{AD10matching,PRV+13stability} as well as in lab experiments \citep{EWY16clearinghouses}.

Motivated by these applications, we consider a generalization of the stable matching model that allows for \emph{fractional} matchings (as opposed to integral) and \emph{cardinal} preferences (as opposed to ordinal). More concretely, we consider a setting in which the preferences are specified in terms of numerical utilities or \emph{valuations} (for example, in the matching instance in \Cref{fig:CardinalvsOrdinal}, $m_1$ values $w_1$ at $0$, and $w_1$ values $m_1$ at $3$). A fractional matching is simply a convex combination of integral matchings, and an agent's utility under a fractional matching is the appropriately weighted sum of its utilities under the constituent integral matchings. A fractional matching $\mu$ is \emph{stable} if no pair of man and woman simultaneously derive greater utility in being integrally matched to each other than they do under $\mu$~\citep{ADN13anarchy,PRV+13stability}. Thus, for instance, in the employee-supervisor example mentioned above, stability ensures that no employee-supervisor pair will abandon their time-sharing arrangements and instead prefer to work with each other full time.

The aforementioned generalization has clear merit in terms of \emph{social welfare}: Stable solutions in the generalized model can have larger welfare than those in the standard model, as the following example illustrates.

\begin{figure}
\footnotesize
\centering
\begin{tikzpicture}
 \tikzset{man/.style = {shape=circle,draw,inner sep=1pt}}
	  \tikzset{woman/.style = {shape=circle,draw,inner sep=1pt}}
	  \tikzset{edge/.style = {solid}}
\node[man]   (1) at (0,3.2) {$m_1$};
\node[man]   (2) at (0,1.6) {$m_2$};
\node[man]   (3) at (0,0)   {$m_3$};
\node[woman] (4) at (2.6,3.2) {$w_1$};
\node[woman] (5) at (2.6,1.6) {$w_2$};
\node[woman] (6) at (2.6,0)   {$w_3$};
\draw[edge] (1) to node [very near start,fill=white,inner sep=0pt] (141) {$0$} node [very near end,fill=white,inner sep=0pt] (142) {$3$} (4);
\draw[edge] (1) to node [very near start,fill=white,inner sep=0pt] (151) {$1$} node [very near end,fill=white,inner sep=0pt] (152) {$0$} (5);
\draw[edge] (1) to node [very near start,fill=white,inner sep=0pt] (161) {$2$} node [very near end,fill=white,inner sep=0pt] (162) {$1$} (6);
\draw[edge] (2) to node [very near start,fill=white,inner sep=0pt] (241) {$2$} node [very near end,fill=white,inner sep=0pt] (242) {$0$} (4);
\draw[edge] (2) to node [very near start,fill=white,inner sep=0pt] (251) {$1$} node [very near end,fill=white,inner sep=0pt] (252) {$1$} (5);
\draw[edge] (2) to node [very near start,fill=white,inner sep=0pt] (261) {$0$} node [very near end,fill=white,inner sep=0pt] (262) {$2$} (6);
\draw[edge] (3) to node [very near start,fill=white,inner sep=0pt] (341) {$1$} node [very near end,fill=white,inner sep=0pt] (342) {$1$} (4);
\draw[edge] (3) to node [very near start,fill=white,inner sep=0pt] (351) {$0$} node [very near end,fill=white,inner sep=0pt] (352) {$2$} (5);
\draw[edge] (3) to node [very near start,fill=white,inner sep=0pt] (361) {$3$} node [very near end,fill=white,inner sep=0pt] (362) {$0$} (6);
\end{tikzpicture}
\caption{An instance with cardinal preferences.\label{fig:CardinalvsOrdinal}}
\end{figure}
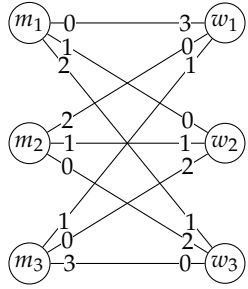

%\paragraph{Motivating example:}
\begin{example}
\label{eg:Motivating_Example}
Consider the instance in \Cref{fig:CardinalvsOrdinal} with three men $m_1$, $m_2$, $m_3$ and three women $w_1$, $w_2$, $w_3$. Among the six possible integral matchings, only two are stable, namely $\mu_1 \coloneqq \{(m_1,w_3),(m_2,w_2),(m_3,w_1)\}$ and $\mu_2 \coloneqq \{(m_1,w_1),(m_2,w_3),(m_3,w_2)\}$; indeed, $\mu_1$ and $\mu_2$ are the men-proposing and women-proposing Gale-Shapley matchings, respectively~\citep{GS62stable}. The social welfare (i.e., the sum of utilities of all agents) of these matchings is $\W(\mu_1) = \W(\mu_2) = 7$.

Define $\mu_3 \coloneqq \{(m_1,w_1),(m_2,w_2),(m_3,w_3)\}$, and notice that $\W(\mu_3) = 8$. Now consider a \emph{fractional} matching $\mu \coloneqq \frac{1}{2} \mu_2 + \frac{1}{2} \mu_3$. The social welfare of $\mu$ is $\W(\mu) = \frac{1}{2} \W(\mu_2) + \frac{1}{2} \W(\mu_3) = 15/2 > 7 = \W(\mu_1) = \W(\mu_2)$, which means that $\mu$ has a higher social welfare than any stable integral matching. Importantly, $\mu$ is a \emph{stable} fractional matching. Indeed, in $\mu$, the utilities of $m_1$, $m_2$, $m_3$, $w_1$, $w_2$, and $w_3$ are $0$, $1/2$, $3/2$, $3$, $3/2$, and $1$ respectively. Thus, for every man-woman pair, at least one of the two agents meets its utility threshold for that pair, implying that $\mu$ is stable.

Overall, the instance in \Cref{fig:CardinalvsOrdinal} admits a stable fractional matching with strictly greater welfare than any stable integral matching.\footnote{Observe that the gain in social welfare of $\mu$ was achieved by including an \emph{unstable} integral matching $\mu_3$ in its support. One could consider alternative notions such as \emph{ex-post stability}~\citep{RRV93stable} wherein the support consists only of stable integral matchings. We discuss this and various other stability notions in \Cref{sec:Stability_Notions}, and find that insisting on a purely stable support results in significant welfare loss (\Cref{rem:Strong_Stability_Suboptimal}). In fact, there exists instances where the support of an optimal stable fractional matchings consists only of unstable integral matchings (\Cref{prop:UnstableSupport}).}
\end{example}

Starting with the seminal work of \cite{GS62stable}, an extensive literature has emerged over the years on algorithms for computing stable solutions, including ones that optimize a variety of objectives pertaining to fairness and economic efficiency \citep{G87three,ILG87efficient,K97stable,TS98geometry,MII+02hard}. Many of these algorithms, however, are tailored to compute stable \emph{integral} matchings. As \Cref{eg:Motivating_Example} demonstrates, such algorithms could, in general, return highly suboptimal outcomes in our setting. Therefore, it becomes pertinent to understand the computational complexity of finding an ``optimal'' stable matching in the generalized model. Our work studies this question from the lens of the fundamental objective of social welfare, and asks the following natural question:

\begin{myquote}{0.5in}
\emph{Can an (approximately) optimal stable fractional matching be efficiently computed?}
\end{myquote}

\subsection{Our results and roadmap}
We formalize the above question by defining the optimization problem \OptStab. To motivate this problem, we strengthen the observation in Example~\ref{eg:Motivating_Example} to show that the social welfare gap between the best stable fractional and the best stable integral matchings can be \emph{arbitrarily} large. We show that the favorable welfare properties of stable fractional matchings come at the cost of limiting the algorithmic tools at our disposal. Specifically, we show that the set of stable fractional matchings can be \emph{non-convex} (\Cref{prop:NonConvexity}), and that, in the worst case, stable fractional matchings can have a \emph{large support} (\Cref{thm:LowerBoundSupportSize}), thus prohibiting the use of support enumeration algorithms. Nevertheless, we present simple algorithms for \OptStab\ with approximation ratio of $1+\sigma_{\max}/\sigma_{\min}$, where $\sigma_{\max}$ and $\sigma_{\min}$ represent the maximum and minimum positive valuation in the input instance, respectively (\Cref{thm:alg-summary}). For the variant \OptEpsStab, where the stability constraints are relaxed by a multiplicative factor of $(1-\eps)$, an embarrassingly simple algorithm computes $1/\eps$-approximate solutions (\Cref{thm:ApproxStable+ApproxOptimal}). We then proceed to our main results (\Cref{thm:strong-inapprox,thm:strong-eps-inapprox}), which show that these approximation guarantees are---somewhat surprisingly---almost the best achievable via polynomial-time algorithms (unless $\text{P}=\text{NP}$). To the best of our knowledge, these are the first computational complexity results for stable fractional matchings in the cardinal preferences model.

The rest of the paper is organized as follows. We present related work in \cref{sec:Related_Work}. We continue in \cref{sec:Preliminaries} with preliminary definitions and warm up with exponential-time algorithms that solve \OptStab\ using linear programming. \cref{sec:StructuralObservations} presents the structural properties of (nearly)-optimal solutions of \OptStab{}. Our algorithms are presented in \cref{sec:TractabilityResults}, and our inapproximability results are presented in \cref{sec:IntractabilityResults}. Some proofs and additional material appear in the appendix. 
\subsection{Related work}
\label{sec:Related_Work}
The stable matching problem has been extensively studied for \emph{integral} matchings. The universal existence of a stable integral matching~\citep{GS62stable} has led to considerable research on stable solutions that, in addition, optimize various measures of fairness or efficiency. A relevant example is that of optimizing the \emph{average rank} of matched partners, which, in our model, corresponds to finding an optimal stable integral matching when the cardinal utilities are completely specified by the ordinal ranks (e.g., if an agent is ranked at position $i$, then it is valued at $n+1-i$). This problem is known to admit combinatorial polynomial-time algorithms~\citep{ILG87efficient,MV18}. Other examples of such problems include minimizing the difference between average ranks of matched partners of men and women~\citep{K93complexity}, optimizing the rank of matched partner for the least well-off agent~\citep{G87three,F92new,K97stable}, and maximizing (or minimizing) the cardinality of the matching~\citep{IMM+99stable,MII+02hard,IMO09stable}.

A growing body of work in artificial intelligence and multiagent systems has studied variants of the stable matching problem motivated by practical considerations such as minimizing the amount of information exchange required in arriving at stable outcomes~\citep{DB13elicitation,RGM16preference} or ensuring strategyproofness in many-to-one matchings with quota constraints~\citep{GIK+16strategyproof,HSK+17strategy}. Various other papers have used computational approaches such as SAT solving~\citep{GPS+02sat,DB15sat,PDB16strategy} and constraint programming~\citep{GIM+01constraint, MMT17almost} in developing practically efficient solutions to computationally hard variants of the problem.%, as well as in axiomatic study of matching procedures~\citep{E20analysis}.

Stability has also been studied in the context of \emph{fractional} matchings. Starting with the works of \cite{V89linear}, \cite{R92characterization}, and \cite{RRV93stable}, there is now a well-developed literature on linear programming formulations of the stable matching problem~\citep{TS98geometry,V12stable,CDH+12maximum}. This line of work has led to novel stability notions such as \emph{strong stability}~\citep{RRV93stable}, \emph{ex-post stability}~\citep{RRV93stable}, and \emph{fractional stability}~\citep{V89linear}, which are discussed further in \Cref{sec:Stability_Notions}.\footnote{The term ``fractional stable matching'' has been overloaded in the literature. For example, \cite{TS98geometry} use it to refer to a feasible solution of stable matching linear program, and \cite{AF03lemma} and \cite{BF16fractional} use it in the study of hypergraphic preference systems to refer to a slightly different solution concept. We refer the reader to~\citep{AK19random} for a detailed overview of these notions.} A noteworthy difference with our work is that these notions have been studied with respect to purely \emph{ordinal} preferences.

Turning to \emph{cardinal} preferences, we note that the notions of exact and approximate stability studied by us (\Cref{def:Stability}) first appeared, to the best of our knowledge, in the work of \cite{ADN13anarchy}. They focus on \emph{qualitative}, rather than computational, questions such as analyzing the ``price of anarchy'' for integral stable matchings under various preference structures as well as its extensions to approximate stability.\footnote{The price of anarchy in this context is the worst-case multiplicative welfare gap between an optimal matching and an integral stable matching.} ~\cite{PRV+13stability} study the same notion with the goal of computing integral stable matchings that, in addition, satisfy economic efficiency (in particular, Pareto optimality and its variants). They also study strategic aspects which are an exciting avenue for future research even in our model.

An interesting special case of our problem is when agents have \emph{symmetric} preferences, i.e., for every $(m,w) \in M \times W$, $U(m,w) = V(m,w)$. \cite{DMS17computational} study this model in the context of \emph{integral} matchings, and show that computing a welfare-maximizing integral stable matching is NP-hard under symmetric valuations. In general, stable fractional matchings can have much higher welfare compared to integral ones (\Cref{eg:Motivating_Example}), and thus it is not clear a priori whether the result of \cite{DMS17computational} for integral matchings has any implications for stable fractional matchings. Nevertheless, as we show in \Cref{lem:Symmetric_Ternary_Vals_OptStab_Is_Integral} in \Cref{subsec:Symmetric_Valuations}, their result implies NP-hardness of \OptStab. In comparison, our results in \Cref{sec:IntractabilityResults} on the hardness of approximation are much stronger and also apply to approximate stability.

Finally, we note that stable \emph{fractional} matchings under cardinal preferences have been previously studied in economics literature~\citep{M13stability,HKN+13stability,HK15multilateral,DY16efficiency,EG17ordinal}. Some of these works~\citep{HKN+13stability,HK15multilateral} consider much more general matching models such as matching with transfers, blocking coalitions of arbitrary size, non-linear utilities, etc. The focus in these papers is primarily on \emph{existential} questions (such as the existence of competitive equilibria) or on examining the \emph{logical} relationship among various notions of stability and economic efficiency. Some of these models are strict generalizations of ours, and therefore computational hardness results in our model readily extend to these more general settings. Whether one can obtain stronger inapproximability results for these models is an interesting avenue for future research.

\section{Preliminaries}
\label{sec:Preliminaries}
An instance of \emph{Stable Matching problem with Cardinal preferences} (\SMC{}) is given by the tuple $\langle M, W, U, V \rangle$, where $M \coloneqq \{m_1,\dots,m_n\}$ and $W \coloneqq \{w_1,\dots,w_n\}$ denote the set of $n$ men and $n$ women, respectively, and $U$ and $V$ are $n \times n$ matrices of non-negative rational numbers that specify the \emph{valuations} of the agents. Specifically, $U(m,w)$ is the value derived by man $m$ from his match with woman $w$, and $V(m,w)$ is the value derived by woman $w$ from her match with man $m$. Many of our results will focus on two special classes of valuations, namely {\em binary} (where $U,V \in \{0,1\}^{n \times n}$) and {\em ternary} valuations (where $U,V \in \{0,1,\alpha\}^{n \times n}$ for some $\alpha>1$).

We will often describe an \SMC{} instance using its {\em graph representation}. An instance $\I =$ $\langle M, W, U, V \rangle$ can be represented as a bipartite graph with vertex sets $M$ and $W$, and an edge for every pair $(m,w)\in M\times W$ such that at least one of $U(m,w)>0$ or $V(m,w)>0$ holds. Each edge $(m,w)$ in this graph has two valuations associated with it, namely $U(m,w)$ and $V(m,w)$.

A \emph{fractional matching} $\mu: M \times W \rightarrow \R_{\geq 0}$ is an assignment of non-negative weights to all man-woman pairs such that $\textstyle{ \sum_{w \in W} \mu(m,w) \leq 1 }$ for each $m \in M$ and $\textstyle{ \sum_{m \in M} \mu(m,w) \leq 1 }$ for each $w \in W$.
A fractional matching $\mu$ is said to be \emph{complete} if $\textstyle{ \sum_{w \in W} \mu(m,w)  = 1 }$ for each man $m \in M$ and $\textstyle{ \sum_{m \in M} \mu(m,w) = 1 }$ for each woman $w \in W$. An {\em integral} matching $\mu$ is a fractional matching with weights $\mu(m,w)\in \{0,1\}$ for every pair $(m,w)$. With slight abuse of notation, we sometimes view an integral matching $\mu$ as a set of pairs and write $(m,w)\in \mu$ in place of $\mu(m,w)=1$. Also, unless stated otherwise, we will assume that any fractional/integral matching is complete.

It is well-known, and follows from the Birkhoff-von Neumann (BvN) theorem, that a (complete) fractional matching $\mu$ can be written as a convex combination of $k=\bigO(n^2)$ integral matchings $\mu^{(1)}, \mu^{(2)}, \dots,\mu^{(k)}$ so that for every pair $(m,w)\in M\times W$, we have
\begin{align*}
	\mu(m,w) = \sum_{j=1}^{k}{\lambda_j \cdot\mu^{(j)}(m,w)},
\end{align*}
where $\lambda_j>0$ for all $j \in \{1,\dots,k\}$ and $\sum_{j=1}^k{\lambda_j}=1$.
%$$\mu(m,w) = \sum_{j=1}^{k}{\lambda_j \cdot\mu^{(j)}(m,w)}$$
 The set of integral matchings $\{\mu^{(1)}, \dots, \mu^{(k)}\}$ is called the {\em support} of the fractional matching $\mu$. Note that the support need not be unique.

We proceed with the formal definitions of {\em stability} and {\em approximate stability}, which, in turn, depend on the {\em utility} derived by agents in a fractional matching. In particular,
the utility derived by the man $m$ in $\mu$ is given by $\textstyle{ u_{m}(\mu) \coloneqq \sum_{w \in W} U(m,w) \mu(m,w)} $, and the utility derived by the woman $w$ is given by $\textstyle{ v_w(\mu) \coloneqq \sum_{m\in M} V(m,w) \mu(m,w) }$. 

\begin{definition}[Stability~\citep{ADN13anarchy,PRV+13stability}]
\label{def:Stability}
Given a fractional matching $\mu$, a man-woman pair $(m,w)$ is said to be a \emph{blocking pair} if $u_{m}(\mu) < U(m,w)$ and $v_w(\mu) < V(m,w)$. A fractional matching $\mu$ is \emph{stable} if there are no blocking pairs, i.e., for each $(m,w) \in M \times W$, either $u_{m}(\mu) \geq U(m,w)$ or $v_w(\mu) \geq V(m,w)$.
\end{definition}

Thus, under a stable fractional matching, no pair of man and woman can simultaneously improve by breaking away from the fractional matching and instead being integrally matched with each other. This notion of deviation is also reasonable from the viewpoint of bounded rationality, as agents only form blocking coalitions of size two, and only deviate to an integral matching between the members of the coalition.

\begin{definition}[$\eps$-Stability~\citep{ADN13anarchy}]
\label{def:Approx-Stability}
Given any $\eps \in [0,1)$ and a fractional matching $\mu$, a man-woman pair $(m,w)$ is said to be \emph{$\eps$-blocking} if $u_{m}(\mu) < (1-\eps) U(m,w)$ and $v_w(\mu) < (1-\eps) V(m,w)$; otherwise, the pair is said to be \emph{$\eps$-stable}. A fractional matching $\mu$ is $\eps$-\emph{stable} if all pairs are $\eps$-stable.
\end{definition}

Thus, a $0.01$-stable fractional matching is one in which, for every man-woman pair, at least one of the two agents already receives (at least) $99\%$ of the utility that he or she would receive by being integrally matched with the other. Note that a stable fractional matching is also $\eps$-stable for every $\eps\geq 0$. 

Notice that \Cref{def:Stability,def:Approx-Stability} entail that agents in a blocking pair prefer to switch to an \emph{integral} match with each other. One could also consider an alternative formulation wherein the agents merely prefer to \emph{increase} their mutual fractional engagement, possibly at the expense of weakening other less preferred matches. This is precisely the notion of \emph{strong stability} (see \Cref{sec:Stability_Notions} for the definition) which has been studied in the context of ordinal preferences~\citep{RRV93stable,AK19random}. However, as we note in \Cref{rem:Strong_Stability_Suboptimal} in \Cref{sec:Stability_Notions}, a strongly stable fractional matching can be strictly suboptimal in terms of social welfare, which further justifies the consideration of the stability notion in \Cref{def:Stability}.

The next statement follows from the seminal result of \cite{GS62stable}.

\begin{proposition}\label{prop:GaleShapley}
	Given any \SMC{} instance $\I$, a stable fractional matching $\mu$ for $\I$ always exists and can be computed in polynomial time.
\end{proposition}

\cref{prop:GaleShapley} was originally proven in \citep{GS62stable} in the standard stable matching model with ordinal preferences and integral matchings. It is easy to see that given any \SMC{} instance $\I$, if an integral matching $\mu$ is stable for an ordinal instance derived from $\I$ (where the ordinal preferences of each agent are consistent with its valuations, breaking ties arbitrarily), then it is also stable for the original instance $\I$.

Next, we define \emph{social welfare}, which is a measure of the efficiency of a fractional matching.

\begin{definition}[Social welfare]
	Given an \SMC{} instance $\langle M, W, U, V \rangle$ and a fractional matching $\mu$, the \emph{social welfare} of $\mu$ is defined as
	\begin{align*}
	\W(\mu) \coloneqq \sum_{m \in M} u_{m}(\mu) + \sum_{w \in W} v_w(\mu) = \sum_{m \in M}\sum_{w\in W} (U(m,w) + V(m,w)) \mu(m,w).
	\end{align*}
%	$$ \textstyle{ \W(\mu) \coloneqq \sum_{m \in M} u_{m}(\mu) + \sum_{w \in W} v_w(\mu) } = \sum_{m \in M}\sum_{w\in W} (U(m,w) + V(m,w)) \mu(m,w).$$
\end{definition}
An {\em optimal} matching is one with the highest social welfare among all fractional matchings. It follows from the BvN decomposition that an optimal matching is, without loss of generality, integral. Similarly, an {\em optimal stable} fractional matching (respectively, {\em optimal $\eps$-stable} fractional matching) is one with the highest social welfare among all stable (respectively, all $\eps$-stable) fractional matchings. We will use \OptStab\ and \OptEpsStab\ to refer to the corresponding optimization problems. 

We will also discuss stable fractional matchings that are \emph{approximately optimal}.
\begin{definition}[$\rho$-efficiency]
\label{def:Approx-Efficiency}
For $\rho\in (0,1]$, the term \emph{$\rho$-efficient} will refer to a stable (respectively, $\eps$-stable) fractional matching with welfare at least $\rho$ times the welfare of the optimal stable (respectively, $\eps$-stable) fractional matching. That is, $\mu$ is a $\rho$-efficient stable fractional matching if it is stable (respectively, $\eps$-stable) and $\W(\mu) \geq \rho \W(\mu^*)$, where $\mu^*$ is an optimal stable (respectively, $\eps$-stable) fractional matching.
\end{definition}

Thus, an optimal stable (or $\eps$-stable) fractional matching is $1$-efficient.

\subsection{Computing optimal stable fractional matchings}
\label{subsec:OptimalStableMatchingProgram}
We will now discuss two exponential-time algorithms for \OptStab{}. The first algorithm uses the following mixed integer linear program \ref{Prog:OPT-Stab}:
\begin{alignat}{2}% right & left & right & left
  \lplabel{{Prog:OPT-Stab}}{\textup{(OPT-Stab)}}&&&\nonumber\\
  \text{maximize}   \quad \sum_{m \in M} u_m & + \sum_{w \in W} v_w && \nonumber \\
%&&& \nonumber\\
\text{subject to} \quad u_m & \geq U(m,w) y(m,w) && \quad \forall \, m \in M, w\in W \label{Prog-OPT-Stab:Stability1}\\
\quad v_w & \geq V(m,w)  (1-y(m,w)) && \quad \forall \, m \in M, w\in W \label{Prog-OPT-Stab:Stability2}\\
\quad u_m & = \sum_{w \in W} U(m,w) \mu(m,w)  && \quad \forall \, m \in M \label{Prog-OPT-Stab:Utility-Men}\\
\quad v_w & = \sum_{m \in M} V(m,w) \mu(m,w)  && \quad \forall \, w \in W \label{Prog-OPT-Stab:Utility-Women}\\
 \quad \sum_{w \in W} \mu(m,w) & \leq 1 && \quad \forall \, m \in M \label{Prog-OPT-Stab:Feasibility-Men}\\
\quad \sum_{m \in M} \mu(m,w) & \leq 1 && \quad \forall \, w \in W \label{Prog-OPT-Stab:Feasibility-Women}\\
\quad \mu(m,w) & \geq 0 && \quad \forall \, m \in M, w\in W \\
\quad y(m,w) & \in \{0,1\} && \quad \forall \, m \in M, w\in W \label{MILP-OPT-Stab:Integral_Constraint}
\end{alignat}

The non-negative weights $\mu(m,w)$ of man-woman pairs as well as the utilities $u_m\coloneqq u_m(\mu)$ and $v_w\coloneqq v_w(\mu)$ of the agents (set in equalities (\ref{Prog-OPT-Stab:Utility-Men}) and (\ref{Prog-OPT-Stab:Utility-Women})) are the fractional variables of~\ref{Prog:OPT-Stab}. The binary variables $y(m,w)$ encode the stability requirements for pair $(m,w)$ in constraints~(\ref{Prog-OPT-Stab:Stability1}) and~(\ref{Prog-OPT-Stab:Stability2}). Indeed, by setting $y(m,w)$ to $1$ or $0$, we can require either $u_m(\mu) \geq U(m,w)$ or $v_w(\mu) \geq V(m,w)$. %In other words, \ref{Prog:OPT-Stab} is a linear program with \emph{disjunctive} constraints.
Constraints (\ref{Prog-OPT-Stab:Feasibility-Men}) and~(\ref{Prog-OPT-Stab:Feasibility-Women}) ensure feasibility. By enumerating over all possible combinations of values for the binary variables $y(m,w)$ for $(m,w)\in M\times W$, we get $2^{n^2}$ different linear programs, and at least one of them must have the optimal stable fractional matching as its optimal solution.

Our second algorithm is slightly faster and solves at most $\bigO(n^n)$ linear programs.
 It exploits the following linear program~\ref{LP:OPT-Thresh}, which is defined using non-negative constants $\theta_m$ for $m\in M$ and $\theta_w$ for $w\in W$, which we call {\em utility thresholds}.\smallskip
\begin{alignat}{2}% right & left & right & left
  \lplabel{{LP:OPT-Thresh}}{\textup{(OPT-Thresh)}}&&&\nonumber \\
  \text{maximize} \quad \sum_{m \in M} u_m & + \sum_{w \in W} v_w && \nonumber \\
%&&& \nonumber\\
\text{subject to} \quad u_m & \geq \theta_m  && \quad \forall \, m \in M \label{LP-OPT-Thresh:Utility-Men-LowerBound}\\
\quad v_w & \geq \theta_w  && \quad \forall \, w \in W \label{LP-OPT-Thresh:Utility-Women-LowerBound}\\
 \quad u_m & = \sum_{w \in W} U(m,w) \mu(m,w)  && \quad \forall \, m \in M \label{LP-OPT-Thresh:Utility-Men}\\
\quad v_w & = \sum_{m \in M} V(m,w) \mu(m,w)  && \quad \forall \, w \in W \label{LP-OPT-Thresh:Utility-Women}\\
 \quad \sum_{w \in W} \mu(m,w) & \leq 1 && \quad \forall \, m \in M \label{LP-OPT-Thresh:Feasibility-Men}\\
\quad \sum_{m \in M} \mu(m,w) & \leq 1 && \quad \forall \, w \in W \label{LP-OPT-Thresh:Feasibility-Women}\\
\quad \mu(m,w) & \geq 0 && \quad \forall \, m\in M, w\in W
 \label{LP-OPT-Thresh:Nonnegativity}
\end{alignat}

When all utility thresholds are set to zero, the solution of \ref{LP:OPT-Thresh} is an optimal (i.e., welfare-maximizing) fractional matching. Using \ref{LP:OPT-Thresh} to maximize social welfare under stability constraints is more challenging. We say that a set of utility thresholds is {\em stability-preserving} if for every pair of agents $m\in M$ and $w\in W$, either $\theta_m\geq U(m,w)$ or $\theta_w\geq V(m,w)$. Note that any fractional matching $\mu$ that is feasible for \ref{Prog:OPT-Stab} is also feasible for \ref{LP:OPT-Thresh} for some stability-preserving set of utility thresholds (in particular, set the threshold of an agent equal to its utility under $\mu$). Conversely, any fractional matching $\mu$ that is feasible for \ref{LP:OPT-Thresh} with some set of stability-preserving utility thresholds is also feasible for \ref{Prog:OPT-Stab}. 

One could now adopt the following strategy to solve \OptStab{}: First, enumerate all $\bigO(n^n)$ tuples of utility thresholds $(\theta_{m_1},\dots,\theta_{m_n})$ with $\theta_m\in \{U(m,w):w\in W\}$ for every man $m\in M$. Next, for every choice of $(\theta_{m_1},\dots,\theta_{m_n})$, solve \ref{LP:OPT-Thresh} after appropriately setting $(\theta_{w_1},\dots,\theta_{w_n})$ where $\theta_w\in \{V(m,w):m\in M\}$ for all $w\in W$, so that the set of utility thresholds is stability-preserving. Among these solutions, the fractional matching with highest social welfare will be the solution of \OptStab. We note that \ref{LP:OPT-Thresh} resembles the integer programming formulations with cut-off variables in other works~\citep{ABM16integer,DGG+19mathematical}.

\subsection{The case of symmetric valuations}
\label{subsec:Symmetric_Valuations}

As mentioned previously in \Cref{sec:Related_Work}, a result by \cite{DMS17computational} implies NP-hardness of the problem of computing an optimal stable \emph{integral} matching on \SMC{} instances. The construction of \cite{DMS17computational} involves a restricted class of \SMC{} instances wherein the agents have \emph{symmetric} (i.e., $U=V$) and  \emph{ternary} valuations in $\{0,1,\alpha\}$ with $\alpha \in (1,2)$.

\begin{proposition}\label{prop:DeligkasEtAl}
	Computing an optimal stable integral matching is NP-hard even for \SMC{} instances with symmetric ternary valuations.
\end{proposition}

Our next result (\Cref{lem:Symmetric_Ternary_Vals_OptStab_Is_Integral}) shows that for these \SMC{} instances, the optimal stable \emph{fractional} matching is, without loss of generality, \emph{integral}. Consequently, the result in \citep{DMS17computational} implies NP-hardness of \OptStab.

\begin{lemma}	\label{lem:Symmetric_Ternary_Vals_OptStab_Is_Integral}
	Let $\I$ be an \SMC{} instance with symmetric and ternary valuations, and let $\mu^*$ be an optimal stable fractional matching for $\I$. Then, there exists a stable integral matching $\mu^\textup{s}$ such that $\W(\mu^{\textup{s}}) = \W(\mu^*)$.
\end{lemma}

%\proof{\PROOF.}
\begin{proof}
Consider an optimal stable fractional matching $\mu^*$ that is not integral (i.e., has support of size at least two), and let $\mu$ be any integral matching in any support of $\mu^*$. We will show that $\mu$ is a stable integral matching. This would imply that all matchings in the support of $\mu^*$ are optimal stable integral matchings, as desired.

Assume, for contradiction, that $\mu$ is not stable, and let $(m,w)$ be a blocking pair in $\mu$. Then, because of symmetry, either $U(m,w)=V(m,w)=1$ or $U(m,w)=V(m,w)=\alpha$. 

Suppose that $U(m,w)=V(m,w)=1$. Then, either both $m$ and $w$ are unmatched in $\mu$, or any pair $(m',w')$ in $\mu$ with either $m'=m$ or $w'=w$ satisfies $U(m',w')=V(m',w')=0$. By replacing such pairs with $(m,w)$ in $\mu$ (and, subsequently, in the support of $\mu^*$), we get a stable fractional matching with a higher welfare than $\mu^*$, contradicting its optimality.

Now suppose that $U(m,w)=V(m,w)=\alpha$. Then, either both $m$ and $w$ are unmatched in $\mu$, or any pair $(m',w')$ in $\mu$ with either $m'=m$ or $w'=w$ satisfies $U(m',w')=V(m',w') \in \{0,1\}$. However, since $\mu$ is in the support of $\mu^*$ (i.e., with strictly positive weight) and all valuations are less than or equal to $\alpha$, the utilities of both $m$ and $w$ in $\mu^*$ will be strictly smaller than $\alpha$, contradicting the stability of $\mu^*$.
\end{proof}
%\endproof

From \cref{lem:Symmetric_Ternary_Vals_OptStab_Is_Integral} and \cref{prop:DeligkasEtAl}, we immediately have the following:

\begin{corollary}
	\label{cor:deligkas}
\OptStab\ is NP-hard.
\end{corollary}
\section{Structural properties}
\label{sec:StructuralObservations}
In this section, we present several observations about the structure of optimal and nearly-optimal stable fractional matchings. We begin by considerably strengthening our observation in Example~\ref{eg:Motivating_Example} regarding the welfare gap between stable fractional and stable integral matchings.

\begin{theorem}\label{thm:gap}
	For every $\delta>0$ and $\alpha\geq 2$, there exists an \SMC{} instance with ternary valuations in $\{0,1,\alpha\}$ and an optimal stable fractional matching $\mu^*$ such that any stable integral matching $\mu^s$ satisfies $\W(\mu^s)\leq \left(\alpha-\frac{1}{2}-\delta\right)^{-1} \W(\mu^*)$.
\end{theorem}

\begin{figure}[ht]
	\footnotesize
	\centering
		\begin{tikzpicture}[scale=0.75]
		\tikzset{man/.style = {shape=circle,draw,inner sep=1pt}}
		\tikzset{woman/.style = {shape=diamond,draw,inner sep=1pt}}
		\tikzset{edge/.style = {solid}}
		\tikzset{earedge/.style = {densely dotted}}
		% vertices
		\node[man] (3) at (2,0)  {$m_1$};
		\node[woman] (4) at (4,0)  {$w_1$};
		\node[man] (5) at (6,0)  {$m_2$};
		\node[man] (6) at (4,-2)  {$e_1^2$};
		\node[woman] (7) at (6,-2)  {$f_2^3$};
		\node[man] (9) at (4,2)  {$e_1^1$};
		\node[man] (10) at (6,4)  {$e_2^3$};
		\node[woman] (11) at (6,2)  {$f_2^2$};
		\node[woman] (17) at (8,0)  {$w_2$};
		\node[man] (18) at (8,-2)  {$e_2^2$};
		\node[man] (19) at (8,2)  {$e_2^1$};
		\node[woman] (20) at (8,4)  {$f_2^1$};
		\node[man] (12) at (11,0)  {$m_k$};
		\node[woman] (13) at (13,0)  {$w_k$};
		\node[man] (14) at (11,4)  {$e_k^3$};
		\node[woman] (15) at (11,2)  {$f_k^2$};
		\node[woman] (16) at (11,-2)  {$f_k^3$};
		% edges
		\draw[edge] (3) to node [near start,fill=white,inner sep=2pt] (324) {\footnotesize{$0$}} node [near end,fill=white,inner sep=2pt] (324) {\footnotesize{$\alpha$}} (4);
		\draw[edge] (4) to node [near start,fill=white,inner sep=2pt] (425) {\footnotesize{$0$}} node [near end,fill=white,inner sep=2pt] (425) {\footnotesize{$\alpha$}} (5);
		\draw[edge] (4) to node [near start,fill=white,inner sep=2pt] (426) {\footnotesize{$0$}} node [near end,fill=white,inner sep=2pt] (426) {\footnotesize{$\alpha$}} (6);
		\draw[edge] (5) to node [near start,fill=white,inner sep=2pt] (527) {\footnotesize{$0$}} node [near end,fill=white,inner sep=2pt] (527) {\footnotesize{$\alpha$}} (7);
		\draw[edge] (5) to node [near start,fill=white,inner sep=2pt] (527) {\footnotesize{$0$}} node [near end,fill=white,inner sep=2pt] (527) {\footnotesize{$\alpha$}} (17);
		\draw[edge] (9) to node [near start,fill=white,inner sep=2pt] (924) {\footnotesize{$1$}} node [near end,fill=white,inner sep=2pt] (924) {\footnotesize{$1$}} (4);
		\draw[edge] (11) to node [near start,fill=white,inner sep=2pt] (11210) {\footnotesize{$0$}} node [near end,fill=white,inner sep=2pt] (11210) {\footnotesize{$\alpha$}} (10);
		\draw[edge] (5) to node [near start,fill=white,inner sep=2pt] (5211) {\footnotesize{$1$}} node [near end,fill=white,inner sep=2pt] (5211) {\footnotesize{$1$}} (11);
		\draw[edge] (17) to node [very near start,fill=white,inner sep=2pt] (5212) {\footnotesize{$0$}} node [midway,fill=white,inner sep=2pt] (5212) {$\dots$} node [very near end,fill=white,inner sep=2pt] (5212) {\footnotesize{$\alpha$}} (12);	
		\draw[edge] (17) to node [near start,fill=white,inner sep=2pt] (527) {\footnotesize{$0$}} node [near end,fill=white,inner sep=2pt] (527) {\footnotesize{$\alpha$}} (18);
		\draw[edge] (17) to node [near start,fill=white,inner sep=2pt] (527) {\footnotesize{$1$}} node [near end,fill=white,inner sep=2pt] (527) {\footnotesize{$1$}} (19);
		\draw[edge] (19) to node [near start,fill=white,inner sep=2pt] (527) {\footnotesize{$0$}} node [near end,fill=white,inner sep=2pt] (527) {\footnotesize{$\alpha$}} (20);
		\draw[edge] (12) to node [near start,fill=white,inner sep=2pt] (12213) {\footnotesize{$0$}} node [near end,fill=white,inner sep=2pt] (12213) {\footnotesize{$\alpha$}} (13);
		\draw[edge] (15) to node [near start,fill=white,inner sep=2pt] (15214) {\footnotesize{$0$}} node [near end,fill=white,inner sep=2pt] (15214) {\footnotesize{$\alpha$}} (14);
		\draw[edge] (15) to node [near start,fill=white,inner sep=2pt] (15212) {\footnotesize{$1$}} node [near end,fill=white,inner sep=2pt] (15212) {\footnotesize{$1$}} (12);
		\draw[edge] (12) to node [near start,fill=white,inner sep=2pt] (12216) {\footnotesize{$0$}} node [near end,fill=white,inner sep=2pt] (12216) {\footnotesize{$\alpha$}} (16);
		\end{tikzpicture}
	\caption{The \SMC{} instance used in the proof of \cref{thm:gap}. As a convention, in graph representations where the two sides of the bipartition do not appear as left and right, we use circles to represent men and diamonds to represent women.}
	\label{fig:gap}
\end{figure}
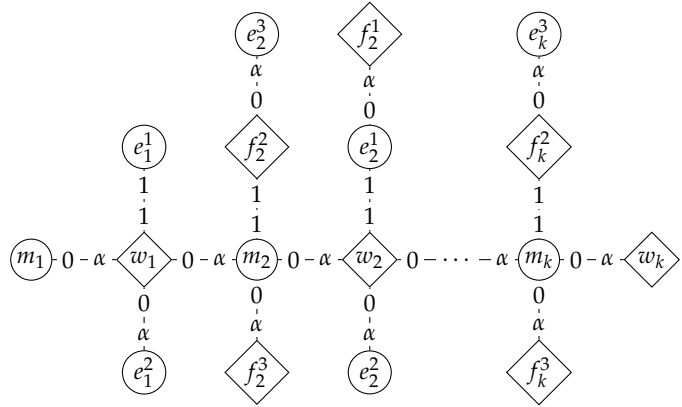

%\proof{\PROOF.}
\begin{proof}
Consider the \SMC{} instance shown in \Cref{fig:gap}, which, for some $k$ (to be determined later), consists of man $m_i$ and woman $w_i$ for all $i \in \{1,\dots,k\}$, men $e_i^1$ and $e_i^2$ for all $i \in \{1,\dots,k-1\}$, woman $f_i^1$ for all $i \in \{2,\dots,k-1\}$, man $e_i^3$ and women $f_i^2$ and $f_i^3$ for all $i \in \{2,\dots,k\}$. 

To specify the valuations, we will use the following notation: For any $a,b \geq 0$, we will say that a man-woman pair $(m,w)$ is an ``$a$\textemdash$b$'' edge if $U(m,w) = a$ and $V(m,w)=b$. In the instance in \Cref{fig:gap}, the pairs $\{(m_i,w_i)\}_{i=1}^{k}$, $\{(m_i,f_i^3)\}_{i=2}^{k}$, and $\{(e_i^1,f_i^1)\}_{i=2}^{k-1}$ are $0$\textemdash$\alpha$ edges, the pairs $\{(m_{i+1},w_i)\}_{i=1}^{k-1}$, $\{(e_i^2,w_i)\}_{i=1}^{k-1}$, and $\{(e_i^3,f_i^2)\}_{i=2}^{k}$ are $\alpha$\textemdash$0$ edges, and the pairs $\{(e_i^1,w_i)\}_{i=1}^{k-1}$ and $\{(m_i,f_i^2)\}_{i=2}^{k}$ are $1$\textemdash$1$ edges. All other pairs are $0$\textemdash$0$ edges. We remark that a slightly modified version of the instance in \Cref{fig:gap} will be used again later in the proof of \cref{thm:strong-inapprox} (refer to the \emph{accumulator gadget} in Figure~\ref{fig:acc}a).

Consider a stable integral matching $\mu^s$. Since the pair $(e^1_1,w_1)$ is a $1$\textemdash$1$ edge, the stability requirement dictates that either $(m_1,w_1)$ or $(e^1_1,w_1)$ is contained in $\mu^s$. Either of these pairs contribute at most $\alpha$ to the social welfare (recall that $\alpha \geq 2$). Additionally, since $\mu^s$ is an integral matching, we get that $\mu^s(m_2,w_1) = 0$. This, in turn, forces the pair $(m_2,f_2^2)$ to be contained in $\mu^s$ as well (or, otherwise, both $m_2$ and $f_2^2$ will have zero utility and violate stability). Continuing in this manner, we observe that the stability requirement for the pairs $\{(m_i,f_i^2)\}_{i=2}^{k}$ and $\{(e_i^1,w_i)\}_{i=2}^{k-1}$ will force these pairs to be contained in $\mu$ as well. Each of these pairs contributes $2$ to the social welfare, and thus, $\W(\mu^s) \leq 4k-6+\alpha$. 
	
Now define a stable \emph{fractional} matching $\mu$ as follows: Set $\mu(m_i,w_i)=1/\alpha$ for all $i \in \{1,\dots,k\}$, $\mu(e^2_i,w_i)=1-2/\alpha$, $\mu(m_{i+1},w_i)=1/\alpha$, $\mu(e^1_i,w_i)=0$ for all $i \in \{1,\dots,k-1\}$, $\mu(m_i,f^2_i)=0$, $\mu(m_i,f^3_i)=1-2/\alpha$, $\mu(e^3_i,f^2_i)=1$ for all $i \in \{2,\dots,k\}$, and $\mu(e^1_i,f^1_i)=1$, for all $i \in \{2,\dots,k-1\}$. (A similar matching will be used later in the proof of \cref{lem:sat}.) It is easy to verify that the social welfare of $\mu$ is $\W(\mu)=4k(\alpha-1/2)-5\alpha+3$. Furthermore, for each $i \in \{1,\dots,k\}$, both man $m_i$ and woman $w_i$ have utility $1$ in $\mu$ (i.e., 
$u_{m_i}(\mu) = v_{w_i}(\mu) = 1$). Every $1$\textemdash$1$ edge includes one of these agents, and these are the only edges that need to be checked for stability (any other pair has at least one agent with a valuation of zero and, hence, the stability constraint for those pairs is trivially satisfied). Therefore, $\mu$ is stable.
The theorem follows by setting $k$ to be sufficiently large; specifically, for any given $\delta > 0$, choosing $k \geq \frac{1}{4} \cdot \left( 6 - \alpha + \frac{\alpha^2 - 23\alpha /2 + 6}{\delta} \right)$ gives the desired bound.
\end{proof}
%\endproof

We emphasize that \cref{thm:gap} is a \emph{positive} result as it establishes that stable fractional matchings can have much higher welfare than their integral counterparts, and highlights the importance of \OptStab.
 
Our next observation (\Cref{prop:NonConvexity}) shows that the set of stable fractional matchings can be \emph{non-convex} even for binary valuations. Interestingly, this does not prove to be a barrier in efficiently solving \OptStab{} in this setting (see \Cref{thm:MaxWelfareBinaryVals} in \Cref{sec:TractabilityResults}).

\begin{proposition}
	\label{prop:NonConvexity}
	There exists an \SMC{} instance with binary valuations for which the set of stable fractional matchings is non-convex.
\end{proposition}

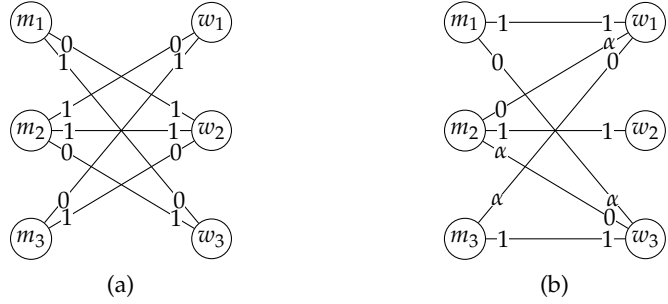
\begin{figure}[t]%[ht]
	\centering
	\footnotesize
	%\scriptsize
%	\vspace{-0.25in}
\subfloat[]{
	\begin{tikzpicture}[scale=1.2]%{0.33\textwidth}
\tikzset{man/.style = {shape=circle,draw,inner sep=1pt}}
\tikzset{woman/.style = {shape=circle,draw,inner sep=1pt}}
\tikzset{edge/.style = {solid}}
\node[man]   (1) at (0,2.4) {$m_1$};
\node[man]   (2) at (0,1.2) {$m_2$};
\node[man]   (3) at (0,0)   {$m_3$};
\node[woman] (4) at (2,2.4) {$w_1$};
\node[woman] (5) at (2,1.2) {$w_2$};
\node[woman] (6) at (2,0)   {$w_3$};

\draw[edge] (1) to node [very near start,fill=white,inner sep=0pt] (151) {$0$} node [very near end,fill=white,inner sep=0pt] (152) {$1$} (5);
\draw[edge] (1) to node [very near start,fill=white,inner sep=0pt] (161) {$1$} node [very near end,fill=white,inner sep=0pt] (162) {$0$} (6);
\draw[edge] (2) to node [very near start,fill=white,inner sep=0pt] (241) {$1$} node [very near end,fill=white,inner sep=0pt] (242) {$0$} (4);
\draw[edge] (2) to node [very near start,fill=white,inner sep=0pt] (251) {$1$} node [very near end,fill=white,inner sep=0pt] (252) {$1$} (5);
\draw[edge] (2) to node [very near start,fill=white,inner sep=0pt] (261) {$0$} node [very near end,fill=white,inner sep=0pt] (262) {$1$} (6);
\draw[edge] (3) to node [very near start,fill=white,inner sep=0pt] (341) {$0$} node [very near end,fill=white,inner sep=0pt] (342) {$1$} (4);
\draw[edge] (3) to node [very near start,fill=white,inner sep=0pt] (351) {$1$} node [very near end,fill=white,inner sep=0pt] (352) {$0$} (5);
\end{tikzpicture}
}
\hspace{1in}
\subfloat[]{
	\begin{tikzpicture}[scale=1.2]%[0.33\textwidth]
\tikzset{man/.style = {shape=circle,draw,inner sep=1pt}}
\tikzset{woman/.style = {shape=circle,draw,inner sep=1pt}}
\tikzset{edge/.style = {solid}}
\node[man]   (1) at (0,2.4) {$m_1$};
\node[man]   (2) at (0,1.2) {$m_2$};
\node[man]   (3) at (0,0)   {$m_3$};
\node[woman] (4) at (2,2.4) {$w_1$};
\node[woman] (5) at (2,1.2) {$w_2$};
\node[woman] (6) at (2,0)   {$w_3$};
\draw[edge] (1) to node [very near start,fill=white,inner sep=0pt] (141) {$1$} node [very near end,fill=white,inner sep=0pt] (142) {$1$} (4);
\draw[edge] (1) to node [very near start,fill=white,inner sep=0pt] (161) {$0$} node [very near end,fill=white,inner sep=0pt] (162) {$\alpha$} (6);
\draw[edge] (2) to node [very near start,fill=white,inner sep=0pt] (241) {$0$} node [very near end,fill=white,inner sep=0pt] (242) {$\alpha$} (4);
\draw[edge] (2) to node [very near start,fill=white,inner sep=0pt] (251) {$1$} node [very near end,fill=white,inner sep=0pt] (252) {$1$} (5);
\draw[edge] (2) to node [very near start,fill=white,inner sep=0pt] (261) {$\alpha$} node [very near end,fill=white,inner sep=0pt] (262) {$0$} (6);
\draw[edge] (3) to node [very near start,fill=white,inner sep=0pt] (341) {$\alpha$} node [very near end,fill=white,inner sep=0pt] (342) {$0$} (4);
%
%\draw[edge] (3) to node [very near start,fill=white,inner sep=0pt] (351) {$0$} node [very near end,fill=white,inner sep=0pt] (352) {$0$} (5);
%
\draw[edge] (3) to node [very near start,fill=white,inner sep=0pt] (361) {$1$} node [very near end,fill=white,inner sep=0pt] (362) {$1$} (6);
\end{tikzpicture}
}
\caption{The \SMC{} instances used in the proofs of \cref{prop:NonConvexity} and \cref{prop:UnstableSupport}.\label{fig:structural}}
\end{figure}

%\proof{\PROOF.}
\begin{proof}
Consider the instance $\I = \langle M,W,U,V \rangle$ with three men $m_1,m_2,m_3$ and three women $w_1,w_2,w_3$, whose graph representation and agent valuations are 
shown in Figure~\ref{fig:structural}a. Consider the integral matchings $\mu^{(1)} \coloneqq \{(m_1,w_3),(m_2,w_1),(m_3,w_2)\}$ and $\mu^{(2)} \coloneqq \{(m_1,w_2),(m_2,w_3),(m_3,w_1)\}$. It is easy to verify that both $\mu^{(1)}$ and $\mu^{(2)}$ are stable for $\I$. However, the fractional matching $\mu \coloneqq 0.5\mu^{(1)} + 0.5\mu^{(2)}$ is not stable since $(m_2,w_2)$ is a blocking pair; indeed, $0.5 = u_{m_2}(\mu) < U(m_2,w_2) = 1$ and $0.5 = v_{w_2}(\mu) < V(m_2,w_2) =~1$.
\end{proof}
%\endproof

\begin{remark}
A follow-up work to the conference version of our paper has shown that non-convexity also holds for strict preferences~\citep{NN20study}.
\label{rem:Nonconvexity_Strict}
\end{remark}

\begin{remark}
It is worth comparing the non-convexity results in \Cref{prop:NonConvexity} and \Cref{rem:Nonconvexity_Strict} with similar results for other stability notions. In particular, \citet[Theorem 1]{AK19random} have shown that the set of strongly stable matchings (see \Cref{sec:Stability_Notions} for the definition) can be non-convex. In \Cref{sec:Non-convexity_Implication}, we revisit the counterexample used in their proof, which is stated in terms of ordinal preferences, and show that there exists a realization of cardinal preferences consistent with their instance such that the relevant convex combination of (strongly) stable matchings is still stable. Thus, the non-convexity result of \cite{AK19random} for strong stability does not imply the same for the set of stable matchings.
\end{remark}

The structure of stable fractional matchings becomes even more interesting (and, as we will see in \Cref{sec:IntractabilityResults}, also computationally unwieldy) when we move to \emph{ternary} valuations. It turns out that the support of a stable fractional matching can comprise entirely of \emph{unstable} integral matchings (\Cref{prop:UnstableSupport}), and its size can grow \emph{linearly} with the input (\Cref{thm:LowerBoundSupportSize}). These observations pose major limitations on the set of algorithmic tools at our disposal.

\begin{proposition}
 \label{prop:UnstableSupport}
 There exists an \SMC{} instance with ternary valuations and a stable fractional matching $\mu$ such that every integral matching in any support of $\mu$ is unstable.
\end{proposition}

\begin{proof}
Consider the \SMC{} instance $\I = \langle M,W,U,V \rangle$ with three men and three women 
%, with the graph representation of
shown in Figure~\ref{fig:structural}b. The parameter $\alpha\geq 3$ is a constant. There are six different perfect integral matchings:
\begin{itemize}
	\item Matching $\mu^{(1)}$, which consists of pairs $(m_1,w_1)$, $(m_2,w_2)$, and $(m_3,w_3)$ and has a social welfare of $6$. It is easy to verify that this is the unique stable integral matching. Also, any subset of $\mu^{(1)}$ is not stable as the pair that is missing from $\mu^{(1)}$ will be blocking.
	\item Matching $\mu^{(2)}$, which consists of pairs $(m_1,w_2)$, $(m_2,w_3)$, and $(m_3,w_1)$ and has a social welfare of $2\alpha$. The matching is not stable since the pair $(m_1,w_1)$ is blocking.
	\item Matching $\mu^{(3)}$, which consists of pairs $(m_1,w_3)$, $(m_2,w_1)$, and $(m_3,w_2)$ and has a social welfare of $2\alpha$. It is not stable since $(m_2,w_2)$ is blocking.
	\item Matching $\mu^{(4)}$, which consists of pairs $(m_1,w_1)$, $(m_2,w_3)$, and $(m_3,w_2)$ and has a social welfare of $\alpha+2$. It is not stable since $(m_3,w_3)$ is blocking.
	\item Matching $\mu^{(5)}$, which consists of pairs $(m_1,w_3)$, $(m_2,w_2)$, and $(m_3,w_1)$ and has a social welfare of $2\alpha+2$. It is not stable since $(m_1,w_1)$ is blocking.
	\item Matching $\mu^{(6)}$, which consists of pairs $(m_1,w_2)$, $(m_2,w_1)$, and $(m_3,w_3)$ and has a social welfare of $\alpha+2$. It is not stable since the pair $(m_2,w_2)$ is blocking.
\end{itemize}

Consider the matching $\textstyle{ \mu \coloneqq \frac{1}{\alpha(\alpha-1)} \cdot \mu^{(2)} + \frac{1}{\alpha} \cdot \mu^{(3)} + \frac{\alpha-2}{\alpha-1} \cdot \mu^{(5)} }$. It is easy to verify that $\mu$ is stable. Indeed, the utilities of the agents in $\mu$ are given by $u_{m_1}(\mu) = 0$, $u_{m_2}(\mu) = 1$, $u_{m_3}(\mu) = \alpha-1$, $v_{w_1}(\mu) = 1$, $v_{w_2}(\mu) = \textstyle{ \frac{\alpha-2}{\alpha-1} }$ and $v_{w_3}(\mu) = \textstyle{ \alpha - \frac{1}{\alpha-1} }$. Notice that only the pairs $(m_1,w_1)$, $(m_2,w_2)$, and $(m_3,w_3)$ need to be checked for stability. %, since any other pair has at least one agent with a valuation of zero (and, hence, the stability constraint for those pairs is trivially satisfied). 
For each of these pairs, at least one of the agents has a utility of at least $1$ in $\mu$, implying that $\mu$ is stable. Finally, notice that $\mu(m_1,w_1) = 0$, which means that the unique stable integral matching $\mu^{(1)}$ cannot occur in any support of $\mu$.\end{proof}

We remark that with some extra work, one can show that the matching $\mu$ in the proof of \cref{prop:UnstableSupport} is the unique optimal stable fractional matching.

One might wonder whether the support of an optimal stable fractional matching always consists of an optimal integral matching. This turns out to not be the case, as illustrated by the following instance: Consider two men $m_1,m_2$ and two women $w_1,w_2$ with valuations given by $U = \begin{bsmallmatrix}
2 & 0 \\
1 & 0 \\
\end{bsmallmatrix}$ and $V = \begin{bsmallmatrix}
0 & 0 \\
1 & 2 \\
\end{bsmallmatrix}$. The unique optimal integral matching is $\{(m_1,w_1),(m_2,w_2)\}$ while the unique optimal stable fractional matching is $\{(m_1,w_2),(m_2,w_1)\}$.

As mentioned previously in Section~\ref{sec:Preliminaries}, a (stable) fractional matching is the convex combination of at most $n^2$ integral ones. \cref{thm:UpperBoundSupportSize} provides a stronger bound on the support size of an \emph{optimal} stable fractional matching.

\begin{theorem}
 \label{thm:UpperBoundSupportSize}
 Given any \SMC{} instance $\I$, there exists an optimal stable fractional matching for $\I$ with at most $4n$ integral matchings in its support.
\end{theorem}

\begin{proof}
Let $\mu^*$ be an optimal stable fractional matching for $\I$. Recall from Section~\ref{subsec:OptimalStableMatchingProgram} that $\mu^*$ solves the program \ref{LP:OPT-Thresh} for some set of stability-preserving utility thresholds. Observe that \ref{LP:OPT-Thresh} has $n^2$ free variables (we ignore here the $2n$ variables $u_m$ for $m\in M$ and $v_w$ for $w\in W$, which depend on the remaining ones according to constraints (\ref{LP-OPT-Thresh:Utility-Men}) and (\ref{LP-OPT-Thresh:Utility-Women})). Without loss of generality, $\mu^*$ is an optimal {\em extreme point} solution of \ref{LP:OPT-Thresh}. That is, when \ref{LP:OPT-Thresh} is instantiated for $\mu^*$, $n^2$ linearly independent inequality constraints become tight. Among them, at most $4n$ can correspond to the sets of constraints (\ref{LP-OPT-Thresh:Utility-Men-LowerBound}), (\ref{LP-OPT-Thresh:Utility-Women-LowerBound}), (\ref{LP-OPT-Thresh:Feasibility-Men}), and (\ref{LP-OPT-Thresh:Feasibility-Women}). The remaining ones must correspond to the set of constraints (\ref{LP-OPT-Thresh:Nonnegativity}), implying that at least $n^2-4n$ free variables will be equal to zero. Thus, $\mu^*$ can assign positive weights to at most $4n$ man-woman pairs and, consequently, can have at most $4n$ integral matchings in its support.
\end{proof}

Next we show that the bound in \cref{thm:UpperBoundSupportSize} is tight up to a constant factor.

\begin{theorem}
 \label{thm:LowerBoundSupportSize}
 For every $\rho\in (0,1]$, there exists a family of \SMC{} instances with ternary valuations for which any support of a $\rho$-efficient stable fractional matching consists of $\Omega(\rho n)$ integral matchings.
\end{theorem}

\begin{proof}
Consider a family of \SMC{} instances $\I_n = \langle M,W,U,V \rangle$ with $M = \{m_1,\dots,m_n\}$ and $W = \{w_1,\dots,w_n\}$, where $n$ is odd. Let $\alpha$ be such that $\textstyle{ \alpha > \max\left\{n+2,\frac{2n}{\rho(n-1)}\right\}}$.
The (ternary) valuations of the agents are defined as follows: For each $i \in \{1,2,\dots,n\}$, $U(m_i,w_i) = V(m_i,w_i) = 1$. For each $\textstyle{ i \in \{1,2,\dots,\frac{n-1}{2}\} }$, $U(m_{2i},w_1) = U(m_{2i+1},w_{2i}) = V(m_{2i},w_{2i+1}) = \alpha$ and $V(m_{2i},w_1) = V(m_{2i+1},w_{2i}) = U(m_{2i},w_{2i+1}) = 0$. Finally, $U(m_{n},w_1) = 0$ and $V(m_{n},w_1) = \alpha$. For all remaining pairs $(m,w) \in M \times W$, $U(m,w) = V(m,w) = 0$. Figure \ref{fig:LowerBoundSupportSize}a illustrates the \SMC{} instance $\I_5$.

\begin{figure}[t]%[ht]
	\centering
	\scriptsize
	%\vspace{-0.25in}
\subfloat[]{
% \vspace{-0.25in}
% \hspace{-0.3in}
\begin{tikzpicture}[scale=0.8]
 \tikzset{man/.style = {shape=circle,draw,inner sep=0pt}}
	  \tikzset{woman/.style = {shape=circle,draw,inner sep=0pt}}
	  \tikzset{edge/.style = {solid}}
\node[man]   (1)  at (0,3.2)   {$m_1$};
\node[man]   (2)  at (0,2.4)   {$m_2$};
\node[man]   (3)  at (0,1.6)   {$m_3$};
\node[man]   (4)  at (0,0.8)   {$m_4$};
\node[man]   (5)  at (0,0)   {$m_5$};
\node[woman] (6)  at (2.5,3.2) {$w_1$};
\node[woman] (7)  at (2.5,2.4) {$w_2$};
\node[woman] (8)  at (2.5,1.6) {$w_3$};
\node[woman] (9)  at (2.5,0.8) {$w_4$};
\node[woman] (10) at (2.5,0) {$w_5$};
\draw[edge] (1) to node [very near start,fill=white,inner sep=0pt] (161) {\footnotesize{$1$}} node [very near end,fill=white,inner sep=0pt] (162) {\footnotesize{$1$}} (6);
\draw[edge] (2) to node [near start,fill=white,inner sep=0pt] (261) {$\alpha$} node [near end,fill=white,inner sep=0pt] (262) {\footnotesize{$0$}} (6);
\draw[edge] (2) to node [very near start,fill=white,inner sep=0pt] (271) {\footnotesize{$1$}} node [very near end,fill=white,inner sep=0pt] (272) {\footnotesize{$1$}} (7);
\draw[edge] (2) to node [near start,fill=white,inner sep=0pt] (281) {\footnotesize{$0$}} node [near end,fill=white,inner sep=0pt] (282) {$\alpha$} (8);
\draw[edge] (3) to node [near start,fill=white,inner sep=0pt] (371) {$\alpha$} node [near end,fill=white,inner sep=0pt] (372) {\footnotesize{$0$}} (7);
\draw[edge] (3) to node [very near start,fill=white,inner sep=0pt] (381) {\footnotesize{$1$}} node [very near end,fill=white,inner sep=0pt] (382) {\footnotesize{$1$}} (8);
\draw[edge] (4) to node [very near start,fill=white,inner sep=0pt] (461) {$\alpha$} node [very near end,fill=white,inner sep=0pt] (462) {\footnotesize{$0$}} (6);
\draw[edge] (4) to node [very near start,fill=white,inner sep=0pt] (491) {\footnotesize{$1$}} node [very near end,fill=white,inner sep=0pt] (492) {\footnotesize{$1$}} (9);
\draw[edge] (4) to node [near start,fill=white,inner sep=0pt] (4101) {\footnotesize{$0$}} node [near end,fill=white,inner sep=0pt] (4102) {$\alpha$} (10);
\draw[edge,out=120, in=150, looseness=1.85] (5) to node [very near start,fill=white,inner sep=0pt] (561) {\footnotesize{$0$}} node [very near end,fill=white,inner sep=0pt] (562) {$\alpha$} (6);
\draw[edge] (5) to node [near start,fill=white,inner sep=0pt] (591) {$\alpha$} node [near end,fill=white,inner sep=0pt] (592) {\footnotesize{$0$}} (9);
\draw[edge] (5) to node [very near start,fill=white,inner sep=0pt] (5101) {\footnotesize{$1$}} node [very near end,fill=white,inner sep=0pt] (5102) {\footnotesize{$1$}} (10);
\end{tikzpicture}
}
\hspace{0.15in}
\subfloat[]{
\begin{tikzpicture}[scale=0.8]
 \tikzset{man/.style = {shape=circle,draw,inner sep=0pt}}
	  \tikzset{woman/.style = {shape=circle,draw,inner sep=0pt}}
	  \tikzset{edge/.style = {solid,line width=1.5pt}}
\node[man]   (1)  at (0,3.2)   {$m_1$};
\node[man]   (2)  at (0,2.4)   {$m_2$};
\node[man]   (3)  at (0,1.6)   {$m_3$};
\node[man]   (4)  at (0,0.8)   {$m_4$};
\node[man]   (5)  at (0,0)   {$m_5$};
\node[woman] (6)  at (2.5,3.2) {$w_1$};
\node[woman] (7)  at (2.5,2.4) {$w_2$};
\node[woman] (8)  at (2.5,1.6) {$w_3$};
\node[woman] (9)  at (2.5,0.8) {$w_4$};
\node[woman] (10) at (2.5,0) {$w_5$};
\draw[edge] (1) to (6);
\draw[edge] (2) to (8);
\draw[edge] (3) to (7);
\draw[edge] (4) to (10);
\draw[edge] (5) to (9);
\end{tikzpicture}
}
\hspace{0.15in}
\subfloat[]{
	\begin{tikzpicture}[scale=0.8]
	\tikzset{man/.style = {shape=circle,draw,inner sep=0pt}}
	\tikzset{woman/.style = {shape=circle,draw,inner sep=0pt}}
	\tikzset{edge/.style = {solid,line width=1.5pt}}
	\tikzset{dedge/.style = {dashed,line width=1.5pt}}
	\node[man]   (1)  at (0,3.2)   {$m_1$};
	\node[man]   (2)  at (0,2.4)   {$m_2$};
	\node[man]   (3)  at (0,1.6)   {$m_3$};
	\node[man]   (4)  at (0,0.8)   {$m_4$};
	\node[man]   (5)  at (0,0)   {$m_5$};
	\node[woman] (6)  at (2.5,3.2) {$w_1$};
	\node[woman] (7)  at (2.5,2.4) {$w_2$};
	\node[woman] (8)  at (2.5,1.6) {$w_3$};
	\node[woman] (9)  at (2.5,0.8) {$w_4$};
	\node[woman] (10) at (2.5,0) {$w_5$};
	\draw[dedge] (1) to (8);
	\draw[edge] (2) to (6);
	\draw[edge] (3) to (7);
	\draw[edge] (4) to (10);
	\draw[edge] (5) to (9);
	\end{tikzpicture}
}
\hspace{0.15in}
\subfloat[]{
	\begin{tikzpicture}[scale=0.8]
	\tikzset{man/.style = {shape=circle,draw,inner sep=0pt}}
	\tikzset{woman/.style = {shape=circle,draw,inner sep=0pt}}
	\tikzset{edge/.style = {solid,line width=1.5pt}}
	\tikzset{dedge/.style = {dashed,line width=1.5pt}}
	\node[man]   (1)  at (0,3.2)   {$m_1$};
	\node[man]   (2)  at (0,2.4)   {$m_2$};
	\node[man]   (3)  at (0,1.6)   {$m_3$};
	\node[man]   (4)  at (0,0.8)   {$m_4$};
	\node[man]   (5)  at (0,0)   {$m_5$};
	\node[woman] (6)  at (2.5,3.2) {$w_1$};
	\node[woman] (7)  at (2.5,2.4) {$w_2$};
	\node[woman] (8)  at (2.5,1.6) {$w_3$};
	\node[woman] (9)  at (2.5,0.8) {$w_4$};
	\node[woman] (10) at (2.5,0) {$w_5$};
	\draw[dedge] (1) to (10);
	\draw[edge] (2) to (8);
	\draw[edge] (3) to (7);
	\draw[edge] (4) to (6);
	\draw[edge] (5) to (9);
	\end{tikzpicture}
}
\hspace{0.15in}
\subfloat[]{
	\begin{tikzpicture}[scale=0.8]
	\tikzset{man/.style = {shape=circle,draw,inner sep=0pt}}
	\tikzset{woman/.style = {shape=circle,draw,inner sep=0pt}}
	\tikzset{edge/.style = {solid,line width=1.5pt}}
	\tikzset{dedge/.style = {dashed,line width=1.5pt}}
	\node[man]   (1)  at (0,3.2)   {$m_1$};
	\node[man]   (2)  at (0,2.4)   {$m_2$};
	\node[man]   (3)  at (0,1.6)   {$m_3$};
	\node[man]   (4)  at (0,0.8)   {$m_4$};
	\node[man]   (5)  at (0,0)   {$m_5$};
	\node[woman] (6)  at (2.5,3.2) {$w_1$};
	\node[woman] (7)  at (2.5,2.4) {$w_2$};
	\node[woman] (8)  at (2.5,1.6) {$w_3$};
	\node[woman] (9)  at (2.5,0.8) {$w_4$};
	\node[woman] (10) at (2.5,0) {$w_5$};
	\draw[dedge] (1) to (9);
	\draw[edge] (2) to (8);
	\draw[edge] (3) to (7);
	\draw[edge] (4) to (10);
	\draw[edge] (5) to (6);
	\end{tikzpicture}
}
\hspace{0.25in}
	\caption{Subfigure (a) illustrates the graph representation of the \SMC{} instance $\I_n$ described in the proof of \Cref{thm:LowerBoundSupportSize} for $n=5$. Subfigure (b) shows the matching $\mu^\text{opt}$. Subfigures (c), (d), and (e) show the matchings $\mu^{(1)}$, $\mu^{(2)}$, and $\mu^{(3)}$, respectively. Dashed lines indicate zero-valuation pairs that do not appear in the graph representation.\label{fig:LowerBoundSupportSize}}        
\end{figure}
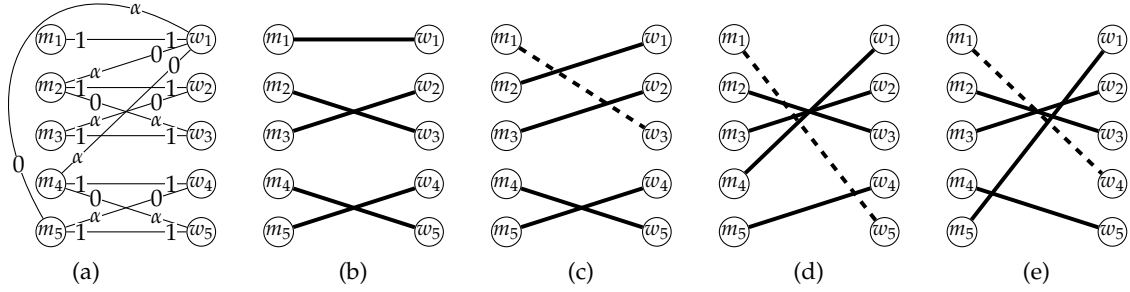

Define $\textstyle{ \mu^\text{opt} \coloneqq \{(m_1,w_1)\} \cup \left\{(m_{2i},w_{2i+1}), (m_{2i+1},w_{2i}): i \in \{1,2,\dots,\frac{n-1}{2}\} \right\} }$ (see Figure~\ref{fig:LowerBoundSupportSize}b). We also define a number of other integral matchings obtained by modifying $\mu^\text{opt}$, as follows: For all $i \in \{1,2,\dots,\frac{n-1}{2}\}$, the matching $\mu^{(i)}$ (see Figures~\ref{fig:LowerBoundSupportSize}c and~\ref{fig:LowerBoundSupportSize}d) is the integral matching which is obtained from ${\textstyle{ \mu^\text{opt}}}$ by replacing $\{(m_1,w_1), (m_{2i},w_{2i+1})\}$ with $\{(m_1,w_{2i+1}),(m_{2i},w_1)\}$, i.e.,
\begin{align*}
\mu^{(i)} \coloneqq \big\{ (m_1,w_{2i+1}) & \cup (m_{2i},w_1) \cup (m_{2i+1},w_{2i}) \cup \left\{ (m_{2\ell},w_{2\ell+1}) \cup (m_{2\ell+1},w_{2\ell}) \right\}_{\ell \in \{1,2,\dots,\frac{n-1}{2}\} \setminus\{i\}} \big\}.
\end{align*}

Also, the matching $\textstyle{ \mu^{\left( \frac{n+1}{2} \right)} }$  (see Figure~\ref{fig:LowerBoundSupportSize}e) is the integral matching obtained from $\mu^\text{opt}$ by replacing $\{(m_1,w_1),(m_{n},w_{n-1})\}$ with $\{(m_1,w_{n-1}),(m_{n},w_1)\}$, i.e.,
\begin{align*}
\mu^{ \left( \frac{n+1}{2} \right)} \coloneqq \big\{ (m_1,w_{n-1}) & \cup (m_{n-1},w_n)  \cup (m_{n},w_{1}) \cup \left\{ (m_{2\ell},w_{2\ell+1}) \cup (m_{2\ell+1},w_{2\ell}) \right\}_{\ell \in \{1,2,\dots,\frac{n-3}{2}\}} \big\}.
\end{align*}

Now, consider the fractional matching $\textstyle{ \mu \coloneqq \sum_{i=1}^{\frac{n+1}{2}} \frac{1}{\alpha} \mu^{(i)} + \left(1 - \frac{n+1}{2\alpha} \right) \mu^\text{opt} }$. Since $\alpha> n+2>\frac{n+1}{2}$, $\mu$ is well-defined and has the matchings $\mu^\text{opt}$ and $\mu^{(i)}$ for all $i \in \left\{1, \dots, \frac{n+1}{2}\right\}$ in its support. Notice that $\W(\mu^\text{opt}) = (n-1)\alpha+2$ and $\W(\mu^{(i)}) = (n-1)\alpha$ for all $\textstyle{ i \in \{1,2,\dots,\frac{n+1}{2}\} }$. Thus, $\textstyle{ \W(\mu) > (n-1)\alpha}$.

It can be verified that $\mu$ is stable. Indeed, we only need to check the blocking condition for the pairs $(m_i,w_i)$ with ${i \in \{1,2,\dots,n\}}$. We have that $v_{w_1}(\mu) \geq 1$ (since $\textstyle{ \mu^{(\frac{n+1}{2})} }$ has weight $\textstyle{ \frac{1}{\alpha} }$ in $\mu$ and $V(m_n,w_1) = \alpha$), $u_{m_{2i}}(\mu) \geq 1$ for each $\textstyle{ i \in \{1,2,\dots,\frac{n-1}{2}\}}$ (since $\textstyle{ \mu^{(i)} }$ has weight $\textstyle{ \frac{1}{\alpha} }$ in $\mu$ and $U(m_{2i},w_1) = \alpha$), and $v_{w_{2i+1}}(\mu) \geq 1$ for each $\textstyle{ i \in \{1,2,\dots,\frac{n-1}{2}\}}$ (since $\textstyle{ \mu^{(i)} }$ has weight $\textstyle{ \frac{1}{\alpha} }$ in $\mu$ and $V(m_{2i},w_{2i+1}) = \alpha$). The welfare of the optimal stable fractional matching must therefore be at least $\W(\mu)$, and thus strictly greater than $(n-1)\alpha$.

We now claim that any $\rho$-efficient stable fractional matching $\mu'$ satisfies $\mu'(m_n,w_1)>0$. Indeed, assuming otherwise that $\mu'(m_n,w_1)=0$, the only pair that can give positive utility to man $m_1$ and woman $w_1$ is $(m_1,w_1)$. Hence, we must also have $\mu'(m_1,w_1)=1$, and, as a result, $\mu'(m_{2i},w_1)=0$ for all $i\in \left\{1, 2, \dots, \frac{n-1}{2}\right\}$. Then, the only pair that can give positive utility to man $m_{2i}$ and woman $w_{2i}$ is $(m_{2i},w_{2i})$, and hence, it must also be that $\mu'(m_{2i},w_{2i})=1$. Consequently, the only pair that can give positive utility to man $m_{2i+1}$ and woman $w_{2i+1}$ is $(m_{2i+1},w_{2i+1})$ and, hence, we must have $\mu'(m_{2i+1},w_{2i+1})=1$. The welfare of matching $\mu'$ would then be $2n$, which is less than $\rho(n-1)\alpha$ by the assumed bound on $\alpha$. In other words, the welfare of $\mu'$ would be less than $\rho$ times the welfare of the stable fractional matching $\mu$, contradicting the assumption that $\mu'$ is $\rho$-efficient.

The final step in the proof involves showing that for any stable fractional matching $\mu'$ with support of size at most $\textstyle{ \frac{n-1}{2}\rho }$, we must have $\W(\mu') < \rho(n-1)\alpha$; the desired bound on the support size would then follow from the contrapositive. Let $T\coloneqq \left\{i\in \left\{1, 2, \dots \frac{n-1}{2}\right\}: \mu'(m_{2i},w_1)>0\right\}$, 
%$$T\coloneqq \left\{i\in \left\{1, 2, \dots \frac{n-1}{2}\right\}: \mu'(m_{2i},w_1)>0\right\}$$
and $\overline{T}\coloneqq \left\{1, 2, \dots \frac{n-1}{2}\right\} \setminus T$. Since $\mu'$ has support of size at most $\textstyle{ \frac{n-1}{2}\rho }$ and $\mu'(m_n,w_1)>0$, it holds that $\textstyle{ |T|\leq  \frac{n-1}{2}\rho - 1}$.

For every $i\in T$, the agents $m_{2i}$, $w_{2i}$, $m_{2i+1}$, and $w_{2i+1}$ can together contribute at most $2\alpha$ to the welfare. On the other hand, when $i\in \overline{T}$, we have $\mu'(m_{2i},w_1)=0$, and the only pair that can give positive utility to man $m_{2i}$ and woman $w_{2i}$ is $(m_{2i},w_{2i})$. Therefore, we must have that $\mu'(m_{2i},w_{2i})=1$. Consequently, the only pair that can give positive utility to man $m_{2i+1}$ and woman $w_{2i+1}$ is $(m_{2i+1},w_{2i+1})$, and it follows that $\mu'(m_{2i+1},w_{2i+1})=1$. Therefore, when $i\in \overline{T}$, the agents $m_{2i}$, $w_{2i}$, $m_{2i+1}$, and $w_{2i+1}$ can together contribute at most $4$ to the welfare. Taking the possible contribution of pair $(m_1,w_1)$ into account, we have that
\begin{align*}
\W(\mu') & \leq 2+2\alpha|T|+4|\overline{T}| = 2n+(2\alpha-4)|T| \\
& \leq 2n+\rho (n-1)\alpha-2\alpha-2(n-1)\rho +4 < \rho (n-1)\alpha.
\end{align*}
The equality follows from the definition of $\overline{T}$, the second inequality follows from the bound on $|T|$ above, and the third one from the definition of $\alpha$. This completes the proof of \Cref{thm:LowerBoundSupportSize}.
\end{proof}

\cref{thm:LowerBoundSupportSize} has an interesting algorithmic implication. The fact that the support size can be large even for approximately optimal solutions suggests that a \emph{support enumeration} strategy---which has proven useful in other economic problems~\citep{LMM03playing,Bar18}---will be strictly more demanding than any of the MILP-based approaches in \Cref{subsec:OptimalStableMatchingProgram}. Specifically, a brute force enumeration of all subsets of $\Omega(\rho n)$ matchings requires $\Omega((n!)^{\rho n}) \approx \Omega(n^{\rho n^2})$ time. By contrast, the MILPs \ref{Prog:OPT-Stab} and \ref{LP:OPT-Thresh} described in \Cref{subsec:OptimalStableMatchingProgram} require $\O(2^{n^2})$ and $\O(n^n)$ time, respectively. 
A similar comparison can be made for optimal $\eps$-stable fractional matchings. Interestingly, as we will show in \Cref{sec:TractabilityResults}, an $\eps$-stable fractional matching of \emph{small} support can be easily computed, and provides nearly the best approximation ratios achievable by efficient algorithms (under standard complexity-theoretic assumptions).
\section{Algorithmic results}
\label{sec:TractabilityResults}
We begin the discussion of our algorithmic results with binary valuations. In this setting, \OptStab{} reduces to computing a maximum weight matching on a specific weighted graph associated with the given instance.

\begin{theorem}
	\label{thm:MaxWelfareBinaryVals}
	Given an \SMC{} instance $\I = \langle M,W,U,V \rangle$ with binary valuations, an optimal stable fractional matching for $\I$ is, without loss of generality, integral, and can be computed in polynomial time.
%	Given an \SMC{} instance $\I = \langle M,W,U,V \rangle$ with binary valuations, an optimal stable fractional matching for $\I$ can be computed in polynomial time.
\end{theorem}

\begin{proof}
Let $G$ be the graph representation of $\I$. We assign to each edge $(m,w)$ in $G$ a weight $\gamma(m,w)$, as follows:
\[ \gamma(m,w) =
  \begin{cases}
    2+\frac{1}{n^2} & \text{if } U(m,w)=V(m,w)=1,\\
    1 & \text{if } U(m,w)=1, V(m,w)=0 \text{ or } U(m,w)=0, V(m,w)=1,\\
    0 & \text{otherwise}.
  \end{cases}
\]
% $\gamma(m,w)=2+1/n^2$ whenever $U(m,w)=V(m,w)=1$, otherwise $\gamma(m,w)=1$.
 Thus, for any integral matching $\mu$ in $G$, if $n_\mu$ denotes the number of agents (men and women) with utility~$1$ in the \SMC{} instance $\I$, then $n_\mu \leq \sum_{(m,w) \in \mu} \gamma(m,w) < n_\mu + 1$.

Let $\mu$ be a maximum weight matching in $G$. Note that $\mu$ can be computed in polynomial time and is, without loss of generality, integral. Also, it follows from the above inequality that $\mu$ is an optimal matching for $\I$. We will now argue that $\mu$ is stable. Indeed, assuming otherwise, any blocking pair $(m,w)$ must have $0 = u_m(\mu) < U(m,w) = 1$, $0 = v_w(\mu) < V(m,w) = 1$ and $(m,w) \notin \mu$. Thus, if $\mu$ contains one or both of the edges $(m,w')$ and $(m',w)$ for some $w' \neq w$ and $m'\neq m$, then we must have that $U(m,w')=0$ and/or $V(m',w)=0$. By our definition of weights, this would imply that $\gamma(m,w')=1$ and/or $\gamma(m',w)=1$. We can now replace one or both of these edges with the edge $(m,w)$, which has weight $\gamma(m,w)=2+1/n^2$, and obtain a new matching with strictly larger weight---a contradiction.
\end{proof}
The proof of \cref{thm:MaxWelfareBinaryVals} can be easily extended to the non-bipartite setting (also known as the \emph{stable roommates} problem) to compute an optimal stable fractional matching with binary valuations.

Next, we consider general valuations and show how to exploit stable integral matchings to get an approximate solution for \OptStab. Let $\sigma_{\max}$ and $\sigma_{\min}$ denote the largest and the smallest non-zero valuation among all agents in $\I$, respectively. We call a man-woman pair $(m,w)$ {\em light} if either $U(m,w)=0$ or $V(m,w)=0$, and {\em heavy} otherwise. Given an \SMC{} instance $\I$ as input, our algorithm computes a stable integral matching for $\I$, say $\mu$, in two steps: First, it computes a stable integral matching $\mu_1$ using only the heavy pairs (and taking into account the stability constraints in heavy pairs only). Then, it \emph{completes} the solution with a matching $\mu_2$ of maximum welfare using the light pairs subject to feasibility constraints, i.e., using light pairs that do not share any agents with the pairs in $\mu_1$. The light pairs impose no additional constraints on stability, so the resulting matching is stable.

We will show that $\mu$ has approximation ratio $1+\sigma_{\max}/\sigma_{\min}$ (recall that $\sigma_{\max}$ and $\sigma_{\min}$ are the largest and the smallest non-zero valuations, respectively, in the instance $\I$). Let $\mu^\text{opt}$ be an optimal matching for $\I$. Also, let $\mu^\text{opt}_1$ be the set of pairs of $\mu^\text{opt}$ that share an agent with some pair of $\mu_1$, i.e., $\mu^\text{opt}_1 \coloneqq \{(m,w) \in \mu^\text{opt}: \text{at least one of  } m \text{ or } w \text{ is matched in } \mu_1\}$. By definition of $\mu_2$, we have $\W(\mu_2)\geq \W(\mu^\text{opt}\setminus \mu^\text{opt}_1)$. To complete the proof, we will need the following lemma.

\begin{lemma}\label{lem:wit}
	$\W(\mu_1\setminus \mu^\text{opt}_1)  \geq \left(1+\sigma_{\max}/\sigma_{\min}\right)^{-1} \W(\mu^\text{opt}_1\setminus \mu_1)$.
\end{lemma}

\begin{proof}
Our proof constructs a mapping in which every pair $(m,w)\in \mu^\text{opt}_1\setminus \mu_1$ is mapped to one of its agents, whom we will call the {\em witness} of the pair. The mapping is such that the utility of the witness in the matching $\mu_1\setminus \mu^\text{opt}_1$ is at least $\left(1+\sigma_{\max}/\sigma_{\min}\right)^{-1}\left(U(m,w)+V(m,w)\right)$. Note that once we establish the said mapping, the proof will follow, since each agent can be the witness of at most one pair of $\mu^\text{opt}_1\setminus \mu_1$ and $\W(\mu_1\setminus \mu^\text{opt}_1)$ is at least the total utility of the witnesses in $\mu_1\setminus \mu^\text{opt}_1$.

Consider a light pair $(m,w) \in \mu^\text{opt}_1\setminus \mu_1$. The witness is an agent ($m$ or $w$) who also belongs to a pair of $\mu_1\setminus \mu^\text{opt}_1$; such an agent certainly exists by the definition of $\mu^\text{opt}_1$. Since all pairs of $\mu_1\setminus \mu^\text{opt}_1$ are heavy, the utility of the witness of $(m,w)$ in $\mu_1\setminus \mu^\text{opt}_1$ is at least $\sigma_{\min} = \frac{\sigma_{\min}}{\sigma_{\max}} (0+\sigma_{\max})\geq \left(1+\sigma_{\max}/\sigma_{\min}\right)^{-1}\left(U(m,w)+V(m,w)\right)$, since $(m,w)$ is light.

Now consider a heavy pair $(m,w) \in \mu^\text{opt}_1\setminus \mu_1$. If $\mu_1$ contains a pair $(m,w')$ with $U(m,w')\geq U(m,w)$, select agent $m$ to be the witness, otherwise select agent $w$. Note that in the latter case, stability of $\mu_1$ implies the existence of $(m',w) \in \mu_1$ such that $V(m',w)\geq V(m,w)$. Hence, the utility of the witness of $(m,w)$ in $\mu_1\setminus \mu^\text{opt}_1$ is at least $\min\{U(m,w),V(m,w)\}$, which, in turn, is at least $\left(1+\sigma_{\max}/\sigma_{\min}\right)^{-1}\left(U(m,w)+V(m,w)\right)$.
%$\min\{U(m,w),V(m,w)\} \geq \left(1+\sigma_{\max}/\sigma_{\min}\right)^{-1}\left(U(m,w)+V(m,w)\right)$.
\end{proof}

Now, \cref{lem:wit} gives the desired approximation ratio, as follows:
\begin{align*}
\W(\mu) & = \W(\mu_1)+\W(\mu_2)= \W(\mu_1\setminus \mu^\text{opt}_1)+\W(\mu_1\cap \mu^\text{opt}_1)+\W(\mu_2)\\
& \geq \left(1+\sigma_{\max}/\sigma_{\min}\right)^{-1} \W(\mu^\text{opt}_1\setminus \mu_1)+\W(\mu^\text{opt}_1\cap \mu_1)+\W(\mu^\text{opt}\setminus \mu^\text{opt}_1)\\
& \geq \left(1+\sigma_{\max}/\sigma_{\min}\right)^{-1} \W(\mu^\text{opt}).
\end{align*}

For ternary valuations in $\{0,1,\alpha\}$, the above algorithm gives a $(1+\alpha)$-approximation. An improved approximation for ternary valuations can be achieved using the following modification: When computing the stable integral matching, resolve ties %in the ordinal preferences of every agent
in favour of the pairs $(m,w)$ with the highest $U(m,w)+V(m,w)$. The next lemma establishes an improved approximation ratio of $\max\{2,\alpha\}$.

\begin{lemma}
The modified algorithm for \SMC{} instances with ternary valuations in $\{0,1,\alpha\}$ satisfies 	$\W(\mu_1\setminus \mu^\text{opt}_1)  \geq \min\{\frac{1}{2},\frac{1}{\alpha}\} \W(\mu^\text{opt}_1\setminus \mu_1)$.
\end{lemma}

\begin{proof}
	For a pair $(m,w)$ of matching $\mu^\text{opt}_1\setminus \mu_1$, we use the term {\em neighborhood} to refer to the pairs of $\mu_1\setminus \mu^\text{opt}_1$ that use agent $m$ or $w$. We will show that the total utility from pairs in the neighborhood of $(m,w)$ is at least $\min\{1,2/\alpha\} \left(U(m,w)+V(m,w)\right)$. Since each pair of $\mu_1\setminus \mu^\text{opt}_1$ can be in the neighborhood of at most two pairs of $\mu^\text{opt}_1\setminus \mu_1$, this will give us the desired inequality.
	
	Indeed, by the particular way we resolve ties in the ordinal preferences before computing the matching $\mu_1$, a heavy pair $(m,w)$ in $\mu^\text{opt}_1\setminus \mu_1$ must have a pair of utility at least $U(m,w)+V(m,w)$ in its neighborhood. A light pair $(m,w)$ has $U(m,w)+V(m,w)\leq \alpha$ and certainly has a heavy pair of utility at least $2$ in its neighborhood.
\end{proof}

The above discussion is summarized in the following statement.
\begin{theorem}\label{thm:alg-summary}
There is a polynomial-time algorithm which, given an \SMC{} instance $\I$ with an optimal matching $\mu^\text{opt}$, computes a stable integral matching $\mu$ with $\W(\mu)\geq \min\{\frac{1}{2},\frac{1}{\alpha}\} \W(\mu^\text{opt})$ if $\I$ has ternary valuations in $\{0,1,\alpha\}$, and $\W(\mu)\geq \left(1+\sigma_{\max}/\sigma_{\min}\right)^{-1}\W(\mu^\text{opt})$ in general, where $\sigma_{\max}$ and $\sigma_{\min}$ denote the highest and lowest non-zero valuation in $\I$, respectively.
\end{theorem}

We conclude this section by considering approximate stability. For general valuations, we present a polynomial-time $1/\eps$-approximation algorithm for \OptEpsStab, which constructs an $\eps$-stable fractional matching with a small support by combining an optimal matching with a stable integral matching.

\begin{theorem}
 \label{thm:ApproxStable+ApproxOptimal}
 There is a polynomial-time algorithm that given any \SMC{} instance $\I = \langle M,W,U,V \rangle$ and any rational $\eps \in [0,1]$, computes a fractional matching $\mu$ that is $\eps$-stable for $\I$ such that $\W(\mu) \geq \eps \W(\mu^\text{opt})$, where $\mu^\text{opt}$ is an optimal matching for $\I$.
\end{theorem}

\begin{proof}
Let $\mu^\textup{s}$ be any stable integral matching and $\mu^\text{opt}$ be an optimal matching for $\I$. Note that both $\mu^\textup{s}$ and $\mu^\text{opt}$ can be computed in polynomial time. We will show that $\mu \coloneqq (1-\eps) \mu^\textup{s} + \eps \mu^\text{opt}$ satisfies the desired properties. Indeed, $\W(\mu) = (1-\eps) \W(\mu^\textup{s}) + \eps \W(\mu^\text{opt}) \geq \eps \W(\mu^\text{opt})$. Furthermore, since $\mu^\textup{s}$ is stable, we have that for any man-woman pair $(m,w) \in M \times W$, either $u_m(\mu^\textup{s}) \geq U(m,w)$ or $v_w(\mu^\textup{s}) \geq V(m,w)$. The former condition implies that $u_m(\mu) \geq (1-\eps) u_m(\mu^\textup{s}) \geq (1-\eps) U(m,w)$, while the latter condition gives $v_w(\mu) \geq (1-\eps) V(m,w)$. Either way, the pair $(m,w)$ is $\eps$-stable.
\end{proof}

In particular, \Cref{thm:ApproxStable+ApproxOptimal} shows that a $\frac{1}{2}$-stable fractional matching with welfare at least half of that of an optimal fractional matching (and therefore, that of an optimal stable fractional matching) can be computed in polynomial time. In \cref{app:Half-stability}, we provide a slightly stronger welfare guarantee: There is a polynomial-time algorithm that computes a $\frac{1}{2}$-stable fractional matching with welfare at least that of an optimal (exactly) stable fractional matching. 
\section{Hardness of approximation}
\label{sec:IntractabilityResults}
In this section, we present our inapproximability statements, which are by far the technically most involved results in the paper. We present polynomial-time reductions which, given a 3SAT formula $\phi$ of a particular structure, construct \SMC{} instances that simulate the evaluation of $\phi$ for every variable assignment. The constructed \SMC{} instances consist of several gadgets including an {\em accumulator}. The simulation of the evaluation of $\phi$ by the \SMC{} instance is such that:
\begin{itemize}
	\item [(a)] when $\phi$ has a satisfying assignment, there is a stable (or $\eps$-stable) fractional matching where the contribution of the agents in the accumulator gadget to the welfare can be large and dominates the contribution from the remaining \SMC{} instance, and
	\item [(b)] when $\phi$ is not satisfiable, the contribution of the accumulator and, subsequently, the total welfare of any stable (or $\eps$-stable) fractional matching is very small.
\end{itemize}
Hence, distinguishing between \SMC{} instances with stable (or $\eps$-stable) fractional matchings of very high and very low welfare would allow us to decide 3SAT. We have two inapproximability statements: \cref{thm:strong-inapprox} for \OptStab\ and \cref{thm:strong-eps-inapprox} for \OptEpsStab.

\begin{theorem}\label{thm:strong-inapprox}
	For every constant $\delta>0$, it is NP-hard to approximate \OptStab{} for \SMC{} instances with ternary valuations in $\{0,1,\alpha\}$to within a factor of (i) $\alpha-1/2-\delta$ if $\alpha=\bigO(n)$, and (ii) $\Omega(n^{1-\delta})$ otherwise.
\end{theorem}

\begin{theorem}\label{thm:strong-eps-inapprox}
	For any constant $\eps \in (0,0.03]$ and $\delta>0$, it is NP-hard to approximate \OptEpsStab\ to within a factor of $1/\eps - \delta$.
\end{theorem}

We will prove \cref{thm:strong-inapprox} here; the proof of \cref{thm:strong-eps-inapprox}, which uses similar gadgets but is slightly more involved, appears in \cref{app:eps-stability}. Since the proof is long, we have divided it into three parts: the description of the reduction (\cref{sec:reduction}), technical claims with gadget properties (\cref{sec:gadgets}), and the proof of the inapproximability result (\cref{sec:proof}).

\subsection{The reduction}\label{sec:reduction}
In particular, we present a polynomial-time reduction from 2P2N-3SAT, the special case of 3SAT consisting of 3-CNF clauses in which every variable appears four times: twice as a positive literal and twice as a negative one. 2P2N-3SAT is known to be NP-hard~\citep{Yos05}. Our reduction takes as input an instance of 2P2N-3SAT consisting of $N$ (boolean) variables $x_1, x_2, \dots, x_N$, and a 3-CNF formula $\phi$ with $L=4N/3$ clauses $c_1,c_2,\dots,c_{L}$. Without loss of generality, we assume that each clause in $\phi$ consists of distinct literals.

Given the instance of 2P2N-3SAT, our reduction generates an instance $\I=\langle M,W,U,V \rangle$ of  \OptStab. As usual, we denote by $n$ the number of men (or women) in $\I$. We will use a positive integer parameter $k$ which will determine the size of $n$; in particular, $n=\bigO(N+k)$. We define $\I$ by referring to its graph representation, which consists of {\em variable gadgets}, {\em clause gadgets}, {\em variable-clause connectors}, an {\em accumulator}, and {\em clause-accumulator connectors}. For each gadget, we classify the edges (i.e., man-woman pairs and their valuations) into the following three types:
%we define its nodes/agents (men and women) and agent values for the edges, which belong to the following three types:
\begin{itemize}
	\item {\em man-heavy} edges $(m,w)$ with $U(m,w)=\alpha$ and $V(m,w)=0$,
	\item {\em woman-heavy} edges $(m,w)$ with $U(m,w)=0$ and $V(m,w)=\alpha$, and
	\item {\em balanced} edges $(m,w)$ with $U(m,w)=V(m,w)=1$.
\end{itemize}
Recall that any pair $(m,w)$ that does not appear as an edge in the graph representation has $U(m,w)=V(m,w)=0$.

The instance $\I$ has a variable gadget for every variable $x$, which consists of five men $m^{x}_1$, $m^{x}_2$, $e^x_1$, $e^x_2$, $e^x_3$, four women $w^{\overline{x}}_1$, $w^{\overline{x}}_2$, $f^x_1$, $f^x_2$ and the ten balanced edges
$(e^x_1,f^x_1)$, $(m^x_1,f^x_1)$, $(m^x_1,w^{\overline{x}}_1)$, $(e^x_3,w^{\overline{x}}_1)$, $(e^x_3,w^{\overline{x}}_2)$, $(m^x_2,w^{\overline{x}}_2)$,  $(m^x_2,f^x_2)$, $(e^x_2,f^x_2)$, $(m^x_1,w^{\overline{x}}_2)$, and $(m^x_2,w^{\overline{x}}_1)$, as shown in \cref{fig:variable-and-clause-gadget}a.

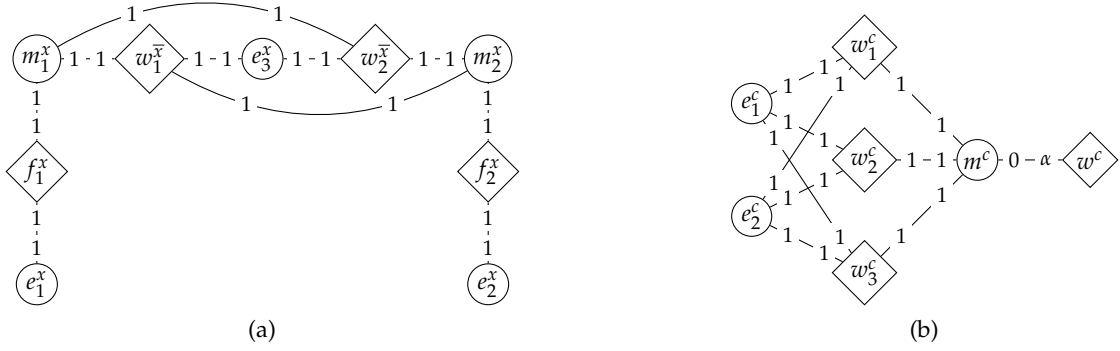
\begin{figure}[t]%[ht]
	\footnotesize
	\centering
\subfloat[]{
	\begin{tikzpicture}[scale=0.75]
	\tikzset{man/.style = {shape=circle,draw,inner sep=1pt}}
	\tikzset{woman/.style = {shape=diamond,draw,inner sep=1pt}}
	\tikzset{edge/.style = {solid}}
	\tikzset{earedge/.style = {densely dotted}}
	% vertices
	\node[man]   (1) at (0,0)  {$e^x_1$};
	\node[woman] (2) at (0,2)  {$f^x_1$};
	\node[man]   (3) at (0,4)  {$m^x_1$};
	\node[woman] (4) at (2,4)  {$w^{\overline{x}}_1$};
	\node[man]   (5) at (4,4)  {$e^x_3$};
	\node[woman] (6) at (6,4)  {$w^{\overline{x}}_2$};
	\node[man]   (7) at (8,4)  {$m^x_2$};
	\node[woman] (8) at (8,2)  {$f^x_2$};
	\node[man]   (9) at (8,0)  {$e^x_2$};
	% edges
	\draw[edge] (1) to node [near start,fill=white,inner sep=2pt] (122) {\scriptsize{$1$}} node [near end,fill=white,inner sep=2pt] (122) {\scriptsize{$1$}} (2);
	\draw[edge] (2) to node [near start,fill=white,inner sep=2pt] (223) {\scriptsize{$1$}} node [near end,fill=white,inner sep=2pt] (223) {\scriptsize{$1$}} (3);
	\draw[edge] (3) to node [near start,fill=white,inner sep=2pt] (324) {\scriptsize{$1$}} node [near end,fill=white,inner sep=2pt] (324) {\scriptsize{$1$}} (4);
	\draw[edge] (4) to node [near start,fill=white,inner sep=2pt] (425) {\scriptsize{$1$}} node [near end,fill=white,inner sep=2pt] (425) {\scriptsize{$1$}} (5);
	\draw[edge] (5) to node [near start,fill=white,inner sep=2pt] (526) {\scriptsize{$1$}} node [near end,fill=white,inner sep=2pt] (526) {\scriptsize{$1$}} (6);
	\draw[edge] (6) to node [near start,fill=white,inner sep=2pt] (627) {\scriptsize{$1$}} node [near end,fill=white,inner sep=2pt] (627) {\scriptsize{$1$}} (7);
	\draw[edge] (7) to node [near start,fill=white,inner sep=2pt] (728) {\scriptsize{$1$}} node [near end,fill=white,inner sep=2pt] (728) {\scriptsize{$1$}} (8);
	\draw[edge] (8) to node [near start,fill=white,inner sep=2pt] (829) {\scriptsize{$1$}} node [near end,fill=white,inner sep=2pt] (829) {\scriptsize{$1$}} (9);
	\draw[edge,bend left] (3) to node [near start,fill=white,inner sep=2pt] (326) {\scriptsize{$1$}} node [near end,fill=white,inner sep=2pt] (326) {\scriptsize{$1$}} (6);
	\draw[edge,bend right] (4) to node [near start,fill=white,inner sep=2pt] (427) {\scriptsize{$1$}} node [near end,fill=white,inner sep=2pt] (427) {\scriptsize{$1$}} (7);
	\end{tikzpicture}
}%\label{fig:variable-gadget}
\hspace{1in}
\subfloat[]{
	\begin{tikzpicture}[scale=0.75]
	\tikzset{man/.style = {shape=circle,draw,inner sep=1pt}}
	\tikzset{woman/.style = {shape=diamond,draw,inner sep=1pt}}
	\tikzset{edge/.style = {solid}}
	\tikzset{earedge/.style = {densely dotted}}
	% vertices
	\node[man]   (1) at (0,1)  {$e^c_2$};
	\node[man]   (2) at (0,3)  {$e^c_1$};
	\node[woman] (3) at (2,0)  {$w^c_3$};
	\node[woman] (4) at (2,2)  {$w^c_2$};
	\node[woman] (5) at (2,4)  {$w^c_1$};
	\node[man]   (6) at (4,2){$m^c$};
	\node[woman] (7) at (6,2)  {$w^c$};
	% edges
	\draw[edge] (1) to node [near start,fill=white,inner sep=2pt] (123) {\scriptsize{$1$}} node [near end,fill=white,inner sep=2pt] (123) {\scriptsize{$1$}} (3);
	\draw[edge] (1) to node [near start,fill=white,inner sep=2pt] (124) {\scriptsize{$1$}} node [near end,fill=white,inner sep=2pt] (124) {\scriptsize{$1$}} (4);
	\draw[edge] (1) to node [very near start,fill=white,inner sep=2pt] (125) {\scriptsize{$1$}} node [very near end,fill=white,inner sep=2pt] (125) {\scriptsize{$1$}} (5);
	\draw[edge] (2) to node [very near start,fill=white,inner sep=2pt] (223) {\scriptsize{$1$}} node [very near end,fill=white,inner sep=2pt] (223) {\scriptsize{$1$}} (3);
	\draw[edge] (2) to node [near start,fill=white,inner sep=2pt] (224) {\scriptsize{$1$}} node [near end,fill=white,inner sep=2pt] (224) {\scriptsize{$1$}} (4);
	\draw[edge] (2) to node [near start,fill=white,inner sep=2pt] (225) {\scriptsize{$1$}} node [near end,fill=white,inner sep=2pt] (225) {\scriptsize{$1$}} (5);
	\draw[edge] (5) to node [near start,fill=white,inner sep=2pt] (526) {\scriptsize{$1$}} node [near end,fill=white,inner sep=2pt] (526) {\scriptsize{$1$}} (6);
	\draw[edge] (4) to node [near start,fill=white,inner sep=2pt] (426) {\scriptsize{$1$}} node [near end,fill=white,inner sep=2pt] (426) {\scriptsize{$1$}} (6);
	\draw[edge] (3) to node [near start,fill=white,inner sep=2pt] (326) {\scriptsize{$1$}} node [near end,fill=white,inner sep=2pt] (326) {\scriptsize{$1$}} (6);
	\draw[edge] (6) to node [near start,fill=white,inner sep=2pt] (627) {\scriptsize{$0$}} node [near end,fill=white,inner sep=2pt] (627) {\scriptsize{$\alpha$}} (7);
	\end{tikzpicture}
}%\label{fig:clause-gadget}
\caption{(a) The variable gadget corresponding to the variable $x$. (b) The clause gadget corresponding to the clause $c$ and its CA-connector $(m^c,w^c)$. %As a convention, we use circles to represent men and diamonds to represent women.
\label{fig:variable-and-clause-gadget}}
\end{figure}

For every clause $c$, instance $\I$ has a clause gadget with three men $m^c$, $e^c_1$, $e^c_2$, three women $w^{c}_1$, $w^c_2$, $w^c_3$, and the nine balanced edges between them, as shown in \cref{fig:variable-and-clause-gadget}b.

For every appearance of a literal in a clause, there is a variable-clause connector (or {\em VC-connector}). The structure of VC-connectors depends on whether they correspond to positive or negative literals and the value of $\alpha$. In each case, we identify one edge of the VC-connector as the {\em input}, and either one or two edges as the {\em output}.

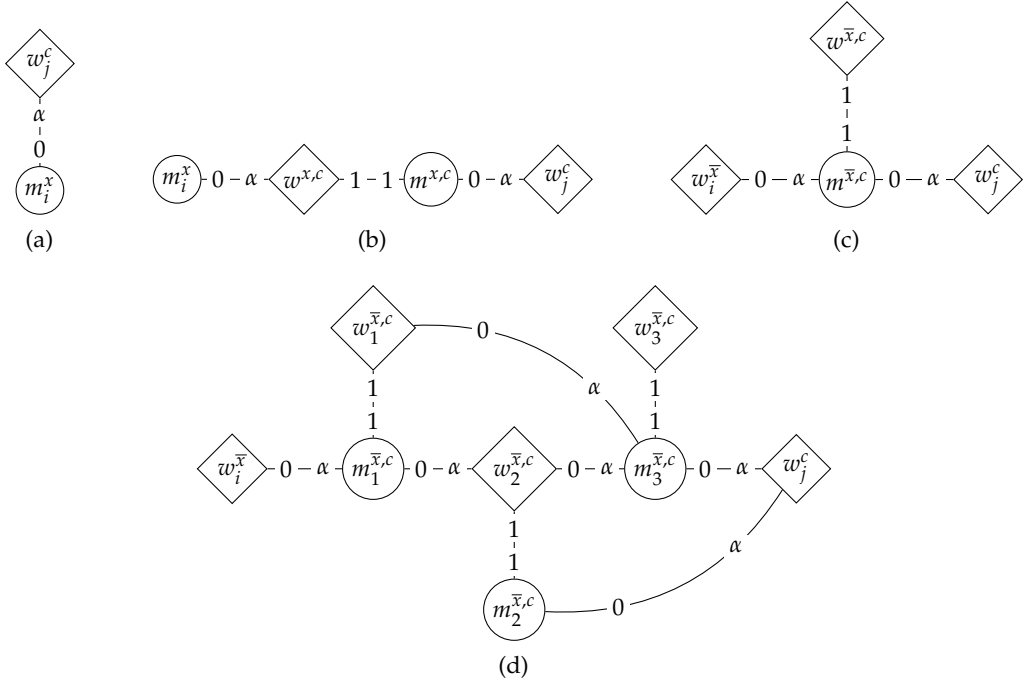
\begin{figure}[t]%[ht]
\footnotesize
\centering
\subfloat[]{
	\begin{tikzpicture}[scale=0.75]
	\tikzset{man/.style = {shape=circle,draw,inner sep=1pt}}
	\tikzset{woman/.style = {shape=diamond,draw,inner sep=1pt}}
	\tikzset{edge/.style = {solid}}
	\tikzset{earedge/.style = {densely dotted}}
	% vertices
	\node[man]   (1) at (0,0)  {$m^x_i$};
	\node[woman] (2) at (0,2.25){$w^c_j$};
	% edges
	\draw[edge] (1) to node [near start,fill=white,inner sep=2pt] (122) {\footnotesize{$0$}} node [near end,fill=white,inner sep=2pt] (122) {\footnotesize{$\alpha$}} (2);
	\end{tikzpicture}
}%\label{subfig:vc-connector_pos_literal_alpha_geq_2}
\hspace{0.3in}
\subfloat[]{
	\begin{tikzpicture}[scale=0.75]
	\tikzset{man/.style = {shape=circle,draw,inner sep=1pt}}
	\tikzset{woman/.style = {shape=diamond,draw,inner sep=1pt}}
	\tikzset{edge/.style = {solid}}
	\tikzset{earedge/.style = {densely dotted}}
	% vertices
	\node[man]   (1) at (0,0)  {$m^x_i$};
	\node[woman] (2) at (2.25,0)  {$w^{x,c}$};
	\node[man]   (3) at (4.5,0)  {$m^{x,c}$};
	\node[woman] (4) at (6.75,0)  {$w^c_j$};
	% edges
	\draw[edge] (1) to node [near start,fill=white,inner sep=2pt] (122) {\footnotesize{$0$}} node [near end,fill=white,inner sep=2pt] (122) {\footnotesize{$\alpha$}} (2);
	\draw[edge] (2) to node [near start,fill=white,inner sep=2pt] (223) {\footnotesize{$1$}} node [near end,fill=white,inner sep=2pt] (223) {\footnotesize{$1$}} (3);
	\draw[edge] (3) to node [near start,fill=white,inner sep=2pt] (324) {\footnotesize{$0$}} node [near end,fill=white,inner sep=2pt] (324) {\footnotesize{$\alpha$}} (4);
	\end{tikzpicture}
}%\label{subfig:vc-connector_pos_literal_alpha_lessthan_2}
\hspace{0.3in}
\subfloat[]{
	\begin{tikzpicture}[scale=0.75]
	\tikzset{man/.style = {shape=circle,draw,inner sep=1pt}}
	\tikzset{woman/.style = {shape=diamond,draw,inner sep=1pt}}
	\tikzset{edge/.style = {solid}}
	\tikzset{earedge/.style = {densely dotted}}
	% vertices
	\node[woman] (1) at (0,0) {$w^{\overline{x}}_i$};
	\node[man] (2) at (2.5,0) {$m^{\overline{x},c}$};
	\node[woman] (3) at (5,0)  {$w^c_j$};
	\node[woman] (4) at (2.5,2.5) {$w^{\overline{x},c}$};
	% edges
	\draw[edge] (1) to node [near start,fill=white,inner sep=2pt] (122) {\footnotesize{$0$}} node [near end,fill=white,inner sep=2pt] (122) {\footnotesize{$\alpha$}} (2);
	\draw[edge] (2) to node [near start,fill=white,inner sep=2pt] (223) {\footnotesize{$0$}} node [near end,fill=white,inner sep=2pt] (223) {\footnotesize{$\alpha$}} (3);
	\draw[edge] (2) to node [near start,fill=white,inner sep=2pt] (224) {\footnotesize{$1$}} node [near end,fill=white,inner sep=2pt] (224) {\footnotesize{$1$}} (4);
	\end{tikzpicture}
}%\label{subfig:vc-connector_neg_literal_alpha_geq_2}
\hspace{0.3in}
\subfloat[]{
	\begin{tikzpicture}[scale=0.75]
	\tikzset{man/.style = {shape=circle,draw,inner sep=1pt}}
	\tikzset{woman/.style = {shape=diamond,draw,inner sep=1pt}}
	\tikzset{edge/.style = {solid}}
	\tikzset{earedge/.style = {densely dotted}}
	% vertices
	\node[woman] (1) at (0,0) {$w^{\overline{x}}_i$};
	\node[man] (2) at (2.5,0) {$m^{\overline{x},c}_1$};
	\node[woman] (3) at (5,0) {$w^{\overline{x},c}_2$};
	\node[man] (4) at (7.5,0) {$m^{\overline{x},c}_3$};
	\node[woman] (5) at (10,0){$w^c_j$};
	\node[woman] (6) at (2.5,2.5) {$w^{\overline{x},c}_1$};
	\node[woman] (7) at (7.5,2.5) {$w^{\overline{x},c}_3$};
	\node[man] (8) at (5,-2.5) {$m^{\overline{x},c}_2$};
	% edges
	\draw[edge] (1) to node [near start,fill=white,inner sep=2pt] (122) {\footnotesize{$0$}} node [near end,fill=white,inner sep=2pt] (122) {\footnotesize{$\alpha$}} (2);
	\draw[edge] (2) to node [near start,fill=white,inner sep=2pt] (223) {\footnotesize{$0$}} node [near end,fill=white,inner sep=2pt] (223) {\footnotesize{$\alpha$}} (3);
	\draw[edge] (3) to node [near start,fill=white,inner sep=2pt] (324) {\footnotesize{$0$}} node [near end,fill=white,inner sep=2pt] (324) {\footnotesize{$\alpha$}} (4);
	\draw[edge] (4) to node [near start,fill=white,inner sep=2pt] (425) {\footnotesize{$0$}} node [near end,fill=white,inner sep=2pt] (425) {\footnotesize{$\alpha$}} (5);
	\draw[edge] (2) to node [near start,fill=white,inner sep=2pt] (226) {\footnotesize{$1$}} node [near end,fill=white,inner sep=2pt] (226) {\footnotesize{$1$}} (6);
	\draw[edge] (3) to node [near start,fill=white,inner sep=2pt] (328) {\footnotesize{$1$}} node [near end,fill=white,inner sep=2pt] (328) {\footnotesize{$1$}} (8);
	\draw[edge] (4) to node [near start,fill=white,inner sep=2pt] (427) {\footnotesize{$1$}} node [near end,fill=white,inner sep=2pt] (427) {\footnotesize{$1$}} (7);
	\draw[edge,bend left] (6) to node [near start,fill=white,inner sep=2pt] (624) {\footnotesize{$0$}} node [near end,fill=white,inner sep=2pt] (624) {\footnotesize{$\alpha$}} (4);
	\draw[edge,bend right] (8) to node [near start,fill=white,inner sep=2pt] (825) {\footnotesize{$0$}} node [near end,fill=white,inner sep=2pt] (825) {\footnotesize{$\alpha$}} (5);
	\end{tikzpicture}
}%\label{subfig:vc-connector_neg_literal_alpha_lessthan_2}
\caption{VC-connectors corresponding to clause $c$ and positive literal $x$ for (a) $\alpha\geq 2$ and (b) $\alpha\in (3/2,2)$, and to clause $c$ and negative literal $\overline{x}$ for (c) $\alpha\geq 2$ and (d) $\alpha\in (3/2,2)$.\label{fig:vc-connector}}
\end{figure}

Specifically, for every positive literal $x$ whose $i$-th appearance ($i\in\{1,2\}$) is as the $j$-th literal ($j\in\{1,2,3\}$) of clause $c$, $\I$ has a VC-connector defined as follows:
\begin{itemize}
	\item When $\alpha\geq 2$, the VC-connector consists of a single woman-heavy edge between $m^x_i$ (from the variable gadget corresponding to variable $x$) and $w^c_j$ (from the clause gadget corresponding to clause $c$), as shown in \cref{fig:vc-connector}a. This edge is simultaneously the input and the output edge of the VC-connector.
	\item When $\alpha\in (3/2,2)$, the VC-connector consists of woman $w^{x,c}$, man $m^{x,c}$, the woman-heavy edges $(m^x_i,w^{x,c})$ and $(m^{x,c},w^c_j)$, and the balanced edge $(m^{x,c},w^{x,c})$, as shown in \cref{fig:vc-connector}b. Here, $(m^x_i,w^{x,c})$ is the input and $(m^{x,c},w^c_j)$ is the output edge.
	\end{itemize}
For every negative literal $\overline{x}$ whose $i$-th appearance ($i\in\{1,2\}$) is as the $j$-th literal ($j\in\{1,2,3\}$) of clause $c$, $\I$ has a VC-connector defined as follows:
\begin{itemize}
	\item When $\alpha\geq 2$, the VC-connector consists of man $m^{\overline{x},c}$, woman $w^{\overline{x},c}$, the man-heavy edge $(m^{\overline{x},c},w^{\overline{x}}_i)$, the balanced edge $(m^{\overline{x},c},w^{\overline{x},c})$, and the woman-heavy edge $(m^{\overline{x},c},w^c_j)$, as shown in \cref{fig:vc-connector}c. Here, $(m^{\overline{x},c},w^{\overline{x}}_i)$ is the input and $(m^{\overline{x},c},w^c_j)$ is the output edge.
	\item When $\alpha\in(3/2,2)$, the VC-connector consists of three men $m^{\overline{x},c}_1$, $m^{\overline{x},c}_2$, $m^{\overline{x},c}_3$, three women $w^{\overline{x},c}_1$, $w^{\overline{x},c}_2$, $w^{\overline{x},c}_3$, the man-heavy edges $(m^{\overline{x},c}_1,w^{\overline{x}}_i)$, $(m^{\overline{x},c}_3,w^{\overline{x},c}_1)$, $(m^{\overline{x},c}_3,w^{\overline{x},c}_2)$, the woman-heavy edges $(m^{\overline{x},c}_1,w^{\overline{x},c}_2)$, $(m^{\overline{x},c}_2,w^{c}_j)$, $(m^{\overline{x},c}_3,w^{c}_j)$, and the balanced edges $(m^{\overline{x},c}_1,w^{\overline{x},c}_1)$, $(m^{\overline{x},c}_2,w^{\overline{x},c}_2)$, $(m^{\overline{x},c}_3,w^{\overline{x},c}_3)$, as shown in \cref{fig:vc-connector}d. In this case, the VC-connector has one input edge $(m^{\overline{x},c}_1,w^{\overline{x}}_i)$ and two output edges $(m^{\overline{x},c}_2,w^{c}_j)$ and $(m^{\overline{x},c}_3,w^{c}_j)$.
\end{itemize}

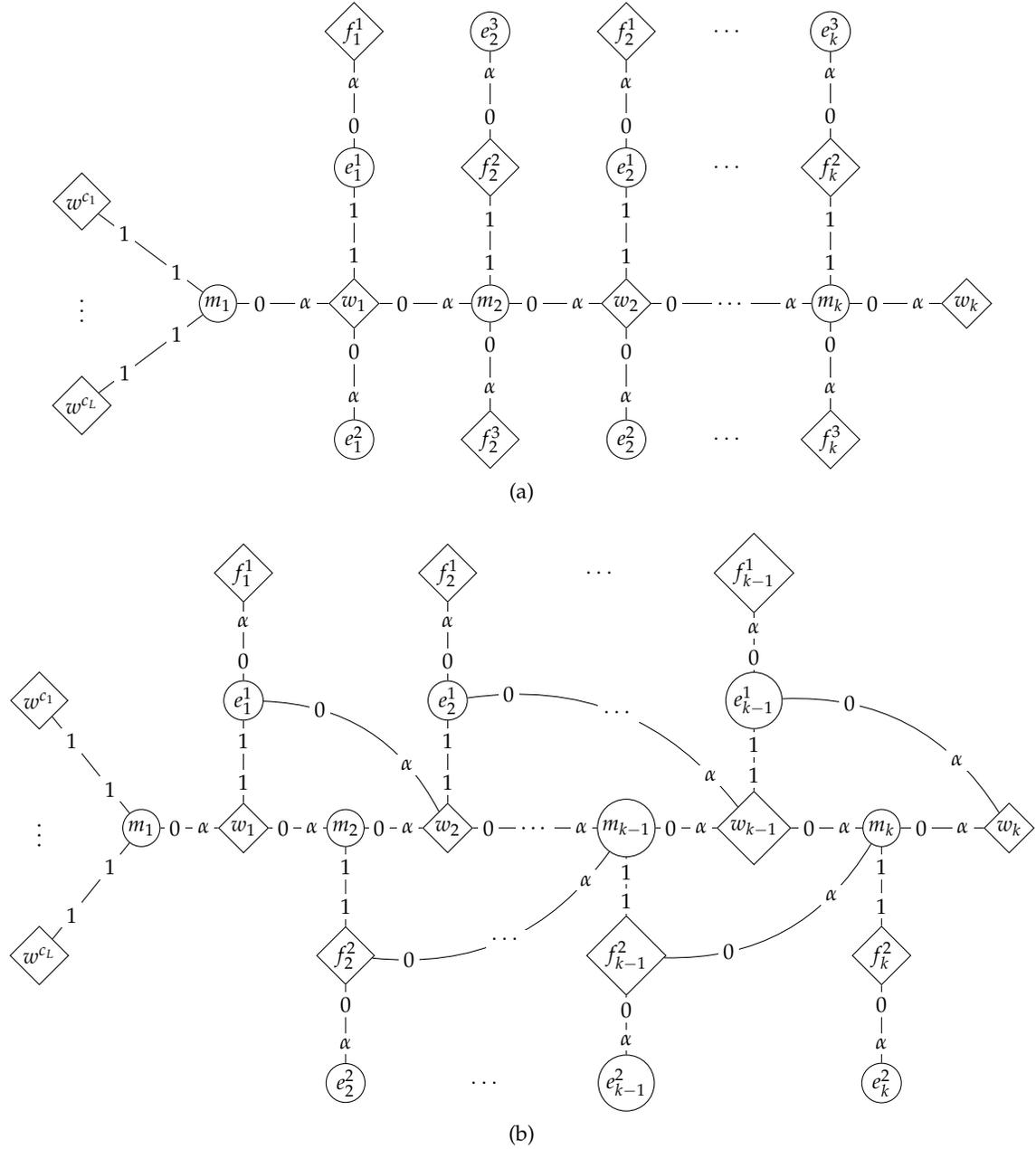
\begin{figure}%[h]%[ht]%[h]
\footnotesize
\centering
\subfloat[]{
	\begin{tikzpicture}%[scale=0.8]
	\tikzset{man/.style = {shape=circle,draw,inner sep=1pt}}
	\tikzset{woman/.style = {shape=diamond,draw,inner sep=1pt}}
	\tikzset{edge/.style = {solid}}
	\tikzset{earedge/.style = {densely dotted}}
	% vertices
	\node[woman] (1) at (0,1.5)  {$w^{c_1}$};
	\node[woman] (2) at (0,-1.5){$w^{c_{L}}$};
	\node at (0,0) {$\vdots$};
	\node[man] (3) at (2,0)  {$m_1$};
	\node[woman] (4) at (4,0)  {$w_1$};
	\node[man] (5) at (6,0)  {$m_2$};
	\node[man] (6) at (4,-2)  {$e_1^2$};
	\node[woman] (7) at (6,-2)  {$f_2^3$};
	\node[woman] (8) at (4,4)  {$f_1^1$};
	\node[man] (9) at (4,2)  {$e_1^1$};
	\node[man] (10) at (6,4)  {$e_2^3$};
	\node[woman] (11) at (6,2)  {$f_2^2$};
	\node[man] (12) at (11,0)  {$m_k$};
	\node[woman] (13) at (13,0)  {$w_k$};
	\node[man] (14) at (11,4)  {$e_k^3$};
	\node[woman] (15) at (11,2)  {$f_k^2$};
	\node[woman] (16) at (11,-2)  {$f_k^3$};
	\node[woman] (17) at (8,4)  {$f_2^1$};
	\node[man] (18) at (8,2)  {$e_2^1$};
	\node[woman] (19) at (8,0)  {$w_2$};
	\node[man] (20) at (8,-2)  {$e_2^2$};
	\node at (9.5,4) {$\cdots$};
	\node at (9.5,2) {$\cdots$};
	\node at (9.5,-2) {$\cdots$};
	% edges
	\draw[edge] (1) to node [near start,fill=white,inner sep=2pt] (123) {\footnotesize{$1$}} node [near end,fill=white,inner sep=2pt] (123) {\footnotesize{$1$}} (3);
	\draw[edge] (2) to node [near start,fill=white,inner sep=2pt] (223) {\footnotesize{$1$}} node [near end,fill=white,inner sep=2pt] (223) {\footnotesize{$1$}} (3);
	\draw[edge] (3) to node [near start,fill=white,inner sep=2pt] (324) {\footnotesize{$0$}} node [near end,fill=white,inner sep=2pt] (324) {\footnotesize{$\alpha$}} (4);
	\draw[edge] (4) to node [near start,fill=white,inner sep=2pt] (425) {\footnotesize{$0$}} node [near end,fill=white,inner sep=2pt] (425) {\footnotesize{$\alpha$}} (5);
	\draw[edge] (4) to node [near start,fill=white,inner sep=2pt] (426) {\footnotesize{$0$}} node [near end,fill=white,inner sep=2pt] (426) {\footnotesize{$\alpha$}} (6);
	\draw[edge] (5) to node [near start,fill=white,inner sep=2pt] (527) {\footnotesize{$0$}} node [near end,fill=white,inner sep=2pt] (527) {\footnotesize{$\alpha$}} (7);
	\draw[edge] (9) to node [near start,fill=white,inner sep=2pt] (928) {\footnotesize{$0$}} node [near end,fill=white,inner sep=2pt] (928) {\footnotesize{$\alpha$}} (8);
	\draw[edge] (9) to node [near start,fill=white,inner sep=2pt] (924) {\footnotesize{$1$}} node [near end,fill=white,inner sep=2pt] (924) {\footnotesize{$1$}} (4);
	\draw[edge] (11) to node [near start,fill=white,inner sep=2pt] (11210) {\footnotesize{$0$}} node [near end,fill=white,inner sep=2pt] (11210) {\footnotesize{$\alpha$}} (10);
	\draw[edge] (5) to node [near start,fill=white,inner sep=2pt] (5211) {\footnotesize{$1$}} node [near end,fill=white,inner sep=2pt] (5211) {\footnotesize{$1$}} (11);
	\draw[edge] (5) to node [near start,fill=white,inner sep=2pt] (5219) {\footnotesize{$0$}} node [near end,fill=white,inner sep=2pt] (5219) {\footnotesize{$\alpha$}} (19);
		\draw[edge] (19) to node [very near start,fill=white,inner sep=2pt] (19212) {\footnotesize{$0$}} node [midway,fill=white,inner sep=2pt] (19212) {$\dots$} node [very near end,fill=white,inner sep=2pt] (19212) {\footnotesize{$\alpha$}} (12);	
	\draw[edge] (12) to node [near start,fill=white,inner sep=2pt] (12213) {\footnotesize{$0$}} node [near end,fill=white,inner sep=2pt] (12213) {\footnotesize{$\alpha$}} (13);
	\draw[edge] (15) to node [near start,fill=white,inner sep=2pt] (15214) {\footnotesize{$0$}} node [near end,fill=white,inner sep=2pt] (15214) {\footnotesize{$\alpha$}} (14);
	\draw[edge] (15) to node [near start,fill=white,inner sep=2pt] (15212) {\footnotesize{$1$}} node [near end,fill=white,inner sep=2pt] (15212) {\footnotesize{$1$}} (12);
	\draw[edge] (12) to node [near start,fill=white,inner sep=2pt] (12216) {\footnotesize{$0$}} node [near end,fill=white,inner sep=2pt] (12216) {\footnotesize{$\alpha$}} (16);
	\draw[edge] (17) to node [near start,fill=white,inner sep=2pt] (17218) {\footnotesize{$\alpha$}} node [near end,fill=white,inner sep=2pt] (17218) {\footnotesize{$0$}} (18);
	\draw[edge] (18) to node [near start,fill=white,inner sep=2pt] (18219) {\footnotesize{$1$}} node [near end,fill=white,inner sep=2pt] (18219) {\footnotesize{$1$}} (19);
	\draw[edge] (19) to node [near start,fill=white,inner sep=2pt] (19220) {\footnotesize{$0$}} node [near end,fill=white,inner sep=2pt] (19220) {\footnotesize{$\alpha$}} (20);
	\end{tikzpicture}
}%\label{subfig:acc_alpha_geq_2}
\hspace{0.1in}
\subfloat[]{
	\begin{tikzpicture}[scale=0.75]%[scale=0.65]
	\tikzset{man/.style = {shape=circle,draw,inner sep=1pt}}
	\tikzset{woman/.style = {shape=diamond,draw,inner sep=1pt}}
	\tikzset{edge/.style = {solid}}
	\tikzset{earedge/.style = {densely dotted}}
	% vertices
	\node[woman] (1) at (0,2.5)  {$w^{c_1}$};
	\node[woman] (2) at (0,-2.5){$w^{c_{L}}$};
	\node at (0,0) {$\vdots$};
	\node[man] (3) at (2,0)  {$m_1$};
	\node[woman] (4) at (4,0)  {$w_1$};
	\node[man] (5) at (6,0)  {$m_2$};
	\node[woman] (8) at (4,5)  {$f_1^1$};
	\node[man] (9) at (4,2.5)  {$e_1^1$};
	\node[man] (10) at (6,-5)  {$e_2^2$};
	\node[woman] (11) at (6,-2.5)  {$f_2^2$};
	\node[man] (12) at (16.5,0)  {$m_k$};
	\node[woman] (13) at (19,0)  {$w_k$};
	\node[man] (14) at (16.5,-5)  {$e_k^2$};
	\node[woman] (15) at (16.5,-2.5)  {$f_k^2$};
	\node[woman] (17) at (8,0)  {$w_2$};
	\node[man] (18) at (8,2.5)  {$e_2^1$};
	\node[woman] (19) at (8,5)  {$f_2^1$};
	\node[man] (20) at (11.5,0)  {$m_{k-1}$};
	\node[woman] (21) at (11.5,-2.5)  {$f_{k-1}^2$};
	\node[man] (22) at (11.5,-5)  {$e_{k-1}^2$};
	\node[woman] (23) at (14,0)  {$w_{k-1}$};
	\node[man] (24) at (14,2.5)  {$e_{k-1}^1$};
	\node[woman] (25) at (14,5)  {$f_{k-1}^1$};
	\node at (11,5) {$\cdots$};
	\node at (8.75,-5) {$\cdots$};
	% edges
	\draw[edge] (1) to node [near start,fill=white,inner sep=2pt] (123) {\footnotesize{$1$}} node [near end,fill=white,inner sep=2pt] (123) {\footnotesize{$1$}} (3);
	\draw[edge] (2) to node [near start,fill=white,inner sep=2pt] (223) {\footnotesize{$1$}} node [near end,fill=white,inner sep=2pt] (223) {\footnotesize{$1$}} (3);
	\draw[edge] (3) to node [near start,fill=white,inner sep=2pt] (324) {\footnotesize{$0$}} node [near end,fill=white,inner sep=2pt] (324) {\footnotesize{$\alpha$}} (4);
	\draw[edge] (4) to node [near start,fill=white,inner sep=2pt] (425) {\footnotesize{$0$}} node [near end,fill=white,inner sep=2pt] (425) {\footnotesize{$\alpha$}} (5);
	\draw[edge] (9) to node [near start,fill=white,inner sep=2pt] (928) {\footnotesize{$0$}} node [near end,fill=white,inner sep=2pt] (928) {\footnotesize{$\alpha$}} (8);
	\draw[edge] (9) to node [near start,fill=white,inner sep=2pt] (924) {\footnotesize{$1$}} node [near end,fill=white,inner sep=2pt] (924) {\footnotesize{$1$}} (4);
	\draw[edge] (11) to node [near start,fill=white,inner sep=2pt] (11210) {\footnotesize{$0$}} node [near end,fill=white,inner sep=2pt] (11210) {\footnotesize{$\alpha$}} (10);
	\draw[edge] (5) to node [near start,fill=white,inner sep=2pt] (5211) {\footnotesize{$1$}} node [near end,fill=white,inner sep=2pt] (5211) {\footnotesize{$1$}} (11);
	\draw[edge] (12) to node [near start,fill=white,inner sep=2pt] (12213) {\footnotesize{$0$}} node [near end,fill=white,inner sep=2pt] (12213) {\footnotesize{$\alpha$}} (13);
	\draw[edge] (15) to node [near start,fill=white,inner sep=2pt] (15214) {\footnotesize{$0$}} node [near end,fill=white,inner sep=2pt] (15214) {\footnotesize{$\alpha$}} (14);
	\draw[edge] (15) to node [near start,fill=white,inner sep=2pt] (15212) {\footnotesize{$1$}} node [near end,fill=white,inner sep=2pt] (15212) {\footnotesize{$1$}} (12);
	\draw[edge] (5) to node [near start,fill=white,inner sep=2pt] (5217) {\footnotesize{$0$}} node [near end,fill=white,inner sep=2pt] (5217) {\footnotesize{$\alpha$}} (17);
	\draw[edge,bend left] (9) to node [near start,fill=white,inner sep=2pt] (9217) {\footnotesize{$0$}} node [near end,fill=white,inner sep=2pt] (9217) {\footnotesize{$\alpha$}} (17);
	\draw[edge] (17) to node [near start,fill=white,inner sep=2pt] (17218) {\footnotesize{$1$}} node [near end,fill=white,inner sep=2pt] (17218) {\footnotesize{$1$}} (18);
	\draw[edge,postaction={decorate,decoration={text along path,text align=center}}] (17) to node [very near start,fill=white,inner sep=2pt] (17220) {\footnotesize{$0$}} node [midway,fill=white,inner sep=2pt] (17220) {$\dots$} node [very near end,fill=white,inner sep=2pt] (17220) {\footnotesize{$\alpha$}} (20);
	\draw[edge] (18) to node [near start,fill=white,inner sep=2pt] (18219) {\footnotesize{$0$}} node [near end,fill=white,inner sep=2pt] (18219) {\footnotesize{$\alpha$}} (19);
	\draw[edge, bend left,postaction={decorate,decoration={text along path,text align=center,text={}}}] (18) to node [very near start,fill=white,inner sep=2pt] (18223) {\footnotesize{$0$}}  node [midway,fill=white,inner sep=2pt] (18223) {$\cdots$} node [very near end,fill=white,inner sep=2pt] (18223) {\footnotesize{$\alpha$}} (23);
	\draw[edge, bend right,postaction={decorate,decoration={text along path,text align=center}}] (11) to node [very near start,fill=white,inner sep=2pt] (11220) {\footnotesize{$0$}} node [midway,fill=white,inner sep=2pt] (11220) {$\cdots$} node [very near end,fill=white,inner sep=2pt] (11220) {\footnotesize{$\alpha$}} (20);
	\draw[edge] (20) to node [near start,fill=white,inner sep=2pt] (20221) {\footnotesize{$1$}} node [near end,fill=white,inner sep=2pt] (20221) {\footnotesize{$1$}} (21);
	\draw[edge] (21) to node [near start,fill=white,inner sep=2pt] (21222) {\footnotesize{$0$}} node [near end,fill=white,inner sep=2pt] (21222) {\footnotesize{$\alpha$}} (22);
	\draw[edge] (20) to node [near start,fill=white,inner sep=2pt] (20223) {\footnotesize{$0$}} node [near end,fill=white,inner sep=2pt] (20223) {\footnotesize{$\alpha$}} (23);
	\draw[edge] (23) to node [near start,fill=white,inner sep=2pt] (23224) {\footnotesize{$1$}} node [near end,fill=white,inner sep=2pt] (23224) {\footnotesize{$1$}} (24);
	\draw[edge] (24) to node [near start,fill=white,inner sep=2pt] (24225) {\footnotesize{$0$}} node [near end,fill=white,inner sep=2pt] (24225) {\footnotesize{$\alpha$}} (25);
	\draw[edge] (23) to node [near start,fill=white,inner sep=2pt] (23212) {\footnotesize{$0$}} node [near end,fill=white,inner sep=2pt] (23212) {\footnotesize{$\alpha$}} (12);
	\draw[edge,bend left] (24) to node [near start,fill=white,inner sep=2pt] (24213) {\footnotesize{$0$}} node [near end,fill=white,inner sep=2pt] (24213) {\footnotesize{$\alpha$}} (13);
	\draw[edge,bend right] (21) to node [near start,fill=white,inner sep=2pt] (21212) {\footnotesize{$0$}} node [near end,fill=white,inner sep=2pt] (21212) {\footnotesize{$\alpha$}} (12);
	\end{tikzpicture}
}%\label{subfig:acc_alpha_lessthan_2}
	\caption{The accumulator for the cases (a) $\alpha\geq 2$ and (b) $\alpha\in (3/2,2)$.\label{fig:acc}}
\end{figure}

The accumulator (\Cref{fig:acc}) of instance $\I$ has different structure depending on the value of $\alpha$. Its size depends on the positive integer parameter $k$.
\begin{itemize}
	\item When $\alpha\geq 2$ (see \cref{fig:acc}a), the accumulator consists of man $m_i$ and woman $w_i$ for all $i \in \{1,\dots,k\}$, men $e^1_i$ and $e^2_i$ and woman $f^1_i$ for all $i \in \{1,\dots,k-1\}$, man $e^3_i$ and women $f^2_i$ and $f^3_i$ for all $i \in \{2,\dots,k\}$, and woman $w^c$ for every clause $c$ of $\phi$. The edges comprise of man-heavy edges $(m_i,w_{i-1})$ and $(e^3_i,f^2_i)$ for all $i \in \{2,\dots,k\}$ and $(e^2_i,w_i)$ for all $i \in \{1,\dots,k-1\}$, balanced edges $(m_1,w^c)$ for every clause $c$, which we call {\em tine} edges, $(e^1_i,w_i)$ for all $i \in \{1,\dots,k-1\}$, and $(m_i,f^2_i)$ for all $i \in \{2,\dots,k\}$, and woman-heavy edges $(m_i,w_i)$ for all $i \in \{1,\dots,k\}$, $(e^1_i,f^1_i)$ for all $i \in \{1,\dots,k-1\}$, and $(m_i,f^3_i)$ for all $i \in \{2,\dots,k\}$.
	\item When $\alpha \in (3/2,2)$ (see \cref{fig:acc}b), the accumulator consists of man $m_i$, woman $w_i$ for all $i \in \{1,\dots,k\}$, man $e^1_i$ and woman $f^1_i$ for all $i \in \{1,\dots,k-1\}$, man $e^2_i$ and woman $f^2_i$ for all $i \in \{2,\dots,k\}$,  and woman $w^c$ for every clause $c$ of $\phi$. The edges comprise of man-heavy edges $(m_i,w_{i-1})$ and $(e^2_i,f^2_i)$ for all $i \in \{2,\dots,k\}$ and $(m_i,f^2_{i-1})$ for all $i \in \{3,\dots,k\}$, the balanced edges $(m_1,w^c)$ for every clause $c$ (tine edges), $(e^1_i,w_i)$ for all $i \in \{1,\dots,k-1\}$ and $(m_i,f^2_i)$ for all $i \in \{2,\dots,k\}$, and the woman-heavy edges $(m_i,w_i)$ for all $i \in \{1,\dots,k\}$, and $(e^1_i,f^1_i)$ and $(e^1_i,w_{i+1})$ for all $i \in \{1,\dots,k-1\}$.
\end{itemize}

Finally, instance $\I$ has a clause-accumulator connector (or {\em CA-connector}) for every clause $c$ of $\phi$ consisting of the woman-heavy edge $(m^c,w^c)$ between the man $m^c$ (from the clause gadget corresponding to clause $c$) and woman $w^c$ (from the accumulator); see \cref{fig:variable-and-clause-gadget}b. Notice that the above construction has more women than men. To restore balance, we pad the instance with extra (isolated) men that neither value nor are valued by any other agent. % as isolated nodes in the graph representation.
This completes the construction of the reduced instance.

\subsection{Gadget properties}\label{sec:gadgets}
We will now prove several important properties (Claims~\ref{claim:variable}-\ref{claim:accum}) of our construction. 

\begin{claim}\label{claim:variable}
	For every variable $x$, a stable fractional matching $\mu$ satisfies at least one of the following:% conditions:
	\begin{enumerate}
		\item[(1)] $\mu(m^x_1,w^{\overline{x}}_1)+\mu(m^x_1,w^{\overline{x}}_2)+\mu(m^x_1,f^x_1)=1$ and $\mu(m^x_2,w^{\overline{x}}_1)+\mu(m^x_2,w^{\overline{x}}_2)+\mu(m^x_2,f^x_2)=1$.
		\item[(2)] $\mu(m^x_1,w^{\overline{x}}_1)+\mu(m^x_2,w^{\overline{x}}_1)+\mu(e^x_3,w^{\overline{x}}_1)=1$ and $\mu(m^x_1,w^{\overline{x}}_2)+\mu(m^x_2,w^{\overline{x}}_2)+\mu(e^x_3,w^{\overline{x}}_2)=1$.
	\end{enumerate}
\end{claim}
\begin{proof} 
Suppose, for contradiction, that for some $i,j \in \{1,2\}$, we have $\mu(m^x_i,w^{\overline{x}}_1)+\mu(m^x_i,w^{\overline{x}}_2)+\mu(m^x_i,f^x_i)<1$ and $\mu(m^x_1,w^{\overline{x}}_j)+\mu(m^x_2,w^{\overline{x}}_j)+\mu(e^x_3,w^{\overline{x}}_j)<1$. Then, both $m^x_i$ and $w^{\overline{x}}_j$ will have utility strictly less than $1$ in $\mu$, and thus the pair $(m^x_i,w^{\overline{x}}_j)$ would be blocking.
\end{proof}

We remark that the two conditions in the statement of \cref{claim:variable} affect the weight of the input edges of the VC-connectors that are attached to the variable gadget in any stable fractional matching. In particular, condition (1) implies that the weight assigned to the input edges of the VC-connectors that correspond to the two appearances of the positive literal $x$ in clauses must be $0$. To see why, observe that these input edges are incident to nodes $m^x_1$ and $m^x_2$, and the total weight of all edges incident to each of these nodes cannot exceed $1$. Condition (2) has a similar implication for the edges associated with the negative literal $\overline{x}$.

\begin{claim}\label{claim:vc-connector}
	Any stable fractional matching that assigns a weight of $0$ to the input edge of a VC-connector must assign a weight of $0$ to its output edge(s) as well.
\end{claim}
\begin{proof} 
For $\alpha\geq 2$, the claim holds trivially for VC-connectors corresponding to positive literals (\cref{fig:vc-connector}a).  Consider a VC-connector corresponding to a negative literal $\overline{x}$ and a clause $c$ containing it (\cref{fig:vc-connector}c). Observe that, besides the input edge $(m^{\overline{x},c},w^{\overline{x}}_i)$, the edge $(m^{\overline{x},c},w^{\overline{x},c})$ is the only balanced or man-heavy edge that is incident to man $m^{\overline{x},c}$ and the only balanced (or woman-heavy) edge incident to woman $w^{\overline{x},c}$. Hence, stability of the edge $(m^{\overline{x},c},w^{\overline{x},c})$ requires a weight of $1$ assigned to it when the weight assigned to input edge $(m^{\overline{x},c},w^{\overline{x}}_i)$ is $0$. Then, the output edge $(m^{\overline{x},c},w^{c}_j)$, which is also incident to the node $m^{\overline{x},c}$, must have a weight of $0$ as well.

We now consider the case $\alpha\in(3/2,2)$. First consider a VC-connector corresponding to a positive literal $x$ and a clause $c$ containing it (\cref{fig:vc-connector}b). The edge $(m^{x,c},w^{x,c})$ is the only balanced or man-heavy edge that is incident to man $m^{x,c}$ and, besides the input edge $(m^{x}_i,w^{x,c})$, the only balanced (or woman-heavy) edge incident to woman $w^{x,c}$. Hence, stability of edge $(m^{x,c},w^{x,c})$ requires a weight of $1$ assigned to it when the weight assigned to edge $(m^{x}_i,w^{x,c})$ is $0$. Then, the output edge $(m^{x,c},w^c_j)$, which is also incident to node $m^{x,c}$, must have a weight of $0$ as well.

Finally, consider a VC-connector corresponding to a negative literal $\overline{x}$ and a clause $c$ containing it (\cref{fig:vc-connector}d). Observe that $(m^{\overline{x},c}_1,w^{\overline{x},c}_1)$ is the only balanced or woman-heavy edge that is incident to woman $w^{\overline{x},c}_1$ and, besides the input edge $(m^{\overline{x},c}_1,w^{\overline{x}}_i)$, the only balanced or man-heavy edge incident to man $m^{\overline{x},c}_1$. Also,  $(m^{\overline{x},c}_2,w^{\overline{x},c}_2)$ is the only balanced or man-heavy edge that is incident to man $m^{\overline{x},c}_2$ and, besides edge $(m^{\overline{x},c}_1,w^{\overline{x},c}_2)$, the only balanced or woman-heavy edge to woman $w^{\overline{x},c}_2$. Furthermore, $(m^{\overline{x},c}_3,w^{\overline{x},c}_3)$ is the only balanced or man-heavy edge that is incident to woman $w^{\overline{x},c}_3$ and, besides edges $(m^{\overline{x},c}_3,w^{\overline{x},c}_1)$ and $(m^{\overline{x},c}_3,w^{\overline{x},c}_2)$, the only balanced or man-heavy edge incident to man $m^{\overline{x},c}_3$.

Hence, stability of edge $(m^{\overline{x},c}_1,w^{\overline{x},c}_1)$ requires a weight of $1$ assigned to it when the weight assigned to edge $(m^{\overline{x},c}_1,w^{\overline{x}}_i)$ is $0$. Then, edges $(m^{\overline{x},c}_1,w^{\overline{x},c}_2)$ and $(m^{\overline{x},c}_3,w^{\overline{x},c}_1)$ must have a weight of $0$. Then, stability of edge $(m^{\overline{x},c}_2,w^{\overline{x},c}_2)$ requires a weight of $1$ assigned to it and the edge  $(m^{\overline{x},c}_3,w^{\overline{x},c}_2)$ and the output edge $(m^{\overline{x},c}_2,w^c_j)$  must have a weight of $0$. Then, stability of edge $(m^{\overline{x},c}_3,w^{\overline{x},c}_3)$ requires a weight of $1$ assigned to it. Hence, the output edge $(m^{\overline{x},c}_3,w^{c}_j)$ must have a weight of $0$ as well.
\end{proof}

\begin{claim}\label{claim:clause}
	Any stable fractional matching that assigns a weight of $0$ to all output edges of the VC-connectors of clause $c$ must assign a weight of $0$ to the CA-connector of clause $c$ as well.	
\end{claim}
\begin{proof} 
Let $\ell_1$, $\ell_2$, and $\ell_3$ be the literals of clause $c$. Consider, for the sake of contradiction, a stable fractional matching that assigns (1) a weight of $0$ to all output edges of the VC-connectors corresponding to literals $\ell_i$ and clause $c$, and (2) strictly positive weight to edge $(m^c,w^c)$ of the CA-connector for clause $c$. Note that condition (2) implies that the total weight on the edges $(m^c,w^c_1)$, $(m^c,w^c_2)$ and $(m^c,w^c_3)$ is strictly smaller than $1$. Since these are the only balanced or man-heavy edges incident to man $m^c$, the stability of these edges is guaranteed by a utility of (at least) $1$ for each of the agents $w^c_1$, $w^c_2$, and $w^c_3$.

Note that besides the output edges of the VC-connectors, the edges $(e^c_1,w^c_i)$, $(e^c_2,w^c_i)$, and $(m^c,w^c_i)$ are the only balanced or woman-heavy edges incident to agent $w^c_i$ for all $i\in\{1,2,3\}$. Along with condition (1), this implies that the weight assigned to these three edges is at least $1$. Hence, the total weight on the nine edges of the clause gadget is at least $3$, i.e., strictly more than $2$ for the six edges incident to men $e^c_1$ and $e^c_2$, violating the definition of a fractional matching.
\end{proof}

\begin{claim}\label{claim:accum}
	Any stable fractional matching that assigns a weight of $0$ to some CA-connector must assign a total weight of $1$ to the tine edges and a weight of $1$ to every balanced edge of the accumulator.
\end{claim}
\begin{proof} 
	Assume that a weight of $0$ has been assigned to the edge $(m^{c'},w^{c'})$ of the CA-connector corresponding to some clause $c'$. Since this is the only woman-heavy edge that is incident to agent $w^{c'}$ and there is no man-heavy edge incident to agent $m_1$, stability on the edge $(m_1,w^{c'})$ requires that the total weight of the tine edges $(m_1,w^c)$ (for every clause $c$) is (at least) $1$. Hence, the weight of the edge $(m_1,w_1)$ is $0$. We will complete the proof by distinguishing between the two different accumulator structures, depending on whether $\alpha\geq 2$ or $\alpha\in (3/2,2)$.

When $\alpha\geq 2$, it suffices to show that for all $i \in \{1,\dots,k-1\}$, if the weight of edge $(m_i,w_i)$ is $0$, then the weight of the balanced edges $(e^1_i,w_i)$ and $(m_{i+1},f^2_{i+1})$ is $1$ and the weight of edge $(m_{i+1},w_{i+1})$ is $0$. Indeed, observe that, edge $(e^1_i,w_i)$ is the only balanced or man-heavy edge incident to man $e^1_i$ and, besides edge $(m_i,w_i)$, the only balanced or woman-heavy edge incident to woman $w_i$. Hence, the balanced edge $(e^1_i,w_i)$ must have a weight  of $1$ and the edge $(m_{i+1},w_i)$ a weight of $0$.

Then, edge $(m_{i+1},f^2_{i+1})$ is the only balanced or woman-heavy edge incident to woman $f^2_{i+1}$ and, besides edge $(m_{i+1},w_i)$, the only balanced or man-heavy edge incident to man $m_{i+1}$. Hence, the balanced edge $(m_{i+1},f^2_{i+1})$ must have a weight of $1$ and the edge $(m_{i+1},w_{i+1})$ a weight of $0$.

When $\alpha\in (3/2,2)$, it suffices to show that for all $i \in \{1,\dots,k-1\}$, if the weight of the edge $(m_i,w_i)$ (and, if they exist, the edges $(e^1_{i-1},w_i)$ and $(m_{i+1},f^2_i)$) is $0$, then the weight of the balanced edges $(e^1_i,w_i)$ and $(m_{i+1},f^2_{i+1})$ is $1$, and the weight of the edges $(m_{i+1},w_{i+1})$ and $(e^1_i,w_{i+1})$ (and, if it exists, the edge $(m_{i+2},f^2_{i+1})$) is $0$. Indeed, observe that the edge $(e^1_i,w_i)$ is the only balanced or man-heavy edge incident to man $e^1_i$. Furthermore, in addition to the edge $(m_i,w_i)$ (and, if it exists, the edge $(e^1_{i-1},w_i)$), it is the only balanced or woman-heavy edge incident to woman $w_i$. Hence, the balanced edge $(e^1_i,w_i)$ must have a weight of $1$, and the edge $(m_{i+1},w_i)$ (and, if it exists, the edge $(e^1_i,w_{i+1})$) must have a weight of $0$. Then, the edge $(m_{i+1},f^2_{i+1})$ is, besides $(m_{i+1},f^2_i)$ and $(m_{i+1},w_i)$, the only balanced or man-heavy edge incident to man $m_{i+1}$, and the only balanced or woman-heavy edge incident to woman $f^2_{i+1}$. Hence, the balanced edge $(m_{i+1},f^2_{i+1})$ must have a weight of $1$, and the edge $(m_{i+1},w_{i+1})$ (and, if it exists, the edge $(m_{i+2},f^2_{i+1})$) must have a weight of $0$.
\end{proof}

\subsection{Proof of inapproximability}\label{sec:proof}
\begin{lemma}\label{lem:unsat}
	If formula $\phi$ is not satisfiable, then any stable fractional matching of $\I$ has welfare at most $80\alpha N+4(k-1)$.
\end{lemma}

\begin{proof} We will first show that if $\phi$ is not satisfiable, then any stable fractional matching of $\I$ assigns weight $0$ to some CA-connector. For the sake of contradiction, consider a stable fractional matching that assigns a strictly positive weight to all CA-connectors. We will construct a truth assignment for the formula $\phi$ (contradicting the assumption of the lemma) by repeating the following process for every clause $c$ of $\phi$: Let $\ell$ be a literal that appears in $c$ such that the output edge(s) of the VC-connector, that corresponds to the appearance of $\ell$ in $c$, have strictly positive total weight. By \cref{claim:clause}, such a literal must exist. We set $\ell$ to $1$ (true). For every variable that has not received a value in this way, we arbitrarily set it to $1$.
	
	The above assignment satisfies all the clauses. To show that it is also valid, we need to argue that there is no variable $x$ such that both literals $x$ and $\overline{x}$ have been set to $1$. Assume, to the contrary, that literal $x$ is set to $1$ due to its appearance in a clause $c_1$, and literal $\overline{x}$ is set to $1$ due to its appearance in a different clause $c_2$. Thus, in the above assignment, the output edge(s) of the VC-connector between the literal $x$ and the clause $c_1$, as well as the VC-connector between the literal $\overline{x}$ and the clause $c_2$ have strictly positive (total) weight. By \cref{claim:vc-connector}, the input edges of both VC-connectors also have strictly positive weight. Let $i_1, i_2 \in \{1,2\}$ be such that the $i_1$-th appearance of $x$ is in the clause $c_1$ and the $i_2$-th appearance of $\overline{x}$ is in the clause $c_2$. Therefore, the said input edges are incident to the nodes $m^x_{i_1}$ and $w^{\overline{x}}_{i_2}$. Using \cref{claim:variable}, we get that the total weight on the edges incident to one of $m^x_{i_1}$ or $w^{\overline{x}}_{i_2}$ exceeds $1$, contradicting feasibility. Thus, the above assignment must be valid, which, in turn, implies that any stable fractional matching assigns weight $0$ to some CA-connector.
	
By \cref{claim:accum}, the contribution of the accumulator to the welfare is exactly $4k-2$ ($2$ from the tine edges plus $2$ from each balanced edge). Let us now consider the contribution of the edges that do not belong to the accumulator. This comprises of
		\begin{itemize}
			\item a total value of $20$ for the ten balanced edges of each of the $N$ variable gadgets,
			\item a total value of $\alpha$ (respectively, $2+2\alpha$) for the edges of each of the $2N$ VC-connectors corresponding to a positive literal when $\alpha\geq 2$ (respectively, $\alpha\in(3/2,2)$),
			\item a total value of $2+2\alpha$ (respectively, $6+6\alpha$) for the edges of each of the $2N$ VC-connectors corresponding to a negative literal when $\alpha\geq 2$ (respectively, $\alpha\in(3/2,2)$),
			\item a total value of $18+\alpha$ for the nine balanced edges of each of the $4N/3$ clause gadgets and their corresponding CA-connectors.
		\end{itemize}	
		It can be easily seen that $80\alpha N-2$ is a (loose) upper bound on the total value from these edges.
\end{proof}

\begin{lemma}\label{lem:sat}
	If $\phi$ is satisfiable, then there exists a stable fractional matching of $\I$ with welfare at least $4(k-1)(\alpha-1/2)$.
\end{lemma}

\begin{proof} Starting from a satisfying assignment for $\phi$, we will construct a stable fractional matching $\mu$ in which the welfare of the accumulator gadget is at least $4(k-1)(\alpha-1/2)$.

\paragraph{Variable gadgets.}
For the edges of the variable gadget of the variable $x$, $\mu$ is defined as:
	 \begin{itemize}
	 	\item If $x$ is true, then $\mu(m^x_1,w^{\overline{x}}_1)=\mu(e^x_3,w^{\overline{x}}_1)=\mu(e^x_3,w^{\overline{x}}_2)=\mu(m^x_2,w^{\overline{x}}_2)=1/2$, $\mu(e^x_1,f^x_1)=\mu(e^x_2,f^x_2)=1$, and
	 	the remaining edges have weight $0$.
	 	% $\mu(m^x_1,f^x_1)=\mu(m^x_1,w^{\overline{x}}_2)=\mu(m^x_2,f^x_2)=\mu(m^x_2,w^{\overline{x}}_1)=0$, and
	 	\item If $x$ is false, then $\mu(e^x_3,w^{\overline{x}}_1)=\mu(e^x_3,w^{\overline{x}}_2)=1/2$, $\mu(m^x_1,f^x_1)=\mu(m^x_2,f^x_2)=1$, and the remaining edges have weight $0$.	 	
	 	% $\mu(m^x_1,w^{\overline{x}}_1)=\mu(m^x_2,w^{\overline{x}}_2)=\mu(m^x_1,w^{\overline{x}}_2)=\mu(m^x_2,w^{\overline{x}}_1)=\mu(e^x_1,f^x_1)=\mu(e^x_2,f^x_2)=0$.
	 \end{itemize}
	
\paragraph{Clause gadgets and CA-connectors.} For each clause, select one of the true literals (tie-break arbitrarily) and call it {\em active}. Note that each clause has an active literal in a satisfying assignment. Consider the clause $c$, and let $\ell_i$ be its active literal for some $i\in \{1,2,3\}$. Also, let $i_1, i_2 \in \{1,2,3\} \setminus \{i\}$ denote the other two indices. Set $\mu(e^c_1,w^c_{i_1})=\mu(e^c_2,w^c_{i_2})=1$, and set the weight of the remaining balanced edges to $0$.
% $\mu(e^c_1,w^c_i)=\mu(e^c_1,w^c_{i_2})=\mu(e^c_2,w^c_i)=\mu(e^c_2,w^c_{i_1})=\mu(m^c,w^c_i)=\mu(m^c,w^c_{i_1})=\mu(m^c,w^c_{i_2})=0$.
 Assign a weight of $1$ to the CA-connector, i.e., $\mu(m^c,w^c)=1$.
	
\paragraph{VC-connectors.} For every non-active VC-connector, set the weight of its balanced edges (if any) to $1$ and the weight of the remaining edges to $0$. For every active VC-connector corresponding to the $i$-th appearance of the positive literal $x$ as the $j$-th literal of clause $c$ ($i\in\{1,2\}$, $j\in \{1,2,3\}$), the weights of its edges are as follows:
	\begin{itemize}
		\item When $\alpha\geq 2$, we set $\mu(m^x_i,w^c_j)=1/2$.
		\item When $\alpha\in (3/2,2)$, we set $\mu(m^x_i,w^{x,c})=1/2$, $\mu(m^{x,c},w^{x,c})=1-\alpha/2$, and $\mu(m^{x,c},w^c_j)=1/\alpha$.
	\end{itemize}

	For every active VC-connector corresponding to the $i$-th appearance of the negative literal $x$ as the $j$-th literal of clause $c$ ($i\in\{1,2\}$, $j\in \{1,2,3\}$), the weights of its edges are as follows:
	\begin{itemize}
		\item When $\alpha\geq 2$, we set $\mu(m^{\overline{x},c},w^{\overline{x}}_i)=\mu(m^{\overline{x},c},w^c_j)=1/2$ and $\mu(m^{\overline{x},c},w^{\overline{x},c})=0$.
		\item When $\alpha\in (3/2,2)$, we set $\mu(m^{\overline{x},c}_1,w^{\overline{x}}_i)=1/2$, $\mu(m^{\overline{x},c}_1,w^{\overline{x},c}_1)=1-\alpha/2$, $\mu(m^{\overline{x},c}_1, w^{\overline{x},c}_2)=(\alpha-1)/2$, $\mu(m^{\overline{x},c}_2,w^{\overline{x},c}_2)=1-(\alpha^2-\alpha)/2$, $\mu(m^{\overline{x},c}_2, w^{c}_j)=2/\alpha-1$, $\mu(m^{\overline{x},c}_3,w^{\overline{x},c}_1)=1/\alpha$, $\mu(m^{\overline{x},c}_3,w^{\overline{x},c}_2)=\mu(m^{\overline{x},c}_3,w^{\overline{x},c}_3)=0$, $\mu(m^{\overline{x},c}_3,w^c_j)=1-1/\alpha$.
	\end{itemize}

\paragraph{Accumulator.} We set $\mu(m_1,w^c)=0$ for every tine edge $(m_1,w^c)$ of the accumulator. Furthermore:
	\begin{itemize}
		\item When $\alpha\geq 2$, we set $\mu(m_i,w_i)=1/\alpha$ for all $i \in \{1,\dots,k\}$, $\mu(e^2_i,w_i)=1-2/\alpha$, $\mu(m_{i+1},w_i)=1/\alpha$, $\mu(e^1_i,f^1_i)=1$, $\mu(e^1_i,w_i)=0$ for all $i \in \{1,\dots,k-1\}$, $\mu(m_i,f^2_i)=0$, $\mu(m_i,f^3_i)=1-2/\alpha$, and $\mu(e^3_i,f^2_i)=1$ for all $i \in \{2,\dots,k\}$. Among these, any edge with a positive weight is either man- or woman-heavy, and hence, its contribution to the social welfare is $\alpha$ times its weight. It can be verified that the total contribution is $4(k-1)(\alpha-1/2)+1$.

		\item When $\alpha\in (3/2,2)$, we set $\mu(m_1,w_1)=1/\alpha$, $\mu(m_2,w_2)=\alpha+1/\alpha-2$,  $\mu(m_i,w_i)=1-1/\alpha$ for all $i \in \{3,\dots,k\}$, $\mu(m_{i+1},w_i)=1-1/\alpha$ for all $i \in \{1,\dots,k-1\}$, $\mu(e^1_i,w_i)=0$ for all $i \in \{1,\dots,k-1\}$, $\mu(m_2,f^2_2)=2-\alpha$, $\mu(m_i,f^2_i)=0$ for all $i \in \{3,\dots,k\}$, $\mu(e^1_1,f^1_1)=\alpha-1$, $\mu(e^2_2,f^2_2)=\alpha-2/\alpha$, $\mu(e^2_k,f^2_k)=1$, $\mu(e^1_i,f^1_i)=2-2/\alpha$ for all $i \in \{2,\dots,k-1\}$, $\mu(e^2_i,f^2_i)=2-2/\alpha$ for all $i \in \{3,\dots,k-1\}$, $\mu(e^1_1,w_2)=2-\alpha$, $\mu(e^1_i,w_{i+1})=2/\alpha-1$ for all $i \in \{2,\dots,k-1\}$, and $\mu(m_{i+1},f^2_i)=2/\alpha-1$ for all $i \in \{2,\dots,k-1\}$. Except for the balanced edge $(m_2,f^2_2)$, every edge with a positive weight among the ones listed above is either man- or woman-heavy, and hence, its contribution to the social welfare is $\alpha$ times its weight. It can be verified that the total contribution in this case is $4(k-1)(\alpha-1/2)+2\alpha^2-7\alpha+7$.
	\end{itemize}
	In each case, the accumulator contributes at least $4(k-1)(\alpha-1/2)$ to the social welfare, as desired.

The feasibility of $\mu$ can be verified by inspection. To see why $\mu$ is stable, note that we only need to check for the balanced edges, as the man- or woman-heavy edges and the remaining pairs do not impose any constraints on stability. For the balanced edges, stability is established by the following series of observations (we will use the term `stabilized by' to denote that an agent's utility is at least $1$): The variable gadget for the variable $x$ (\cref{fig:variable-and-clause-gadget}a) is stabilized by the agents $f_1^x$, $f_2^x$, $e_3^x$ along with $m_1^x$, $m_2^x$ (if $x$ is true) or $w_1^{\overline{x}}$, $w_2^{\overline{x}}$ (if $x$ is false). The clause gadget for clause $c$ (\cref{fig:variable-and-clause-gadget}b) with active index $i$ (and non-active indices $i_1$ and $i_2$) is stabilized by the agents $e^c_1$, $e^c_2$, $w_i^c$, $w_{i_1}^c$, $w_{i_2}^c$; in particular, the edge $(m^c,w_i^c)$ is stabilized by $w_i^c$ because an active literal triggers the woman-heavy edge in the VC-connector. A VC connector is stabilized by $w^{x,c}$ (\cref{fig:vc-connector}b), $m^{\overline{x},c}$ (\cref{fig:vc-connector}c), or $m_1^{\overline{x},c}$, $w_2^{\overline{x},c}$, and $m_3^{\overline{x},c}$ (\cref{fig:vc-connector}d). Finally, the tine edges in the accumulator (\Cref{fig:acc}) are stabilized by $w^{c_1},\dots,w^{c_L}$ (because we trigger the CA-connector), and the remaining balanced edges are stabilized by $w_i$'s and $m_i$'s except for $m_1$. Overall, $\mu$ is a feasible stable fractional matching.
\end{proof}

We are ready to prove \cref{thm:strong-inapprox}. If $\alpha <N^{1+1/\delta}$, we use our construction with any $k$ satisfying $k-1\geq \frac{20\alpha N(\alpha-1/2-\delta)}{\delta}$. It is easy to verify that the reduction is polynomial-time. Furthermore, from \cref{lem:unsat}, we know that the welfare of $\mu$ when $\phi$ is not satisfiable is at most
$$80\alpha N+4(k-1) \leq \frac{4(k-1)\delta}{\alpha-1/2-\delta}+4(k-1)=\frac{4(k-1)(\alpha-1/2)}{\alpha-1/2-\delta}.$$
This number is at least $\alpha-1/2-\delta$ times smaller than the welfare of $\mu$ when $\phi$ is satisfiable (\cref{lem:sat}). This establishes the inapproximability bound in part (i) of \Cref{thm:strong-inapprox}.

If $\alpha\geq N^{1+1/\delta}$, we use our construction with $k=N^{1+1/\delta}$. Once again, the reduction is polynomial-time, and the instance $\I$ has $n=\Theta(N^{1+1/\delta})$ men and women. Observe that $\alpha = \Omega(n)$, $k= \Theta(n)$, and $N=\bigO(n^\delta)$. Hence, the welfare of $\mu$ when $\phi$ is not satisfiable is at most $$80\alpha N+4(k-1) \leq 80\alpha N + 4N^{1+1/\delta} \leq 84\alpha N = \bigO(\alpha n^{\delta}).$$
On the other hand, the welfare of $\mu$ when $\phi$ is satisfiable is at least $4(k-1)(\alpha-1/2)$, i.e., $\Omega(\alpha n)$. This establishes the bound in part (ii), and with it, completes the proof of \Cref{thm:strong-inapprox}. 

\section{Concluding Remarks}
We studied stable fractional matchings in a cardinal model and provided a number of computational and structural results. Going forward, it would be very interesting to resolve the complexity of \OptStab{} for the case when agents have \emph{strict} cardinal preferences (i.e., the no-ties case).\footnote{This problem has recently been resolved in the work of \cite{CRS20fractional}.} It would also be very interesting to see if stronger inapproximability results can be obtained for more general matching models, such as stable roommates \citep{I85efficient}. Another relevant direction is to consider a stronger solution concept with blocking coalitions of any size (also known as the \emph{core}), wherein the deviating agents can form a fractional matching among themselves to achieve a higher utility for each member of the coalition.

\section*{Acknowledgments.}
We are grateful to Elliot Anshelevich for bringing the work of \cite{DMS17computational} to our attention, to Argyrios Deligkas for sharing with us the full version of their paper \citep{DMS17computational}, to Haris Aziz for pointing us to the work of \cite{M13stability}, and to Jiehua Chen, Sanjukta Roy, and Manuel Sorge for helpful correspondence. Many thanks to the anonymous reviewers for their thoughtful comments and suggestions. 
Ioannis Caragiannis and Aris Filos-Ratsikas acknowledge partial support from COST Action CA15210. Rohit Vaish acknowledges support from ONR\#{}N00014-17-1-2621 while he was affiliated with Rensselaer Polytechnic Institute, and is currently supported by project no. RTI4001 of the Department of Atomic Energy, Government of India. Part of this work was done while Rohit Vaish was supported by the Prof. R Narasimhan postdoctoral award. 

\bibliography{References}
\bibliographystyle{plainnat}

% Appendix
\newpage{}
\appendix
\section*{APPENDIX}
\section{Comparison with other notions of stability}
\label{sec:Stability_Notions}

In this section, we will discuss other notions of stability of fractional matchings that have been studied for ordinal preferences. For a detailed overview of these and other notions, we refer the reader to the work of \cite{AK19random}.

Given an \SMC{} instance $\I =$ $\langle M, W, U, V \rangle$ with cardinal preferences, we can define an ordinal instance $\I^{\text{ord}} = \langle M, W, \succeq \rangle$ where $\succeq \coloneqq (\succeq_{m_1},\dots,\succeq_{m_n},\succeq_{w_1},\dots,\succeq_{w_n})$ is a preference profile consisting of the ordinal preferences of the agents specified as weak total orders. Specifically, for every $i \in \{1,\dots,n\}$, $\succeq_{m_i}$ and $\succeq_{w_i}$ are weak total orders over the sets $W$ and $M$, respectively, such that
$U(m_i,w) \geq U(m_i,w') \Leftrightarrow w \succeq_{m_i} w'$ and
$V(m,w_i) \geq V(m',w_i) \Leftrightarrow m \succeq_{w_i} m'$. We will write $w \succ_m w'$ if $w \succeq_m w'$ but not $w' \succeq_m w$. The relation $m \succ_w m'$ is analogously defined.

Below we will discuss three notions of stability for fractional matchings---\emph{strong stability}~\citep{RRV93stable}, \emph{fractional stability}~\citep{V89linear}, and \emph{ex-post stability}~\citep{RRV93stable}---that were originally formulated in the context of strict ordinal preferences, and were subsequently studied for weak preferences by \cite{AK19random}.

The first notion called \emph{strong stability} asserts that no pair of man and woman should positively engage with agents they like less than each other.

\begin{definition}[Strong stability~\citep{RRV93stable}]
\label{def:Strong_Stability}
A fractional matching $\mu$ is \emph{strongly stable} if there are no $m,m' \in M$ and $w,w' \in W$ such that $\mu(m,w')>0$, $\mu(m',w)>0$, $w \succ_m w'$, and $m \succ_w m'$.
%%
%A fractional matching $\mu$ is \emph{strongly stable} if for every $(m,w) \in M \times W$, $$\left( \sum_{w' \in W : \, w \succ_m w'} \mu(m,w') \right) \cdot \left( \sum_{m' \in M : \, m \succ_w m'} \mu(m',w) \right) = 0.$$
\end{definition}

Next, recall that a fractional matching can be interpreted as a lottery (or a probability distribution) over integral matchings. Under this interpretation, \emph{ex-post stability} requires that every realization of this lottery should be stable.

\begin{definition}[Ex-post stability~\citep{RRV93stable}]
\label{def:Ex-post_Stability}
A fractional matching $\mu$ is \emph{ex-post stable} if it can be expressed as a convex combination of integral stable matchings.
\end{definition}

The third notion, \emph{fractional stability}, formalizes the idea that ``if man $m$ is matched with someone less preferred than woman $w$, then $w$ should be matched with someone more preferred than $m$'' as a linear constraint.

\begin{definition}[Fractional stability~\citep{V89linear}]
\label{def:Fractional_Stability}
A fractional matching $\mu$ is \emph{fractionally stable} if for every $(m,w) \in M \times W$,
$$\sum_{w' \in W: \, w' \succeq_{m} w} \mu(m,w') + \sum_{m' \in M: \, m' \succeq_{w} m} \mu(m',w) - \mu(m,w) \geq 1.$$
\end{definition}

Notice that an integral matching satisfies the \emph{ordinal} notions of strong (\Cref{def:Strong_Stability}), ex-post (\Cref{def:Ex-post_Stability}), or fractional (\Cref{def:Fractional_Stability}) stability in $\I^{\text{ord}}$ if and only if it satisfies the \emph{cardinal} notion of stability (\Cref{def:Stability}) in $\I$. Thus, the distinction between these notions is meaningful only for fractional matchings. The following observations describe the relationship between these notions (also refer to \Cref{fig:Stability_relations}). We recall that any integral or fractional matching is assumed to be complete unless stated otherwise.

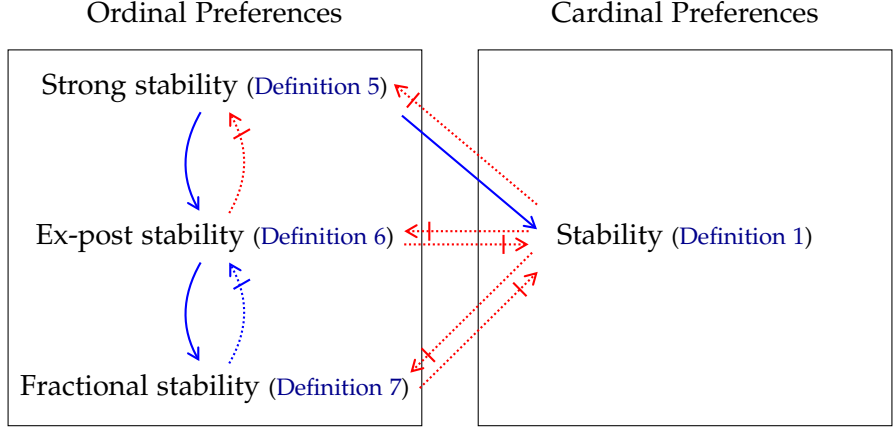
\begin{figure}%[t]
\begin{center}
%\scalebox{0.9}{
\begin{tikzpicture}
\tikzstyle{myedge}=[->,>=angle 60, shorten >=1pt,draw]
%\tikzstyle{myedgecancel}=[->,>=angle 60, shorten >=1pt,draw]
%
%\tikzset{
%myedgecancelbent/.style={->,>=angle 60, shorten >=1pt,draw,decoration={markings,mark=at position .5 with {\draw[-,thick] (-5pt,5pt) -- (5pt,-5pt);}},
%    postaction={decorate}
%	}
%}
%
%
%\tikzset{
%myedgecancel/.style={->,>=angle 60,dotted,draw,decoration={markings,mark=at position .5 with {\pgfuseplotmark{x}};},
%%{node [rotate=135, scale=2] at (1, 0) {\pgfuseplotmark{x}};},
%    postaction={decorate}
%	}
%}
%
\tikzset{
myedgecancel/.style={>=angle 60,densely dotted,draw,-{Bar[sep=4]>}}
}
% draw rectangles
	\draw[draw=black] (0.25,1.5) rectangle ++(5.5,5);
	\draw[draw=black] (6.5,1.5) rectangle ++(5.5,5);
	\node (strong) at (3,6) {Strong stability \footnotesize{(\Cref{def:Strong_Stability})}};
    \node (expost) at (3,4) {Ex-post stability \footnotesize{(\Cref{def:Ex-post_Stability})}};
	\node (fractional) at (3,2) {Fractional stability \footnotesize{(\Cref{def:Fractional_Stability})}};
	\node (stability) at (9.25,4) {Stability \footnotesize{(\Cref{def:Stability})}};
	\node (ordinal) at (3,7) {Ordinal Preferences};
	\node (cardinal) at (9.25,7) {Cardinal Preferences};
%%
% strong implies ex-post for weak
	\draw[myedge, thick, blue, bend right] (strong) to (expost);
% ex-post does not imply strong for strict
	\draw[myedgecancel, thick, red, bend right] (expost) to (strong);
% ex-post implies fractional for weak
	\draw[myedge, thick, blue, bend right] (expost) to (fractional);
% fractional does not imply ex-post for weak
	\draw[myedgecancel, thick, blue, bend right] (fractional) to (expost);
%%
% strong implies stability for weak
	\draw[myedge, thick, blue, shorten <= 2pt, shorten >= 5pt] (strong.south east) to (stability.west);
% stability does not imply strong for strict
	\draw[myedgecancel, thick, red, shorten <= 5pt, shorten >= -2pt] (stability.north west) to (strong.east);
% stability does not imply fractional for strict
	\draw[myedgecancel, thick, red, shorten <= 8pt, shorten >= -5pt] (stability.west) to (fractional.north east);
% fractional does not imply stability for strict
	\draw[myedgecancel, thick, red, shorten >= 5pt] (fractional.east) to (stability.south west);
% stability does not imply ex-post for strict
	\draw[myedgecancel, thick, red, shorten <= 7pt] ([yshift=2.5pt]stability.west) to ([yshift=2.5pt]expost.east);
% ex-post does not imply stability for strict
	\draw[myedgecancel, thick, red, shorten >= 7pt] ([yshift=-2.5pt]expost.east) to ([yshift=-2.5pt]stability.west);
\end{tikzpicture}
%}
\end{center}	
\caption{Relationship between stability notions for fractional matchings under ordinal and cardinal preferences. A solid (respectively, dotted) arrow indicates that the implication holds (respectively, does not hold) in that direction. The implications for strict and weak preferences are denoted by red and blue arrows, respectively.}
\label{fig:Stability_relations}
\end{figure}

\begin{itemize}
	\item It is known that strong stability implies ex-post stability even under weak ordinal preferences, but the converse is not true even under strict ordinal preferences~\citep{RRV93stable,AK19random}. Additionally, for strict preferences, ex-post stability is known to be equivalent to fractional stability~\citep{V89linear,R92characterization,RRV93stable,TS98geometry}, but for weak preferences, the former is strictly stronger~\citep{AK19random}.
	\item A stable matching (\Cref{def:Stability}) may not be fractionally stable (\Cref{def:Fractional_Stability}) even under strict preferences. To see this, consider the fractional matching $\mu$ in \Cref{eg:Motivating_Example}, which was shown to be stable. For the pair $(m_1,w_3)$, we have
	$$\mu(m_1,w_3) + \sum_{w' : w' \>_{m_1} w_3} \mu(m_1,w') + \sum_{m' : m' \>_{w_3} \> m_1} \mu(m',w_3) = 0 + 0 + 0.5 < 1,$$
	implying that $\mu$ is not fractionally stable (and thus also not ex-post/strongly stable).
	
	\begin{remark}
		It is interesting to note the welfare implication of the above observation. Since fractionally stable matchings are equivalent to ex-post stable matchings under strict preferences, and since all integral stable matchings in \Cref{eg:Motivating_Example} have social welfare of $7$, it follows that the social welfare of any fractionally stable matching is also equal to $7$. On the other hand, there exists a stable matching $\mu$ with a strictly higher welfare. This illustrates that fractionally stable (and, in particular, strongly stable) matchings can be strictly suboptimal in the cardinal model, further justifying the need to study the computational aspects of optimal stable matchings.		\label{rem:Strong_Stability_Suboptimal}
		\end{remark}
	
	\item Given an \SMC{} instance $\I$, a strongly stable (\Cref{def:Strong_Stability}) matching in the corresponding ordinal instance $\I^{\text{ord}}$ can be shown to be stable (\Cref{def:Stability}) in the original instance $\I$. Indeed, if $\mu$ is strongly stable in $\I^{\text{ord}}$, then for any pair $(m,w) \in M \times W$, we have that either $\sum_{w' : w \succ_m w'} \mu(m,w') = 0$ or $\sum_{m' : m \succ_w m'} \mu(m',w) = 0$. Suppose, without loss of generality, that the former holds. Then, the utility of man $m$ in $\I$ is given by
	\begin{align*}
		u_m(\mu) &= U(m,w)\mu(m,w) + \sum_{w' \neq w \, : \, w' \succeq_m w} U(m,w')\mu(m,w')\\
		& \geq U(m,w) \left( \mu(m,w) + \sum_{w' \neq w \, : \, w' \succeq_m w} \mu(m,w') \right)\\
		& = U(m,w),
	\end{align*}
	as desired.
	\item As mentioned previously in \Cref{rem:Nonconvexity_Strict}, \cite{NN20study} have observed that the set of stable matchings is non-convex even under strict cardinal preferences. Their counterexample involves the convex combination of integral stable matchings, and therefore also establishes that an ex-post stable (and hence also fractionally stable) matching may not be stable.
%	\item Even under strict preferences, fractional stability (\Cref{def:Fractional_Stability}) in $\I^{\text{ord}}$ may not imply stability (\Cref{def:Stability}) in the original instance $\I$. This can be shown by considering an \SMC{} instance with the following valuations:
%	$$U = \begin{bmatrix}
%2 & 3 & 0\\
%0 & 2 & 3\\
%3 & 0 & 2
%\end{bmatrix} \text{ and } V = \begin{bmatrix}
%2 & 0 & 3\\
%3 & 2 & 0\\
%0 & 3 & 2
%\end{bmatrix}.$$
%The fractional matching $\mu$ given by
%$$\mu = \begin{bmatrix}
%0.4 & 0.3 & 0.3\\
%0.3 & 0.4 & 0.3\\
%0.3 & 0.3 & 0.4
%\end{bmatrix}$$
%satisfies fractional stability but violates stability for the pairs $(m_1,w_1)$, $(m_2,w_2)$, and $(m_3,w_3)$.
\end{itemize}

\section{Non-convexity of strong stability may not imply non-convexity of stability}
\label{sec:Non-convexity_Implication}
In this section, we show that non-convexity of strong stability (\Cref{def:Strong_Stability}) might not imply the same for stability (\Cref{def:Stability}), even though, as observed in \Cref{sec:Stability_Notions}, the former is a strictly stronger notion. To this end, we revisit the counterexample used by \cite{AK19random} in establishing the non-convexity of the set of strongly stable matchings. Recall from \Cref{sec:Stability_Notions} that a fractional matching $\mu$ is strongly stable if there are no $m,m' \in M$ and $w,w' \in W$ such that $\mu(m,w')>0$, $\mu(m',w)>0$, $w \succ_m w'$, and $m \succ_w m'$.

The counterexample of \citet[Theorem 1]{AK19random}, which in turn is adapted from \citep[Example 2]{RRV93stable}, consists of three men $m_1,m_2,m_3$ and three women $w_1,w_2,w_3$ with the following ordinal preferences:
\begin{align*}
	m_1 : w_1 \succ w_2 \succ w_3 && w_1 : m_2 \succ m_3 \succ m_1\\
	m_2 : w_2 \succ w_3 \succ w_1 && w_2 : m_3 \succ m_1 \succ m_2\\
	m_3 : w_3 \succ w_1 \succ w_2 && w_3 : m_1 \succ m_2 \succ m_3.
\end{align*}
Consider the integral matchings $\mu^{(1)}$, $\mu^{(2)}$, and $\mu^{(3)}$ defined as follows:
\begin{align*}
	\mu^{(1)} \coloneqq \{(m_1,w_1),(m_2,w_2),(m_3,w_3)\},&\\
	\mu^{(2)} \coloneqq \{(m_1,w_3),(m_2,w_1),(m_3,w_2)\},& \text{ and }\\
	\mu^{(3)} \coloneqq \{(m_1,w_2),(m_2,w_3),(m_3,w_1)\}.&
\end{align*}
Notice that $\mu^{(1)}$, $\mu^{(2)}$, and $\mu^{(3)}$ are strongly stable (and therefore also stable). Also note that the fractional matching $\mu \coloneqq \frac{1}{3} \mu^{(1)} + \frac{1}{3} \mu^{(2)} + \frac{1}{3} \mu^{(3)}$ violates strong stability for the agents $m_1,m_2$ and $w_2,w_3$ since $\mu(m_1,w_3) > 0$, $\mu(m_2,w_2) > 0$, $w_2 \, \>_{m_1} \, w_3$ and $m_1 \, \>_{w_2} \, m_2$. That is, the convex combination of strongly stable integral matchings $\mu^{(1)}$, $\mu^{(2)}$, and $\mu^{(3)}$ is not strongly stable.

Let us now consider an \SMC{} instance $\I = \langle M,W,U,V \rangle$ for the same set of agents, i.e., $M = \{m_1,m_2,m_3\}$, $W = \{w_1,w_2,w_3\}$, and the following valuations:
$$U = \begin{bmatrix}
2 & 1 & 0\\
0 & 2 & 1\\
1 & 0 & 2
\end{bmatrix} \text{ and } V = \begin{bmatrix}
0 & 1 & 2\\
2 & 0 & 1\\
1 & 2 & 0
\end{bmatrix}.$$
Observe that the valuations in $\I$ are consistent with the aforementioned ordinal preferences. It is easy to see that the utility of each agent in $\mu$ is equal to $1$, and that $\mu$ is stable for $\I$.

\section[Hardness for approximate stability]{Hardness for $\eps$-stability}\label{app:eps-stability}
The proof of \cref{thm:strong-eps-inapprox} follows along very similar lines to the proof of \cref{thm:strong-inapprox}, again using a reduction from 2P2N-3SAT.

\subsection{The reduction}
Starting from an instance of 2P2N-3SAT, we first preprocess and augment it in the following way. For each variable of the original instance, we create a copy-variable and, for each clause of the original instance, we create a copy-clause that contains the copy-variables corresponding to the variables of the original clause. Each variable and its copy are {\em coupled variables} and, similarly, each clause and its copy are {\em coupled clauses}.

Let the modified input consist of $N$ (boolean) variables $x_1, x_2,\dots, x_N$, and a 3-CNF formula $\phi$ with $L=4N/3$ clauses $c_1$, $c_2$,\dots, $c_{4N/3}$. Note that if $\phi$ is not satisfiable, then, due to the instance augmentation, there exist at least two clauses that are not satisfied.

We now proceed to describe the instance $\I=\langle M,W,U,V \rangle$ of  \OptEpsStab\ whose graph representation consists of variable gadgets, clause gadgets, VC-connectors, an accumulator, and CA-connectors. We remark that $\I$ is not an instance with ternary valuations as we use valuations from the set $\{0,1,\beta,\gamma\}$. We set $\beta = 2(1-\eps)$, and observe that, when $\eps < 0.03$, we have that $(3\beta^2+4)\eps < 1/2$; the value of $\gamma$ will be set later. Again, we denote by $n$ the number of men (or women) in $\I$; the following reduction is such that $n= \bigO(N)$.

For each gadget, we classify the edges into the following three types:
\begin{itemize}
\item {\em balanced} edges $(m,w)$ with $U(m,w) = V(m,w) = 1$,
\item {\em man-heavy} edges $(m,w)$ with $U(m,w)>0$ and $V(m,w)=0$, and
\item {\em woman-heavy} edges $(m,w)$ with $U(m,w)=0$ and $V(m,w)>0$.
\end{itemize}
Any other pair $(m,w)$ that does not appear as an edge in the graph representation has $U(m,w)=V(m,w)=0$.

The instance $\I$ contains a variable gadget for every variable $x$, which consists of five men $m^{x}_1$, $m^{x}_2$, $e^x_1$, $e^x_2$, $e^x_3$, four women $w^{\overline{x}}_1$, $w^{\overline{x}}_2$, $f^x_1$, $f^x_2$ and the ten balanced edges $(e^x_1,f^x_1)$, $(m^x_1,f^x_1)$, $(m^x_1,w^{\overline{x}}_1)$, $(e^x_3,w^{\overline{x}}_1)$, $(e^x_3,w^{\overline{x}}_2)$, $(m^x_2,w^{\overline{x}}_2)$,  $(m^x_2,f^x_2)$, $(e^x_2,f^x_2)$, $(m^x_1,w^{\overline{x}}_2)$, and $(m^x_2,w^{\overline{x}}_1)$ (see \cref{fig:eps:variable-and-clause-gadget}a).

For every clause $c$, instance $\I$ has a clause gadget with three men $m^c$, $e^c_1$, $e^c_2$, three women $w^{c}_1$, $w^c_2$, $w^c_3$, and the nine balanced edges between them (see \cref{fig:eps:variable-and-clause-gadget}b).

\begin{figure}%[ht]
	\footnotesize
	\centering
\subfloat[]{
	\begin{tikzpicture}[scale=0.75]
	\tikzset{man/.style = {shape=circle,draw,inner sep=1pt}}
	\tikzset{woman/.style = {shape=diamond,draw,inner sep=1pt}}
	\tikzset{edge/.style = {solid}}
	\tikzset{earedge/.style = {densely dotted}}
	% vertices
	\node[man]   (1) at (0,0)  {$e^x_1$};
	\node[woman] (2) at (0,2)  {$f^x_1$};
	\node[man]   (3) at (0,4)  {$m^x_1$};
	\node[woman] (4) at (2,4)  {$w^{\overline{x}}_1$};
	\node[man]   (5) at (4,4)  {$e^x_3$};
	\node[woman] (6) at (6,4)  {$w^{\overline{x}}_2$};
	\node[man]   (7) at (8,4)  {$m^x_2$};
	\node[woman] (8) at (8,2)  {$f^x_2$};
	\node[man]   (9) at (8,0)  {$e^x_2$};
	% edges
	\draw[edge] (1) to node [near start,fill=white,inner sep=2pt] (122) {\scriptsize{$1$}} node [near end,fill=white,inner sep=2pt] (122) {\scriptsize{$1$}} (2);
	\draw[edge] (2) to node [near start,fill=white,inner sep=2pt] (223) {\scriptsize{$1$}} node [near end,fill=white,inner sep=2pt] (223) {\scriptsize{$1$}} (3);
	\draw[edge] (3) to node [near start,fill=white,inner sep=2pt] (324) {\scriptsize{$1$}} node [near end,fill=white,inner sep=2pt] (324) {\scriptsize{$1$}} (4);
	\draw[edge] (4) to node [near start,fill=white,inner sep=2pt] (425) {\scriptsize{$1$}} node [near end,fill=white,inner sep=2pt] (425) {\scriptsize{$1$}} (5);
	\draw[edge] (5) to node [near start,fill=white,inner sep=2pt] (526) {\scriptsize{$1$}} node [near end,fill=white,inner sep=2pt] (526) {\scriptsize{$1$}} (6);
	\draw[edge] (6) to node [near start,fill=white,inner sep=2pt] (627) {\scriptsize{$1$}} node [near end,fill=white,inner sep=2pt] (627) {\scriptsize{$1$}} (7);
	\draw[edge] (7) to node [near start,fill=white,inner sep=2pt] (728) {\scriptsize{$1$}} node [near end,fill=white,inner sep=2pt] (728) {\scriptsize{$1$}} (8);
	\draw[edge] (8) to node [near start,fill=white,inner sep=2pt] (829) {\scriptsize{$1$}} node [near end,fill=white,inner sep=2pt] (829) {\scriptsize{$1$}} (9);
	\draw[edge,bend left] (3) to node [near start,fill=white,inner sep=2pt] (326) {\scriptsize{$1$}} node [near end,fill=white,inner sep=2pt] (326) {\scriptsize{$1$}} (6);
	\draw[edge,bend right] (4) to node [near start,fill=white,inner sep=2pt] (427) {\scriptsize{$1$}} node [near end,fill=white,inner sep=2pt] (427) {\scriptsize{$1$}} (7);
	\end{tikzpicture}
}%\label{fig:variable-gadget}
\hspace{1in}
\subfloat[]{
	\begin{tikzpicture}[scale=0.75]
	\tikzset{man/.style = {shape=circle,draw,inner sep=1pt}}
	\tikzset{woman/.style = {shape=diamond,draw,inner sep=1pt}}
	\tikzset{edge/.style = {solid}}
	\tikzset{earedge/.style = {densely dotted}}
	% vertices
	\node[man]   (1) at (0,1)  {$e^c_2$};
	\node[man]   (2) at (0,3)  {$e^c_1$};
	\node[woman] (3) at (2,0)  {$w^c_3$};
	\node[woman] (4) at (2,2)  {$w^c_2$};
	\node[woman] (5) at (2,4)  {$w^c_1$};
	\node[man]   (6) at (4,2){$m^c$};
	\node[woman] (7) at (6,2)  {$w^c$};
	% edges
	\draw[edge] (1) to node [near start,fill=white,inner sep=2pt] (123) {\scriptsize{$1$}} node [near end,fill=white,inner sep=2pt] (123) {\scriptsize{$1$}} (3);
	\draw[edge] (1) to node [near start,fill=white,inner sep=2pt] (124) {\scriptsize{$1$}} node [near end,fill=white,inner sep=2pt] (124) {\scriptsize{$1$}} (4);
	\draw[edge] (1) to node [very near start,fill=white,inner sep=2pt] (125) {\scriptsize{$1$}} node [very near end,fill=white,inner sep=2pt] (125) {\scriptsize{$1$}} (5);
	\draw[edge] (2) to node [very near start,fill=white,inner sep=2pt] (223) {\scriptsize{$1$}} node [very near end,fill=white,inner sep=2pt] (223) {\scriptsize{$1$}} (3);
	\draw[edge] (2) to node [near start,fill=white,inner sep=2pt] (224) {\scriptsize{$1$}} node [near end,fill=white,inner sep=2pt] (224) {\scriptsize{$1$}} (4);
	\draw[edge] (2) to node [near start,fill=white,inner sep=2pt] (225) {\scriptsize{$1$}} node [near end,fill=white,inner sep=2pt] (225) {\scriptsize{$1$}} (5);
	\draw[edge] (5) to node [near start,fill=white,inner sep=2pt] (526) {\scriptsize{$1$}} node [near end,fill=white,inner sep=2pt] (526) {\scriptsize{$1$}} (6);
	\draw[edge] (4) to node [near start,fill=white,inner sep=2pt] (426) {\scriptsize{$1$}} node [near end,fill=white,inner sep=2pt] (426) {\scriptsize{$1$}} (6);
	\draw[edge] (3) to node [near start,fill=white,inner sep=2pt] (326) {\scriptsize{$1$}} node [near end,fill=white,inner sep=2pt] (326) {\scriptsize{$1$}} (6);
	\draw[edge] (6) to node [near start,fill=white,inner sep=2pt] (627) {\scriptsize{$0$}} node [near end,fill=white,inner sep=2pt] (627) {\scriptsize{$1$}} (7);
	\end{tikzpicture}
}%\label{fig:clause-gadget}
\caption{(a) The variable gadget corresponding to the variable $x$. (b) The clause gadget corresponding to the clause $c$ and its CA-connector $(m^c,w^c)$.\label{fig:eps:variable-and-clause-gadget}}
\end{figure}
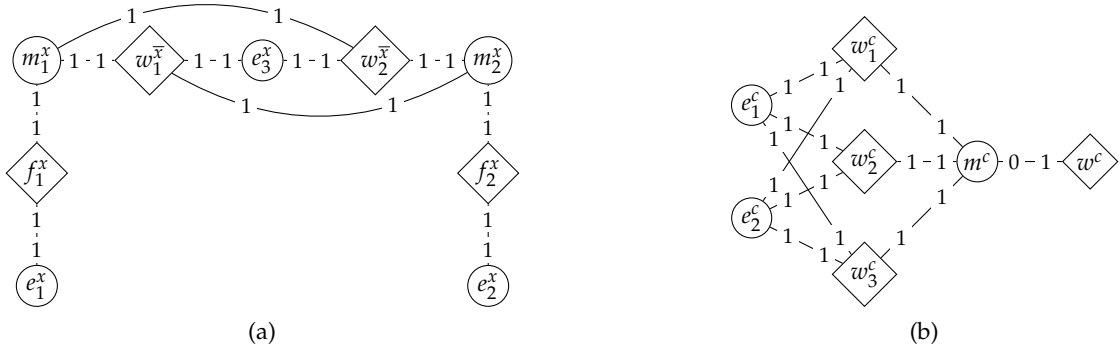

For every appearance of a literal in a clause, there is a variable-clause connector (or VC-connector) whose structure depends on whether it corresponds to a positive or a negative literal. In particular, for every positive literal $x$ whose $i$-th appearance ($i \in \{1,2\}$) is as the $j$-th literal ($j \in \{1,2,3\}$) of clause $c$, instance $\I$ has a VC-connector that consists of a single woman-heavy edge between $m^x_i$ (from the variable gadget corresponding to variable $x$) and $w^c_j$ (from the clause gadget corresponding to clause $c$) such that $U(m^x_i,w^c_j)=0$ and $V(m^x_i,w^c_j)=\beta$ (see \cref{fig:eps:vc-connector}a). This edge is simultaneously the input and output edge of the VC-connector.

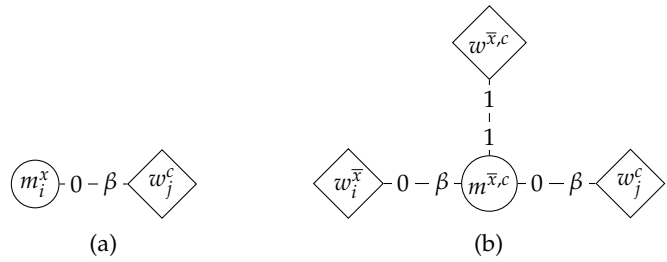
\begin{figure}%[ht]
\footnotesize
\centering
\subfloat[]{
	\begin{tikzpicture}[scale=0.75]
	\tikzset{man/.style = {shape=circle,draw,inner sep=1pt}}
	\tikzset{woman/.style = {shape=diamond,draw,inner sep=1pt}}
	\tikzset{edge/.style = {solid}}
	\tikzset{earedge/.style = {densely dotted}}
	% vertices
	\node[man]   (1) at (0,0)  {$m^x_i$};
	\node[woman] (2) at (2.25,0){$w^c_j$};
	% edges
	\draw[edge] (1) to node [near start,fill=white,inner sep=2pt] (122) {\footnotesize{$0$}} node [near end,fill=white,inner sep=2pt] (122) {\footnotesize{$\beta$}} (2);
	\end{tikzpicture}
}%\label{subfig:vc-connector_pos_literal_alpha_geq_2}
\hspace{0.5in}
\subfloat[]{
	\begin{tikzpicture}[scale=0.75]
	\tikzset{man/.style = {shape=circle,draw,inner sep=1pt}}
	\tikzset{woman/.style = {shape=diamond,draw,inner sep=1pt}}
	\tikzset{edge/.style = {solid}}
	\tikzset{earedge/.style = {densely dotted}}
	% vertices
	\node[woman] (1) at (0,0) {$w^{\overline{x}}_i$};
	\node[man] (2) at (2.5,0) {$m^{\overline{x},c}$};
	\node[woman] (3) at (5,0)  {$w^c_j$};
	\node[woman] (4) at (2.5,2.5) {$w^{\overline{x},c}$};
	% edges
	\draw[edge] (1) to node [near start,fill=white,inner sep=2pt] (122) {\footnotesize{$0$}} node [near end,fill=white,inner sep=2pt] (122) {\footnotesize{$\beta$}} (2);
	\draw[edge] (2) to node [near start,fill=white,inner sep=2pt] (223) {\footnotesize{$0$}} node [near end,fill=white,inner sep=2pt] (223) {\footnotesize{$\beta$}} (3);
	\draw[edge] (2) to node [near start,fill=white,inner sep=2pt] (224) {\footnotesize{$1$}} node [near end,fill=white,inner sep=2pt] (224) {\footnotesize{$1$}} (4);
	\end{tikzpicture}
}%\label{subfig:vc-connector_neg_literal_alpha_geq_2}
\caption{VC-connectors corresponding to (a) clause $c$ and positive literal $x$, and (b) clause $c$ and negative literal $\overline{x}$.\label{fig:eps:vc-connector}}
\end{figure}

Similarly, for every negative literal $\overline{x}$ whose $i$-th appearance ($i \in \{1,2\}$) is as the $j$-th literal ($j \in \{1,2,3\}$) of clause $c$, instance $\I$ has a VC-connector that consists of man $m^{\overline{x},c}$, woman $w^{\overline{x},c}$, the man-heavy edge $(m^{\overline{x},c},w^{\overline{x}}_i)$ with $U(m^{\overline{x},c},w^{\overline{x}}_i)=\beta$ and $V(m^{\overline{x},c},w^{\overline{x}}_i)=0$, the balanced edge $(m^{\overline{x},c},w^{\overline{x},c})$, and the woman-heavy edge $(m^{\overline{x},c},w^c_j)$ with $U(m^{\overline{x},c},w^c_j)=0$ and $V(m^{\overline{x},c},w^c_j)=\beta$ (see \cref{fig:eps:vc-connector}b). Here, $(m^{\overline{x},c},w^{\overline{x}}_i)$ is the input and $(m^{\overline{x},c},w^c_j)$ is the output edge.

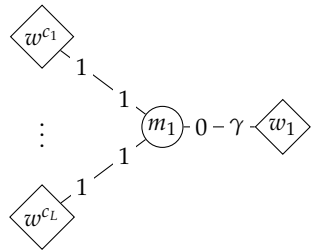
\begin{figure}%[ht]
\footnotesize
\centering
%
%\subfloat[]{
	\begin{tikzpicture}[scale=0.8]
	\tikzset{man/.style = {shape=circle,draw,inner sep=1pt}}
	\tikzset{woman/.style = {shape=diamond,draw,inner sep=1pt}}
	\tikzset{edge/.style = {solid}}
	\tikzset{earedge/.style = {densely dotted}}
	% vertices
	\node[woman] (1) at (0,1.5)  {$w^{c_1}$};
	\node[woman] (2) at (0,-1.5){$w^{c_{L}}$};
	\node at (0,0) {$\vdots$};
	\node[man] (3) at (2,0)  {$m_1$};
	\node[woman] (4) at (4,0)  {$w_1$};
	% edges
	\draw[edge] (1) to node [near start,fill=white,inner sep=2pt] (123) {\footnotesize{$1$}} node [near end,fill=white,inner sep=2pt] (123) {\footnotesize{$1$}} (3);
	\draw[edge] (2) to node [near start,fill=white,inner sep=2pt] (223) {\footnotesize{$1$}} node [near end,fill=white,inner sep=2pt] (223) {\footnotesize{$1$}} (3);
	\draw[edge] (3) to node [near start,fill=white,inner sep=2pt] (324) {\footnotesize{$0$}} node [near end,fill=white,inner sep=2pt] (324) {\footnotesize{$\gamma$}} (4);
	\end{tikzpicture}
%}%\label{subfig:acc_alpha_geq_2}
%
	\caption{The accumulator.\label{fig:eps:acc}}
\end{figure}

The accumulator gadget consists of man $m_1$, woman $w_1$ and woman $w^c$ for every clause $c$ of $\phi$. In addition, it contains the balanced edges $(m_1,w^c)$ for every clause $c$, which we once again call \emph{tine} edges and the woman-heavy edge $(m_1,w_1)$ with $U(m_1,w_1)=0$ and $V(m_1,w_1)=\gamma$ (see \cref{fig:eps:acc}).

Finally, instance $\I$ contains a clause-accumulator connector (or CA-connector) for every clause $c$ consisting of the woman-heavy edge $(m^c,w^c)$ between the man $m^c$ (from the clause gadget corresponding to clause $c$) and woman $w^c$ (from the accumulator) so that $U(m^c,w^c) = 0$ and $V(m^c,w^c) =1$ (see \cref{fig:eps:variable-and-clause-gadget}b). Any other edge $(m,w)$ is such that $U(m,w) = V(m,w) = 0$.

Notice that the above construction has more women than men. To restore balance, we pad the instance with extra (isolated) men that neither value nor are valued by any other agent. % as isolated nodes in the graph representation.
This completes the construction of the reduced instance.

\subsection{Gadget properties}
We will now prove several important properties (Claims~\ref{claim:eps-variable}-\ref{claim:eps-accum}) of our construction.
\begin{claim}\label{claim:eps-variable}
	For every variable $x$, an $\eps$-stable fractional matching $\mu$ satisfies at least one of the following conditions:
	\begin{enumerate}
		\item[(1)] $\mu(m^x_1,w^{\overline{x}}_1)+\mu(m^x_1,w^{\overline{x}}_2)+\mu(m^x_1,f^x_1) \geq 1-\eps$ and $\mu(m^x_2,w^{\overline{x}}_1)+\mu(m^x_2,w^{\overline{x}}_2)+\mu(m^x_2,f^x_2) \geq 1-\eps$.
		\item[(2)] $\mu(m^x_1,w^{\overline{x}}_1)+\mu(m^x_2,w^{\overline{x}}_1)+\mu(e^x_3,w^{\overline{x}}_1) \geq 1-\eps$ and $\mu(m^x_1,w^{\overline{x}}_2)+\mu(m^x_2,w^{\overline{x}}_2)+\mu(e^x_3,w^{\overline{x}}_2)\geq 1-\eps$.
	\end{enumerate}
\end{claim}

\begin{proof}
	Assume otherwise that, for some $i,j \in \{1,2\}$, $\mu(m^x_i,w^{\overline{x}}_1)+\mu(m^x_i,w^{\overline{x}}_2)+\mu(m^x_i,f^x_i)<1-\eps$ and $\mu(m^x_1,w^{\overline{x}}_j)+\mu(m^x_2,w^{\overline{x}}_j)+\mu(e^x_3,w^{\overline{x}}_j)<1-\eps$. Then, since instance $\I$ contains no man-heavy edge $(m^x_i,w)$ and no woman-heavy edge $(m,w^{\overline{x}}_j)$, both $m^x_i$ and $w^{\overline{x}}_j$ have utility less than $1-\eps$ in $\mu$ and hence the pair $(m^x_i,w^{\overline{x}}_j)$ is $\eps$-blocking---a contradiction.
\end{proof}

As before, the two conditions in the statement of \cref{claim:eps-variable} affect the weight of the input edges of the VC-connectors that are attached to the variable gadget in any $\eps$-stable fractional matching. Thus, condition (1) implies that the weight assigned to each input edge of the VC-connectors that correspond to the two appearances of the positive literal $x$ in clauses must be at most $\eps$. To see why, observe that these input edges are incident to nodes $m^x_1$ and $m^x_2$, and the total weight of all edges incident to each of these nodes cannot exceed $1$. Similarly, condition (2) implies that the weight assigned to each input edge of the VC-connectors that correspond to the two appearances of the negative literal $\overline{x}$ in clauses must be at most $\eps$.

\begin{claim}\label{claim:eps-vc-connector}
	Any $\eps$-stable fractional matching that assigns a weight of at most $\eps$ to the input edge of a VC-connector must assign a weight of at most $\beta \eps$ to its output edge as well.
\end{claim}

\begin{proof}
	The claim holds trivially for VC-connectors corresponding to positive literals.  Consider a VC-connector corresponding to a negative literal $\overline{x}$ and a clause $c$ containing it, and let $\zeta\leq \eps$ be the weight assigned to the input edge $(m^{\overline{x},c},w^{\overline{x}}_i)$. Observe that, besides the input edge $(m^{\overline{x},c},w^{\overline{x}}_i)$, the edge $(m^{\overline{x},c},w^{\overline{x},c})$ is the only balanced or man-heavy edge that is incident to man $m^{\overline{x},c}$ and the only balanced (or woman-heavy) edge incident to woman $w^{\overline{x},c}$. Hence, $\eps$-stability of edge $(m^{\overline{x},c},w^{\overline{x},c})$ requires a weight of at least $1-\eps-\beta \zeta$ assigned to it. Then, the output edge, which is also incident to node $m^{\overline{x},c}$, must have a weight of at most $\eps+(\beta-1)\zeta \leq \beta \eps$.
\end{proof}

\begin{claim}\label{claim:eps-clause}
	Any $\eps$-stable fractional matching that assigns a weight of at most $\beta \eps$ to each output edge of the VC-connectors corresponding to clause $c$ must assign a weight of at most $3(\beta^2+1)\eps$ to the CA-connector of clause $c$ as well.	
\end{claim}

\begin{proof}
	Let $\ell_1$, $\ell_2$, and $\ell_3$ be the literals of clause $c$. Consider, for the sake of contradiction, an $\eps$-stable fractional matching $\mu$ that assigns (1) a weight of at most $\beta \eps$ to each output edge of the VC-connectors corresponding to literals $\ell_i$ and clause $c$, and (2) a weight of more than $3(\beta^2+1)\eps$ to edge $(m^c,w^c)$ of the CA-connector for clause $c$. Note that condition (2) implies that the total weight on the edges $(m^c,w^c_1)$, $(m^c,w^c_2)$ and $(m^c,w^c_3)$ is strictly smaller than $1-3(\beta^2+1)\eps$. Since these are the only balanced or man-heavy edges incident to man $m^c$, the $\eps$-stability of these edges is guaranteed by a utility of (at least) $1-\eps$ for each of the agents $w^c_1$, $w^c_2$, and $w^c_3$.

Note that besides the output edges of the VC-connectors, the edges $(e^c_1,w^c_i)$, $(e^c_2,w^c_i)$, and $(m^c,w^c_i)$ are the only balanced or woman-heavy edges incident to agent $w^c_i$ for all $i\in\{1,2,3\}$. Along with condition (1), this implies that the weight assigned to these three edges is at least $1-(\beta^2+1)\eps$. Hence, the total weight on the nine edges of the clause gadget is at least $3-3(\beta^2+1)\eps$, which, by the definition of $\beta$, is strictly more than $2$ for the six edges incident to men $e^c_1$ and $e^c_2$, violating the definition of a fractional matching.
\end{proof}

\begin{claim}\label{claim:eps-accum}
		Any $\eps$-stable fractional matching that assigns a weight of at most $3(\beta^2+1)\eps$ to at least two CA-connectors must assign a total weight of $1-\eps$ to the tine edges and a weight of at most $\eps$ to the edge $(m_1,w_1)$ of the accumulator.
\end{claim}

\begin{proof}
	Assume that a weight of at most $3(\beta^2+1)\eps$ has been assigned to the edges $(m^{c_1},w^{c_1})$ and $(m^{c_2},w^{c_2})$ of the CA-connectors corresponding to some clauses $c_1$ and $c_2$. Since these are the only edges for which agents $w^{c_1}$ and $w^{c_2}$ have positive value, and there is no man-heavy edge incident to agent $m_1$, stability on the edges $(m_1,w^{c_1})$ and $(m_1, w^{c_2})$ requires that the total weight of the tine edges $(m_1,w^c)$ (for every clause $c$) is (at least) $1-\eps$. Indeed, since by the definition of $\beta$ we have $(3\beta^2+4)\eps <1/2$, it is not possible to guarantee stability of the edges $(m_1,w^{c_1})$ and $(m_1, w^{c_2})$ by assigning weight at least $1-(3\beta^2+4)\eps$ to each of them so that both $w^{c_1}$ and $w^{c_2}$ have utility at least $1-\eps$. Hence, the weight of the edge $(m_1,w_1)$ of the accumulator is at most $\eps$.
\end{proof}

\subsection{Proof of inapproximability}
\begin{lemma}\label{lem:eps-unsat}
	If formula $\phi$ is not satisfiable, then any $\eps$-stable fractional matching of $\I$ has welfare at most $56\beta N+\gamma\eps$.
\end{lemma}

\begin{proof}
	We will first show that if $\phi$ is not satisfiable, then any $\eps$-stable fractional matching of $\I$ assigns weight at most $3(\beta^2+1)\eps$ to at least two CA-connectors. For the sake of contradiction, consider an $\eps$-stable fractional matching that assigns weight at most $3(\beta^2+1)\eps$ to at most one CA-connector; let $c$ be the relevant clause, if such a clause exists. We will construct a truth assignment for formula $\phi$ (contradicting the assumption of the lemma) by repeating the following process for every clause $c'\neq c$ of $\phi$. Let $\ell$ be a literal that appears in $c'$ such that the output edge of the VC-connector, that corresponds to the appearance of $\ell$ in $c'$, has weight greater than $\beta \eps$. Recall that such a literal exists by \cref{claim:eps-clause}. We set $\ell$ to $1$ (true). For every variable that has not received a value in this way, we arbitrarily set it to $1$.
	
	The above assignment satisfies all clauses except possibly for clause $c$. Since clause $c$ is coupled with another clause $\hat{c}$ that is satisfied, it suffices to assign one of the variables appearing in $c$ the same value as its corresponding coupled variable appearing in $\hat{c}$. To show that the assignment is also valid, we need to argue that there is no variable $x$ such that both literals $x$ and $\overline{x}$ have been set to $1$. Assume otherwise that this is the case. Furthermore, assume that literal $x$ was set to $1$ due to its appearance in a clause $c_1$, and that this is its $i_1$-th appearance (with $i_1\in\{1,2\}$). Also, literal $\overline{x}$ was set to $1$ due to its appearance in a different clause $c_2$, where $\overline{x}$ makes its $i_2$-th appearance (again, $i_2\in\{1,2\}$). Hence, the output edge of the VC-connector that corresponds to literal $x$ and clause $c_1$ (respectively, the VC-connector that corresponds to literal $\overline{x}$ and clause $c_2$) has weight greater than $\beta \eps$. Then, by \cref{claim:eps-vc-connector}, the input edges of both VC-connectors have weight greater than $\eps$. As these input edges are incident to nodes $m^x_{i_1}$ and $w^{\overline{x}}_{i_2}$, \cref{claim:eps-variable} yields that the total weight in the edges incident to some of the nodes $m^x_{i_1}$ and $w^{\overline{x}}_{i_2}$ is strictly higher than $1$, contradicting the definition of a fractional matching.
	
	Since any $\eps$-stable fractional matching assigns a weight of at most $3(\beta^2+1)\eps$ to at least two CA-connectors, by \cref{claim:eps-accum}, the contribution of the accumulator to the welfare is at most $2(1-\eps)+\gamma \eps$ ($2(1-\eps)$ from the tine edges plus $\gamma \eps$ from the accumulator). The upper bound follows by considering the sum of valuations of all agents for edges that do not belong to the accumulator. This sum consists of
		\begin{itemize}
			\item total value of $20$ for the ten balanced edges of each of the $N$ variable gadgets,
			\item total value of $\beta$ for the edges of each of the $2N$ VC-connectors corresponding to a positive literal,
			\item total value of $2+2\beta$ for the edges of each of the $2N$ VC-connectors corresponding to a negative literal,
			\item total value of $19$ for the nine balanced edges of each of the $4N/3$ clause gadgets and their corresponding CA-connectors.
		\end{itemize}	
		It can be easily seen that $56\beta N-2(1-\eps)$ is a (loose) upper bound on the total value from these edges.
\end{proof}

\begin{lemma}\label{lem:eps-sat}
	If $\phi$ is satisfiable, then there exists an $\eps$-stable fractional matching on $\I$ that has welfare at least $\gamma$.
\end{lemma}

\begin{proof}
Consider an assignment of boolean values to the variables that satisfies $\phi$. We construct an $\eps$-stable fractional matching $\mu$ in $\I$ so that the contribution of the accumulator gadget to the welfare is at least $\gamma$.
	
\paragraph{Variable gadgets.} The weights on the edges of the variable gadget corresponding to variable $x$ are:
	 \begin{itemize}
	 	\item $\mu(m^x_1,w^{\overline{x}}_1)=\mu(e^x_3,w^{\overline{x}}_1)=\mu(e^x_3,w^{\overline{x}}_2)=\mu(m^x_2,w^{\overline{x}}_2)=1/2$, $\mu(e^x_1,f^x_1)=\mu(e^x_2,f^x_2)=1$, and $\mu(m^x_1,f^x_1)=\mu(m^x_1,w^{\overline{x}}_2)=\mu(m^x_2,f^x_2)=\mu(m^x_2,w^{\overline{x}}_1)=0$ if $x$ is true, and
	 	\item $\mu(e^x_3,w^{\overline{x}}_1)=\mu(e^x_3,w^{\overline{x}}_2)=1/2$, $\mu(m^x_1,f^x_1)=\mu(m^x_2,f^x_2)=1$, $\mu(m^x_1,w^{\overline{x}}_1)=\mu(m^x_2,w^{\overline{x}}_2)=\mu(m^x_1,w^{\overline{x}}_2)=\mu(m^x_2,w^{\overline{x}}_1)=\mu(e^x_1,f^x_1)=\mu(e^x_2,f^x_2)=0$ if $x$ is false.
	 \end{itemize}
	
\paragraph{Clause gadgets and CA-connectors.} For every clause we select arbitrarily one of the true literals of the clause and call it {\em active}; since the assignment satisfies $\phi$, there is certainly such a literal. Consider clause $c$ and let $\ell_i$ (with $i\in \{1,2,3\}$) be its active literal; let $i_1$ and $i_2$ be the indices from $\{1,2,3\}$ than are different than $i$. We set $\mu(e^c_1,w^c_{i_1})=\mu(e^c_2,w^c_{i_2})=1$ and $\mu(e^c_1,w^c_i)=\mu(e^c_1,w^c_{i_2})=\mu(e^c_2,w^c_i)=\mu(e^c_2,w^c_{i_1})=\mu(m^c,w^c_i)=\mu(m^c,w^c_{i_1})=\mu(m^c,w^c_{i_2})=0$. We also assign a weight of $1$ to the CA-connector corresponding to $c$, i.e., $\mu(m^c,w^c)=1$.
	
\paragraph{VC-connectors.} For every non-active VC-connector, we set the weight of its balanced edge, if it exists, to $1$ and the weight of the remaining edges to $0$. For every active VC-connector corresponding to the $i$-th appearance of the positive literal $x$ as the $j$-th literal of clause $c$ ($i\in\{1,2\}$, $j\in \{1,2,3\}$), we set $\mu(m^x_i,w^c_j)=1/2$.

	For every active VC-connector corresponding to the $i$-th appearance of the negative literal $x$ as the $j$-th literal of clause $c$ ($i\in\{1,2\}$, $j\in \{1,2,3\}$), we set the weights of its edges as follows: $\mu(m^{\overline{x},c},w^{\overline{x}}_i)=\mu(m^{\overline{x},c},w^c_j)=1/2$ and $\mu(m^{\overline{x},c},w^{\overline{x},c})=0$.

\paragraph{Accumulator.} We set $\mu(m_1,w^c)=0$ for every tine edge $(m_1,w^c)$ of the accumulator. Furthermore, we set $\mu(m_1,w_1)=1$. So, the contribution of the accumulator to the social welfare is $\gamma$, as desired.
	
It can be easily verified that the total weight of the edges that are incident to any node is at most $1$. Hence, $\mu$ is a valid fractional matching. Regarding stability, it suffices to verify that either the man or the woman of a balanced pair has a utility of at least $1-\eps$. Note that we only need to check for the balanced edges, as the man- or woman-heavy edges and the remaining pairs do not impose any constraints on stability. For the balanced edges, stability is established by the following series of observations (we will use the term `stabilized by' to denote that an agent's utility is at least $1-\eps$): The variable gadget for the variable $x$ (\cref{fig:eps:variable-and-clause-gadget}a) is stabilized by the agents $f_1^x$, $f_2^x$, $e_3^x$ along with $m_1^x$, $m_2^x$ (if $x$ is true) or $w_1^{\overline{x}}$, $w_2^{\overline{x}}$ (if $x$ is false). The clause gadget for clause $c$ (\cref{fig:eps:variable-and-clause-gadget}b) with active index $i$ (and non-active indices $i_1$ and $i_2$) is stabilized by the agents $e^c_1$, $e^c_2$, $w_i^c$, $w_{i_1}^c$, $w_{i_2}^c$; in particular, the edge $(m^c,w_i^c)$ is stabilized by $w_i^c$ because an active literal triggers the woman-heavy edge in the VC-connector. A VC connector is stabilized by $m^{\overline{x},c}$ (\cref{fig:eps:vc-connector}b). Finally, the tine edges in the accumulator (\Cref{fig:eps:acc}) are stabilized by $w^{c_1},\dots,w^{c_L}$ (because we trigger the CA-connector). Overall, $\mu$ is a feasible stable fractional matching.
\end{proof}

We are ready to prove \cref{thm:strong-eps-inapprox}. We select a value of $\gamma$ such that $\gamma \geq \frac{56\beta N(1/\eps-\delta)}{\eps \delta}$. By \cref{lem:eps-unsat}, the welfare of $\mu$ if $\phi$ was not satisfiable would be at most
$$56\beta N+\gamma \eps \leq \frac{\gamma \eps \delta}{1/\eps-\delta} + \gamma\eps =\frac{\gamma}{1/\eps - \delta}.$$
By \cref{lem:eps-sat}, we have that the welfare of $\mu$ if $\phi$ was not satisfiable would be at least $1/\eps -\delta$ times smaller than the welfare $\I$ could have if $\phi$ was satisfiable.

\section[Approximation algorithm for approximate stability]{Approximation Algorithm for $\frac{1}{2}$-stability}\label{app:Half-stability}

We will now provide a polynomial-time algorithm that computes a $\frac{1}{2}$-stable fractional matching with welfare at least that of an optimal (exactly) stable fractional matching. Notice that unlike \cref{thm:alg-summary} and \cref{thm:ApproxStable+ApproxOptimal}, where the quality of the computed matching is compared to the optimal matching $\mu^\text{opt}$, the guarantee in \cref{thm:OptimalHalfStablePolytime} is considerably weaker.

\begin{theorem}
 \label{thm:OptimalHalfStablePolytime}
 Let $\I$ be an \SMC{} instance and $\mu^*$ be an optimal stable fractional matching for $\I$. Then, a $\frac{1}{2}$-stable fractional matching $\mu$ that satisfies $\W(\mu) \geq \W(\mu^*)$ can be computed in polynomial time.
\end{theorem}
\begin{proof}
Consider the mixed integer linear program \ref{Prog:OPT-Stab} from \Cref{subsec:OptimalStableMatchingProgram} for finding an optimal stable fractional matching for $\I$. Relaxing the integrality constraint (\ref{MILP-OPT-Stab:Integral_Constraint}) to $y(m,w) \in [0,1]$ results in a linear program. Since a stable fractional matching always exists (see Proposition~\ref{prop:GaleShapley}), this relaxation is feasible. Let $\mu$ be a solution of the relaxed program. Since $\max\{ y(m,w), 1 - y(m,w) \} \geq \frac{1}{2}$, we have that for every man-woman pair $(m,w) \in M \times W$, either $\textstyle{ u_m \geq \frac{1}{2} U(m,w) }$ or $\textstyle{ v_w \geq \frac{1}{2} V(m,w) }$, implying that $\mu$ is $\textstyle{ \frac{1}{2} }$-stable. It is also clear that $\W(\mu) \geq \W(\mu^*)$ since $\mu^*$ satisfies \ref{Prog:OPT-Stab}.
\end{proof} 

\end{document}